\PassOptionsToPackage{
		pagebackref
	}{hyperref}

\documentclass[format=acmsmall, review=false, screen, nonacm]{acmart}
\usepackage{booktabs} %
\usepackage[ruled]{algorithm2e} %

\SetAlFnt{\small}
\SetAlCapFnt{\small}
\SetAlCapNameFnt{\small}
\SetAlCapHSkip{0pt}
\IncMargin{-\parindent}

\usepackage{amsmath}
\usepackage{amsthm}
\usepackage{microtype}
\usepackage{thmtools}
\usepackage{thm-restate}

\AtEndPreamble{%
	\theoremstyle{acmdefinition}
	\newtheorem{remark}[theorem]{Remark}
	
	}

\usepackage[nameinlink]{cleveref}
\usepackage{crossreftools}
\usepackage{subcaption}
\usepackage{tikz}
\usetikzlibrary{positioning}
\usepackage{xcolor}
\usepackage{nicefrac}
\usepackage{paralist}
\usepackage[colorinlistoftodos,textsize=footnotesize]{todonotes}
\tikzset{/tikz/notestyleraw/.append style={rounded corners=0pt,inner sep=0.6ex}}

\usepackage{graphicx} %
\usepackage{multirow} %

\usepackage{wrapfig}
\usepackage{etoolbox}
\makeatletter
\patchcmd\WF@putfigmaybe{\lower\intextsep}{}{}{\fail}%
\AddToHook{env/wrapfigure/begin}{\setlength{\intextsep}{0pt}}
\makeatother

\usepackage{xcolor}
\usepackage{ifthen}
\definecolor{alt-a}{HTML}{ffadad}
\definecolor{alt-b}{HTML}{ffd6a5}
\definecolor{alt-c}{HTML}{cdd1fa}
\definecolor{alt-c'}{HTML}{cdd1fa}
\definecolor{alt-x}{HTML}{fad4d7}
\definecolor{alt-d}{HTML}{fad4d7}
\definecolor{alt-e}{HTML}{b9d2fa}
\definecolor{alt-f}{HTML}{cbc2ff}
\definecolor{alt-g}{HTML}{ffc6ff}
\newcommand{\alternative}[2][default]{%
	\tikz[transform shape,scale=0.85]{
		\def\colorparam{#1}
		\ifthenelse{\equal{#1}{default}}{
			\def\colorparam{alt-#2}
		}{}
		\node[anchor=base, yshift=0.5pt, text height=6pt,inner sep=1pt,fill=\colorparam, circle,minimum width=12pt] {};
		\node[anchor=base, text height=6pt, inner sep=1pt,minimum width=12pt,scale=1.15] {$#2$};
	}
}
\newcommand{\borderalternative}[2][default]{%
	\tikz[transform shape,scale=0.85]{
		\def\colorparam{#1}
		\ifthenelse{\equal{#1}{default}}{
			\def\colorparam{alt-#2}
		}{}
		\node[anchor=base, yshift=0.5pt, text height=6pt,inner sep=1pt,fill=\colorparam, circle,minimum width=12pt,draw=black!70,very thin] {};
		\node[anchor=base, text height=6pt, inner sep=1pt,minimum width=12pt,scale=1.15] {$#2$};
	}
}
\newcommand{\weakorder}[2][noborder]{%
	\begin{tikzpicture}
		[yscale=0.5]
		\foreach \row [count=\y] in {#2} {
			\ifthenelse{\y>1 \OR \equal{#1}{noborder}}{
				\node [anchor=south, inner sep=0.5pt] at (0, -\y) {
					\tikz[xscale=0.37]{
						\foreach \val/\color [count=\x] in \row {
							\ifthenelse{\equal{\val}{\color}}
								{\def\color{default}} %
								{}
							\node [inner sep=0] at (\x, 0) {\alternative[\color]{\val}};
						};}
				};
			}{
				\node [anchor=south, inner sep=0.5pt] at (0, -\y) {
					\tikz[xscale=0.37]{
						\foreach \val/\color [count=\x] in \row {
							\ifthenelse{\equal{\val}{\color}}
								{\def\color{default}} %
								{}
							\node [inner sep=0] at (\x, 0) {\borderalternative[\color]{\val}};
						};}
				};
			}
		}
	\end{tikzpicture}
}
\newcommand{\votermultiplicity}[2]{%
	\begin{tikzpicture}
		\node[inner sep=1pt] at (0,0) {#1};
		\node[anchor=north,inner sep=0] at (0,-0.3) {#2};
	\end{tikzpicture}
}

\newcommand{\ExternalLink}{%
	\tikz[x=1.2ex, y=1.2ex, baseline=-0.05ex]{%
		\begin{scope}[x=1ex, y=1ex]
			\clip (-0.1,-0.1) 
			--++ (-0, 1.2) 
			--++ (0.6, 0) 
			--++ (0, -0.6) 
			--++ (0.6, 0) 
			--++ (0, -1);
			\path[draw, 
			line width = 0.5, 
			rounded corners=0.5] 
			(0,0) rectangle (1,1);
		\end{scope}
		\path[draw, line width = 0.5] (0.5, 0.5) 
		-- (1, 1);
		\path[draw, line width = 0.5] (0.6, 1) 
		-- (1, 1) -- (1, 0.6);
	}%
}

\newcommand{\xmark}{%
	\tikz[scale=0.23,draw=black!50!red] {
		\draw[line width=0.7,line cap=round] (0,0) to [bend left=6] (1,1);
		\draw[line width=0.7,line cap=round] (0.2,0.95) to [bend right=3] (0.8,0.05);
}}
\newcommand{\cmark}{%
	\tikz[scale=0.23,draw=green!50!black] {
		\draw[line width=0.7,line cap=round] (0.25,0) to [bend left=10] (1,1);
		\draw[line width=0.8,line cap=round] (0,0.35) to [bend right=1] (0.23,0);
}}

\renewcommand{\epsilon}{\varepsilon}
\renewcommand{\le}{\leqslant}

\renewcommand{\ge}{\geqslant}

\newcommand{\pref}{\succcurlyeq}
\renewcommand{\succeq}{\succcurlyeq}

\usepackage{rotating}
\newcommand{\rot}[1]{\begin{turn}{90}#1\enspace\end{turn}}

\usepackage{tikz}
\usepackage{pgfplots}
\pgfplotsset{compat=1.15,
legend image code/.code={
\draw[mark repeat=2,mark phase=2]
plot coordinates {
(0cm,0cm)
(0.15cm,0cm)        %
(0.3cm,0cm)         %
};%
}}
\usepgfplotslibrary{groupplots}
\setcitestyle{authoryear}

\title{Generalizing Instant Runoff Voting to Allow Indifferences}

\author{Théo Delemazure}
\affiliation{%
	\institution{CNRS, LAMSADE, Université Paris Dauphine - PSL}
	\country{France}
}
\email{theo.delemazure@lamsade.dauphine.fr}

\author{Dominik Peters}
\affiliation{%
	\institution{CNRS, LAMSADE, Université Paris Dauphine - PSL}
	\country{France}
}
\email{dominik.peters@lamsade.dauphine.fr}

\begin{abstract}
	{\large Manuscript: April 2024}
	
	\bigskip
	\noindent
	Instant Runoff Voting (IRV) is used in elections for many political offices around the world. It allows voters to specify their preferences among candidates as a ranking. We identify a generalization of the rule, called Approval-IRV, that allows voters more freedom by allowing them to give equal preference to several candidates. Such weak orders are a more expressive input format than linear orders, and they help reduce the cognitive effort of voting.
	
	Just like standard IRV, Approval-IRV proceeds in rounds by successively eliminating candidates. It interprets each vote as an approval vote for its most-preferred candidates among those that have not been eliminated. At each step, it eliminates the candidate who is approved by the fewest voters. Among the large class of scoring elimination rules, we prove that Approval-IRV is the unique way of extending IRV to weak orders that preserves its characteristic axiomatic properties, in particular independence of clones and respecting a majority's top choices. We also show that Approval-IRV is the unique extension of IRV among rules in this class that satisfies a natural monotonicity property defined for weak orders.
	
	Prior work has proposed a different generalization of IRV, which we call Split-IRV, where instead of approving, each vote is interpreted as splitting 1 point equally among its top choices (for example, 0.25 points each if a vote has 4 top choices), and then eliminating the candidate with the lowest score. Split-IRV fails independence of clones, may not respect majority wishes, and fails our monotonicity condition.
	
	The multi-winner version of IRV is known as Single Transferable Vote (STV). We prove that Approval-STV continues to satisfy the strong proportional representation properties of STV, underlining that the approval way is the right way of extending the IRV/STV idea to weak orders.
\end{abstract}

\begin{document}

\begin{titlepage}
\maketitle
\vspace{5pt}
\hrule
\vspace{10pt}
\tableofcontents
\vspace{-18pt}
\hrule
\end{titlepage}
\addtocounter{page}{1}

\section{Introduction}
\begin{wrapfigure}{r}{0.25\textwidth}
	\includegraphics[width=\linewidth]{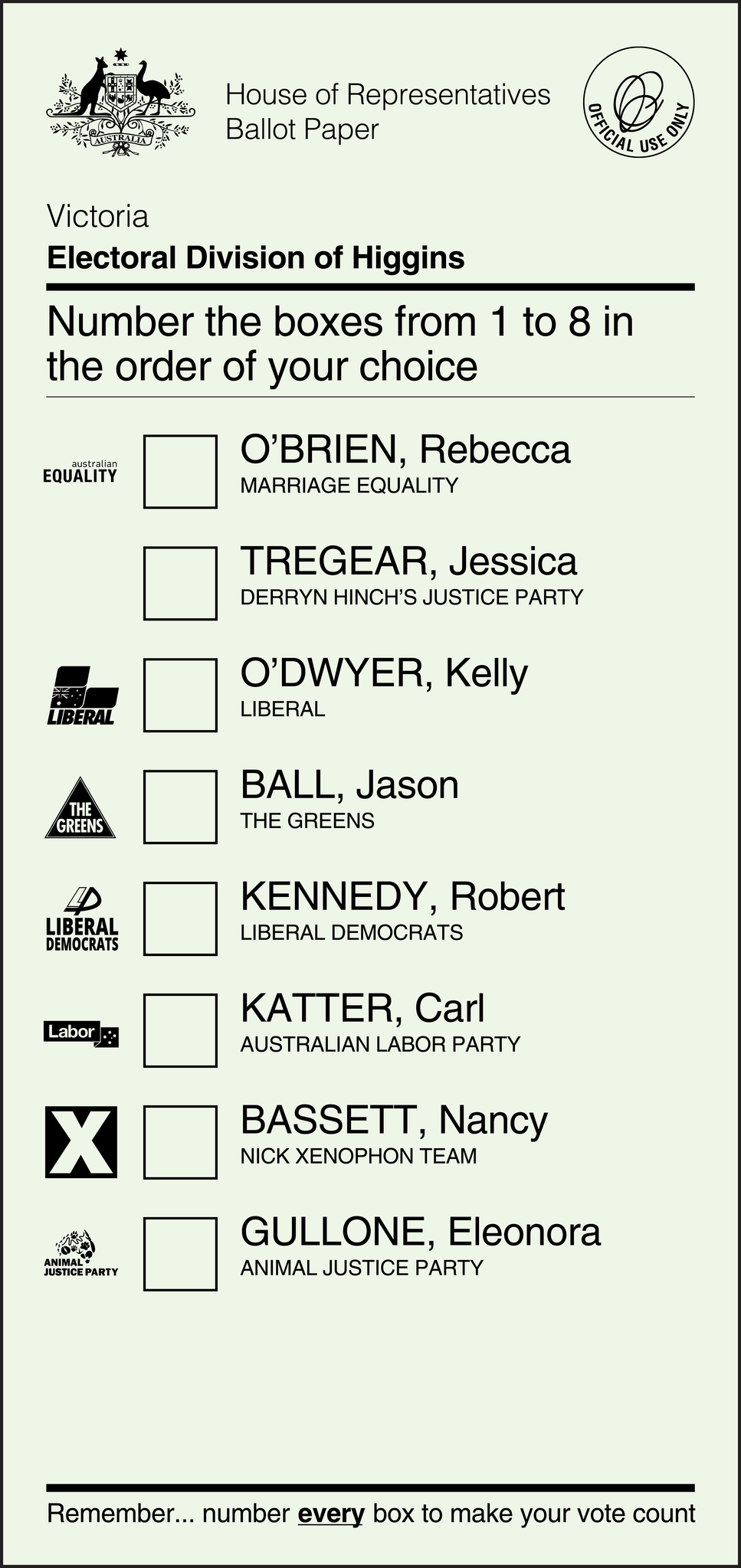}
	\vspace{-16pt}
	\caption{Ballot from 2016 Australian federal election}
	\label{fig:australia-ballot}
\end{wrapfigure}

Instant Runoff Voting (IRV) is a voting rule for selecting a single winner, based on voters ranking the available candidates or alternatives. It works in rounds by sequentially eliminating candidates: the rule repeatedly identifies the candidate $c$ for whom the fewest voters currently vote (in the sense that $c$ is their highest-ranked candidates among those that have not been eliminated), and eliminates that candidate. The last candidate remaining is the winner.

While a wide variety of ranking-based voting rules have been studied over the centuries, IRV is essentially the only ranking-based voting rule that has been adopted for elections to political offices. Australia has used the rule for electing its House of Representatives since 1918 (\Cref{fig:australia-ballot}), Ireland uses it to elect its President, and it is used in Alaska and Maine for several offices. It is also in use at the local level in many jurisdictions, as well as within some societies and political parties. Electoral reform advocates such as \href{https://fairvote.org/}{FairVote \ExternalLink} are pushing for further adoption, arguing that IRV and its ranking-based input leads to election outcomes that better reflect voters' preferences.

An important drawback of IRV is the burden it imposes on voters who may need to rank-order a large number of candidates. This is particularly severe in Australia, where voting is compulsory and a ballot is invalid if it fails to rank every candidate. As we can see in \Cref{fig:australia-candidate-number}, 20\% of Australian voters need to rank 10 or more candidates.

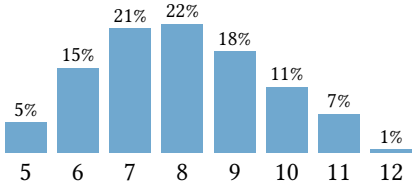
\begin{wrapfigure}{l}{0.42\textwidth}
	\centering
\begin{tikzpicture}

\definecolor{darkgray176}{RGB}{176,176,176}
\definecolor{skyblue}{RGB}{112,168,207}

\begin{axis}[
tick align=outside,
xtick=data,
xmin=4.21, xmax=12.79,
xtick style={draw=none},
xtick distance=1,
xticklabel shift=-3,
xtick pos=lower,
ymin=0, ymax=27.0,
ymajorticks=false,
axis line style={draw=none},
xscale=0.87,
yscale=0.36
]
\draw[draw=none,fill=skyblue] (axis cs:4.6,0) rectangle (axis cs:5.4,5.19739507455496);
\draw[draw=none,fill=skyblue] (axis cs:5.6,0) rectangle (axis cs:6.4,14.6863119525913);
\draw[draw=none,fill=skyblue] (axis cs:6.6,0) rectangle (axis cs:7.4,21.4995153492098);
\draw[draw=none,fill=skyblue] (axis cs:7.6,0) rectangle (axis cs:8.4,22.3127883962412);
\draw[draw=none,fill=skyblue] (axis cs:8.6,0) rectangle (axis cs:9.4,17.5263946044147);
\draw[draw=none,fill=skyblue] (axis cs:9.6,0) rectangle (axis cs:10.4,11.406868691559);
\draw[draw=none,fill=skyblue] (axis cs:10.6,0) rectangle (axis cs:11.4,6.7324205259767);
\draw[draw=none,fill=skyblue] (axis cs:11.6,0) rectangle (axis cs:12.4,0.638305405452467);
\draw (axis cs:5,6.19739507455496) node[
  scale=0.8, yshift=10,
  anchor=north,
  text=black,
  rotate=0.0
]{5\%};
\draw (axis cs:6,15.6863119525913) node[
  scale=0.8, yshift=10,
  anchor=north,
  text=black,
  rotate=0.0
]{15\%};
\draw (axis cs:7,22.4995153492098) node[
  scale=0.8, yshift=10,
  anchor=north,
  text=black,
  rotate=0.0
]{21\%};
\draw (axis cs:8,23.3127883962412) node[
  scale=0.8, yshift=10,
  anchor=north,
  text=black,
  rotate=0.0
]{22\%};
\draw (axis cs:9,18.5263946044147) node[
  scale=0.8, yshift=10,
  anchor=north,
  text=black,
  rotate=0.0
]{18\%};
\draw (axis cs:10,12.406868691559) node[
  scale=0.8, yshift=10,
  anchor=north,
  text=black,
  rotate=0.0
]{11\%};
\draw (axis cs:11,7.7324205259767) node[
  scale=0.8, yshift=10,
  anchor=north,
  text=black,
  rotate=0.0
]{7\%};
\draw (axis cs:12,1.63830540545247) node[
  scale=0.8, yshift=10,
  anchor=north,
  text=black,
  rotate=0.0
]{1\%};
\end{axis}

\end{tikzpicture}
 	\vspace{-20pt}
	\caption{The fraction of voters in the 2022 Australian federal election that needed to rank each number of candidates between 5 and 12.}
	\label{fig:australia-candidate-number}
\end{wrapfigure}

Most other jurisdictions allow voters to submit a truncated ranking, where the voter may stop after ranking only some of the candidates, and the vote is not taken into account (or ``exhausted'') after all the candidates that were ranked have been eliminated. However, voters who wish to rank some disfavored candidates in low positions must rank the whole field. It also forces voters to clearly distinguish all candidates that they favor, even when they may not have sufficient information to do so. In particular, voters cannot just submit a simple ``approval vote'' where they indicate several candidates as acceptable.

A possible solution to these issues is to allow voters to express indifferences, that is, to assign several candidates  an equal rank. 
Because IRV is only defined for \emph{linear orders} (rankings without indifferences), to implement this solution we need to decide how to generalize IRV to \emph{weak orders} (rankings with indifferences), and the right way to generalize it is not obvious.

In this paper, we will argue that the right generalization is what we call \emph{Approval-IRV}, which combines the ideas of IRV and of approval voting. This rule interprets each weak order as an approval vote for the highest-ranked candidates that have not yet been eliminated. It then repeatedly eliminates the candidate with the lowest approval score (i.e., the candidate who is top-ranked by the fewest voters), until only one candidate remains and is declared the winner. \Cref{fig:approval-irv-example} shows an example of how the rule works.

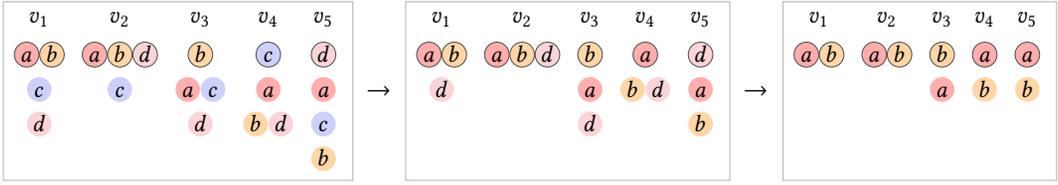
\begin{figure}[t]
	\makebox[\textwidth][c]{\scalebox{0.91}{
		\begin{tikzpicture}
		\node at (0,0) [draw=black!30] (step1) {	
			\votermultiplicity{$v_1$}{\weakorder[border]{{a,b},{c},{d},{}}}
			\votermultiplicity{$v_2$}{\weakorder[border]{{a,b,d},{c},{},{}}}
			\votermultiplicity{$v_3$}{\weakorder[border]{{b},{a,c},{d},{}}}
			\votermultiplicity{$v_4$}{\weakorder[border]{{c},{a},{b, d},{}}}
			\votermultiplicity{$v_5$}{\weakorder[border]{{d},{a},{c},{b}}}
		};
		\node (elim1) [right=0.07cm of step1, minimum height=2cm] {$\rightarrow$};
		\node at (5.7,0) [draw=black!30] (step2) {%
			\votermultiplicity{$v_1$}{\weakorder[border]{{a,b},{d},{},{}}}
			\votermultiplicity{$v_2$}{\weakorder[border]{{a,b,d},{},{},{}}}
			\votermultiplicity{$v_3$}{\weakorder[border]{{b},{a},{d},{}}}
			\votermultiplicity{$v_4$}{\weakorder[border]{{a},{b, d},{},{}}}
			\votermultiplicity{$v_5$}{\weakorder[border]{{d},{a},{b},{}}}
		};
		\node (elim2) [right=0.07cm of step2, minimum height=2cm] {$\rightarrow$};
		\node at (10.85,0) [draw=black!30] (step3) {%
			\votermultiplicity{$v_1$}{\weakorder[border]{{a,b},{},{},{}}}
			\votermultiplicity{$v_2$}{\weakorder[border]{{a,b},{},{},{}}}
			\votermultiplicity{$v_3$}{\weakorder[border]{{b},{a},{},{}}}
			\votermultiplicity{$v_4$}{\weakorder[border]{{a},{b},{},{}}}
			\votermultiplicity{$v_5$}{\weakorder[border]{{a},{b},{},{}}}
		};
		\end{tikzpicture}
	}}
	\caption{An example of Approval-IRV with voters $v_1, \dots, v_5$. The first eliminated alternative is $c$, which is ranked on top only once. Then $d$ is eliminated, and finally $a$ wins the majority vote against $b$. Thus, $a$ is the winner.}
	\label{fig:approval-irv-example}
\end{figure}

There are several potential arguments for moving from standard IRV (allowing only rankings or truncated rankings) to a system like Approval-IRV that allows for weak orders.

(1) \emph{Less effort}. Allowing voters to give equal preference to several candidates reduces the problems with rankings that we discussed above. In particular, voters can give a high ranking to some candidates and a low ranking to others, without having to rank all candidates in between (see \Cref{fig:sfballots:veto}), and it is possible to submit a simple approval vote  (see \Cref{fig:sfballots:approval}). Moreover, voters who do not have enough information to strictly rank some candidates can still vote sincerely.

(2) \emph{More expressive}. Allowing equal rankings gives voters more ways to express their preferences. For example, some people might be truly indifferent between candidates, and forcing them to rank such candidates would be distorting. (For this reason, some have argued that Australia's compulsory voting forces voters to lie \citep{rydon1968compulsory,orr1997choice}.)
The additional expressive power is also useful in U.S. jurisdictions that allow voters to only use 3--5 ranks. Weak orders allow voting for more than 3--5 candidates within the same number of ranks.

(3) \emph{Fewer invalid ballots.} Compared to plurality elections (``choose your favorite candidate''), voting in an IRV election is more difficult, and the voting instructions are more complicated to follow. In particular, some voters may not realize that they must only place one candidate in each rank, and submit a ballot which encodes a weak order and which will be counted as invalid.
In the American context, this mistake is known as an ``overvote''.%
\footnote{In many U.S. jurisdictions using IRV, such ballots are partially counted: if the top ranks contain unique choices, the ballot is counted as a vote for them. Once those candidates are eliminated and a rank with several candidates is reached, the vote is counted as exhausted and ignored in the following rounds.}
Allowing weak orders could reduce the number of inadvertent invalid ballots.

\begin{figure}[ht]
	\centering
	\begin{subfigure}{0.37\textwidth}
		\centering
		\includegraphics[height=2.93cm]{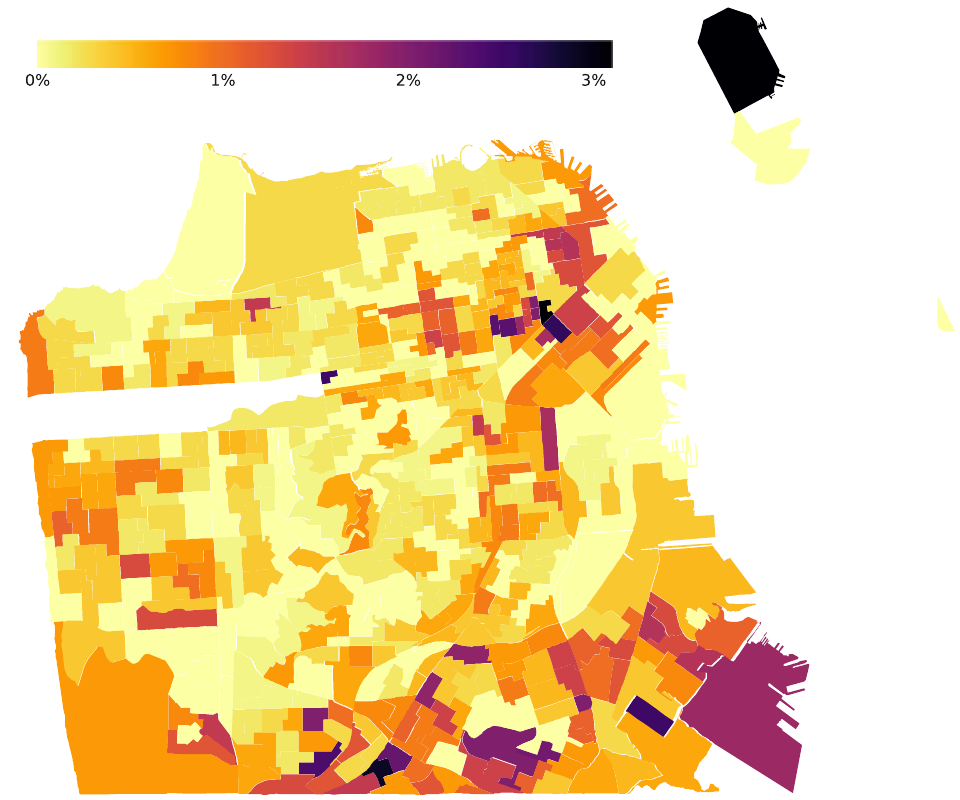}
		\caption{Map of San Francisco election precincts, colored by the fraction of votes that could be interpreted as a weak order with indifferences.}
		\label{fig:sfmap}
	\end{subfigure}
	\qquad
	\begin{subfigure}{0.43\textwidth}
		\centering
		\scalebox{0.88}{
\begin{tikzpicture}

\definecolor{darkgray176}{RGB}{176,176,176}
\definecolor{steelblue31119180}{RGB}{31,119,180}

\begin{axis}[
tick align=outside,
tick pos=left,
x grid style={darkgray176},
xmajorgrids,
xmin=0, xmax=260000,
xtick style={color=black},
xtick={0,100000,200000,300000},
xticklabels={\$0k,\$100k,\$200k,\$300k},
scaled x ticks=false,
y grid style={darkgray176},
ymajorgrids,
ymin=0, ymax=3.25,
ytick style={color=black},
ytick={0,1,2,3},
yticklabels={0\%,1\%,2\%,3\%},
height=4.6cm, 
width=7.7cm 
]
\addplot [draw=steelblue31119180, fill=steelblue31119180, mark=*, only marks, opacity=0.5, mark size=1.3pt]
table{%
x  y
72927.362 1.1
112542.316 1.3
74992.007 1.5
90474.74 0.4
96346 0.7
109351.999 1.2
114821 0.6
72936 0.7
69916.561 1.1
91444.026 1.2
98763.593 2.9
18660.417 1.4
97284.474 1
63626.094 0.7
108465.697 0.2
96867.651 0.3
75221 0.6
123466 0.5
94631 0.6
90665.519 2.2
72387.585 0
85621.28 2.4
119061.208 0.7
56518.981 2.1
98355.036 0.3
87777.84 0.6
108194.622 0.2
97898.008 0.7
57493.562 1.2
38281.181 1.4
81432.056 0.9
98578.813 0.9
28536.853 2
114807.64 1.1
54934.521 0
108307.439 1.2
112687.894 0.3
87768.943 0
11205 0
116250 0
70244.096 1.5
111223.1 0.4
79769.082 0.2
87672.967 0.2
86188.509 1.1
143845.864 0.3
91878.294 0.6
88668.051 2
81485.987 0.6
100662.248 0.9
78300.352 1.1
95979.2 1.4
107208.877 0.8
108978.031 1
113622.228 0.4
107578.092 0.6
36841.635 0.3
73980.516 0.8
74591.574 0.8
76745.28 2
89011.696 0.1
11205 0
109792.096 0.5
118384.503 0.9
2576.462 0
51574.767 0
68121.179 1.9
56983.115 0.4
102426.234 0.5
128807.2 0.3
37963.182 1.4
122117.245 0.2
117973 0.3
98141.369 0.7
66978.581 1.2
98757.51 0.4
141722.548 0
127778.064 0
182087.629 0.7
86250 0.8
173401.351 0
103818.652 0.5
68676.604 0.5
124433.64 0.3
115814.348 0.6
124246.013 0.8
144036.987 0.4
112293.608 0
40820.085 1.6
72122.983 1.9
29589.928 0
47719.74 0
189106.941 0.2
56184.124 0
97683.864 0.7
93811.363 1.4
159270.988 0
114801.45 0.7
98000.392 0.1
83392.953 1.3
20740.5 1.3
158417.96 0
95602.694 1.1
179248.921 0.2
179294.139 0
68601.931 0
98583.52 0.5
169681.315 0.6
120921.06 0.8
199739.947 0.4
56580.957 0.7
67425.545 0.4
89875.252 1.8
196720.1 0.4
113170.23 0.3
202133.977 0.1
126222.6 0.3
149783.498 0.6
37906 0
75805.597 1
84634.072 0.3
183775.119 0
173816 0.6
139002.1 0.4
27187.075 1.1
131694.365 0
322.656 0
55766.997 1.6
128212.906 0
236477.035 0
188770.76 0.1
161702.712 0.4
89975.857 1.3
138770.481 0.6
165788.933 0.5
227660.852 0.6
181392.02 0.4
214667.156 0
136012.982 0
13800.605 1.3
175781.119 0.9
168521.152 0
103267.561 0.4
232340.399 0
196234.388 0.2
65095.517 0.4
118825.456 0
124106.837 0.4
154125 0.2
160750 0
160750 0.4
145599.35 0
165987 0
134827.918 0.4
38746.161 0.4
109576.698 1
195248.744 0.6
100255.322 0.4
212370.478 0.3
165905.432 0
191313.47 0
93904.064 0.5
112133.481 1.3
113528.476 0.7
236384.133 0
95754.153 0
111943.93 0.8
150618.694 0.2
86096.782 0.1
141064.07 0.2
184772.381 0.5
204257.25 0.5
105618.338 0
63309.6 0.4
142097.103 0.8
115611.778 0.3
138743.918 0.2
196482.447 0
104155.44 0.7
152751.185 0.2
168140.976 0.5
98423.621 0.1
166626.189 0.4
31738.112 0
36472.37 0.2
150612.26 0.2
125401.54 0.3
140157.376 0
183749.721 0
124369.016 0.9
121310.995 0.2
124185.215 0.2
27294.836 0
123373.41 0.5
186878.997 0.3
87196.189 0.6
168080.862 0.2
116815.575 0.6
104.358 0.3
104357.348 1.1
165184.096 0
7604.757 0.9
168422.96 0
183439.456 0.2
226787.892 0.3
111544.641 0.3
119909.756 0.9
95283.818 0
98306.851 0.1
140771.3 0.3
115789.536 0
104989.181 0.9
139613.296 0.5
170247.523 0.8
199864.285 0.3
168662.292 0.2
115843.6 0.3
131102.863 0.4
199871.26 0.3
172509.072 1
133636.292 0
107265.105 0.8
170704.174 0
104953.125 0.7
127179.296 0.2
138017 0.2
146859.572 0.3
147134.301 0.2
155379.61 1
206430.734 0.2
141108.167 1
192811.3 0.1
206096.464 0
155688.345 0.1
66061.37 0.9
195922.765 0
168281 0.6
153137.36 0
164666.369 0.2
146225.285 0.1
131771.358 0.9
135836.955 0.6
135424.719 0.8
117787.439 0.2
41298.34 0
109349.051 0.4
127326.312 0.4
191250 0.1
152738 0
165391.017 0
130052.375 0.6
211890.011 0.4
148094.17 0.9
223734.074 0.1
137959.78 0.9
137350.983 0
156289 0
121795 0.3
181732.35 0.1
135682 0
146232.94 0.6
111570.8 0.3
169938.182 0.3
125426.781 0.4
110315.862 0.6
210794.514 0.1
152426.834 0.5
153255 0
151326.035 0.1
127940.974 0
240235.566 0.1
185188.47 0.6
160271.608 0.1
178874.339 0.3
83709.553 0.7
119993.13 0
121001.704 0.3
112139.338 0.3
90634.338 0.4
233352.852 0
151595.06 0.2
70458.081 0.3
223225.863 0.1
185916.175 0
131732.698 0
126428.999 0.7
116831.473 0.3
145531.893 0.2
118400.733 0.3
97209.443 0
142267.491 0.5
97040.112 0.5
102801.088 0.5
114560.756 0.4
112834.647 0.4
126343.397 0
157408.6 0.1
127966.474 0
125524 0.2
175600.676 0.1
134341 0
99864.2 0.4
147217.904 0.7
78353.225 1.2
160681.33 0.2
192036.723 0
209147.373 0
133047.969 0
124651.128 0.4
104850 1.5
150132 0
208681.51 0
220626.608 0
198339.358 0
177394.89 0.5
164269.162 1.7
93426.104 0
86704.914 0.2
129949.285 0
85756.751 0.2
123909.457 0.3
224411.724 0
203392.951 0.6
176143 0.2
142572.496 0.8
116010.864 0.3
187225.242 0.4
159348.58 0.5
160430.85 0
195997.122 0.5
139714 0.4
318.394 0.7
120546.8 0.4
179773 0.6
144458.237 0.3
113032.5 0.4
182671.024 0
101250.798 0.5
171905.636 0.6
136281 0.7
141011.295 0.3
189719 2.5
103908 0.3
175118.698 0
167777.32 0
157940.82 0.2
99248.214 0.3
101459.25 0.8
145612.585 0.6
154291.104 0
124068.875 0.2
108102.913 0.9
114734.308 0.6
125092.226 0.2
172541.5 0.3
116003.303 0
107716.306 0.8
221663.432 0
176391.624 0.1
41912.312 0.7
160579.503 0
120921 0.1
174736.032 0.1
126352.332 0.3
160588.833 0.2
151974.45 0.1
150319.243 0.2
100099.879 0.8
118339 0.2
181331 0.3
121449.25 0.6
120014.488 0.3
55570.932 0.4
178162.353 0
119592.194 0.3
4911.88 0.4
155015.783 0
124214.786 0
139602.58 0.8
100723.046 0
170488.159 0.3
155547.635 0.3
118821.698 0.1
22083 1
103075.994 0
112388.716 0
136185.118 0.9
132115.569 0.3
148303.796 0
99264.85 0.6
69090.01 0
91119.17 0.3
89501.75 0.3
138437.06 0.3
113797.505 0.3
148520.992 0
127289.367 0.2
125752.037 0.3
107289.6 0.3
116230.032 0.3
34287.25 1.2
122176.6 0.7
116681.727 0.2
67993.004 0.9
162319.768 0.9
21371.628 1.4
108383.775 0.4
50579.761 1.7
112703.056 0.6
38051.059 0.9
102825.397 0.3
114970.802 0.5
95278.889 0
127852.12 0.2
163471.452 0.6
1843.656 0.8
62578.762 2.3
1905.776 0
28252.13 1.8
87786.366 0
94764.966 0.3
26132 1.4
82048.12 0.1
964.125 0
80240.025 0.3
124980.018 0.4
89750.359 0
146501.764 0.4
114489.672 1.1
57965.04 0.5
75183.336 0.9
165194 0
160487.357 0.4
14576 1.3
133451.923 0.3
78945.215 1.1
14603.694 2
9234.004 0.3
214.464 1.1
24767.332 0.3
57535.105 0.8
131895.906 0.3
165194 0.6
111389 1.1
111499.905 0.1
167031.46 0.3
25429.897 3.1
144150.778 0.5
144281.904 0.4
141950.105 1.5
202562.977 0
120705.338 1.2
42911.144 1.3
139126.037 0
108816.381 0.1
41363.796 0.1
33469.17 2
47910.5 0.4
54413.716 0.7
170636.671 0
43040 0.7
234937.42 0
205063.721 0
188438 0
133750 1.2
147730.552 0
86662.615 0.8
61402.476 0.3
158149.004 0
88695.131 1.4
136047.047 0.2
109230.085 0
124308.154 0.7
87816.73 0
81537.599 0.5
180104.995 0.3
6262.55 0.7
228267.356 0
80966 0.3
194883 0.2
129373.707 0.3
21629.84 1
223260.476 0
83670.158 0.5
39805.224 0.3
88298.944 0.1
3177.445 0.3
92273.725 0.5
248418.605 0
130810.13 0.6
224852.345 0.2
168634.294 0.3
193630.72 0
222552.236 0
106673.678 0.6
163098.112 0.3
37559.483 1.3
132906.099 0.7
159622.15 0
211399.672 0.3
194946.536 0.3
117871.367 0.5
141569.872 0.5
27530.266 1.3
138300.58 0.1
111401.079 0
71280.625 0
130224.031 0
165729.725 0
108737.76 0.7
191961.495 0
21437.594 1.6
110431.914 0
104758.5 0.7
124837.394 0
250001 0
21424.172 1.6
91842.209 0.3
193209.864 0
56044.308 1.5
41258.59 0
126648.47 0.2
33239.002 1.5
158998 0
132646.448 0.2
18022.075 1.2
165322.709 0.2
165468.85 0.3
173782.5 0.5
179740.149 0
118133.611 0.6
142745.589 0
171690.8 0
18863.124 0.6
156449.5 1
63797.382 0.4
189537.618 0.2
45282.76 0.4
145922.735 0.3
37601.393 0
69127.144 0
143206.73 0.2
140162.362 0
210436.605 0.2
166164.546 0.2
168951.017 0.3
147151.242 0.2
14297.269 0.3
129066.038 0
78867.404 0
143111.686 1
170558.798 0.2
205991.595 0
209603.351 0.3
17149.614 0.3
138167.544 0.2
194424.707 0
250001 0.3
11650.942 0.7
59211 0
59211 3.1
};
\addplot [line width=1pt, red!80!black] coordinates {
    (104.358, 0.928125794593659) 
    (250001.0, -0.0257468637869019)
};
\end{axis}

\end{tikzpicture}
		}
		\vspace{-16pt}
		\caption{Among election precincts, median household income (horizontal axis) is negatively correlated with percent of ballots showing a weak order (vertical axis; $r = -0.4$, $p < 0.001$).}
		\label{fig:income-scatterplot}
	\end{subfigure}
	\caption{Ballot data from the 2019 mayoral election in San Francisco.}
\end{figure}
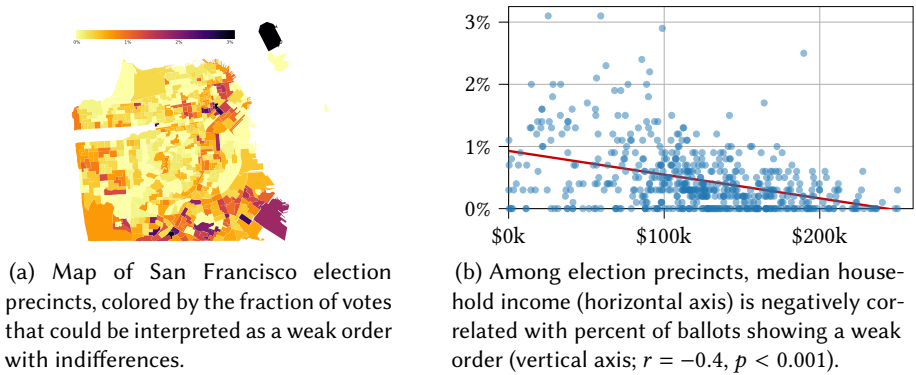

For how many ballots would this make a difference, in practice?
One can try to quantify this via the number of overvotes reported by election officers. \citet{mccune2023oakland} reports that in a very close 2021 City Council election in Portland, Maine, a different treatment of overvotes would have changed the winner.
One issue with depending on reported overvotes is that not all ballots containing overvotes encode valid weak orders -- in particular they might assign multiple ranks to the same candidate and are therefore definitely invalid.
Conveniently, San Francisco makes very detailed vote data \href{https://sfelections.sfgov.org/november-5-2019-election-results-detailed-reports}{publicly available \ExternalLink} in a CVR JSON format, and even complements these with image scans of every ballot.
We analyzed this data to see how many ballots could be interpreted as a weak order with at least one indifference between candidates.
For the 2019 mayoral election in San Francisco, we found 899 such ballots out of 206\,117 ballots submitted (0.4\%); see \Cref{fig:sfballots} for some examples. Voters casting such weak order ballots are geographically concentrated (\Cref{fig:sfmap}) and the fraction of ballots that encode weak orders in specific election precincts correlates negatively with median income (\Cref{fig:income-scatterplot}).
This mirrors previous finding of higher rates of invalid ballots in precincts with lower incomes in San Francisco \citep{neely2015overvoting} and New York City \citep{cormack2023}, as well as of the frequency of overvotes \citep{pettigrew2023ballot}. In Scotland's 2017 local elections, 1.6\% of ballots were rejected because of multiple top choices \citep{scotland2017}.

\begin{figure}[t]
	\begin{subfigure}{0.32\textwidth}
		\includegraphics[width=\linewidth]{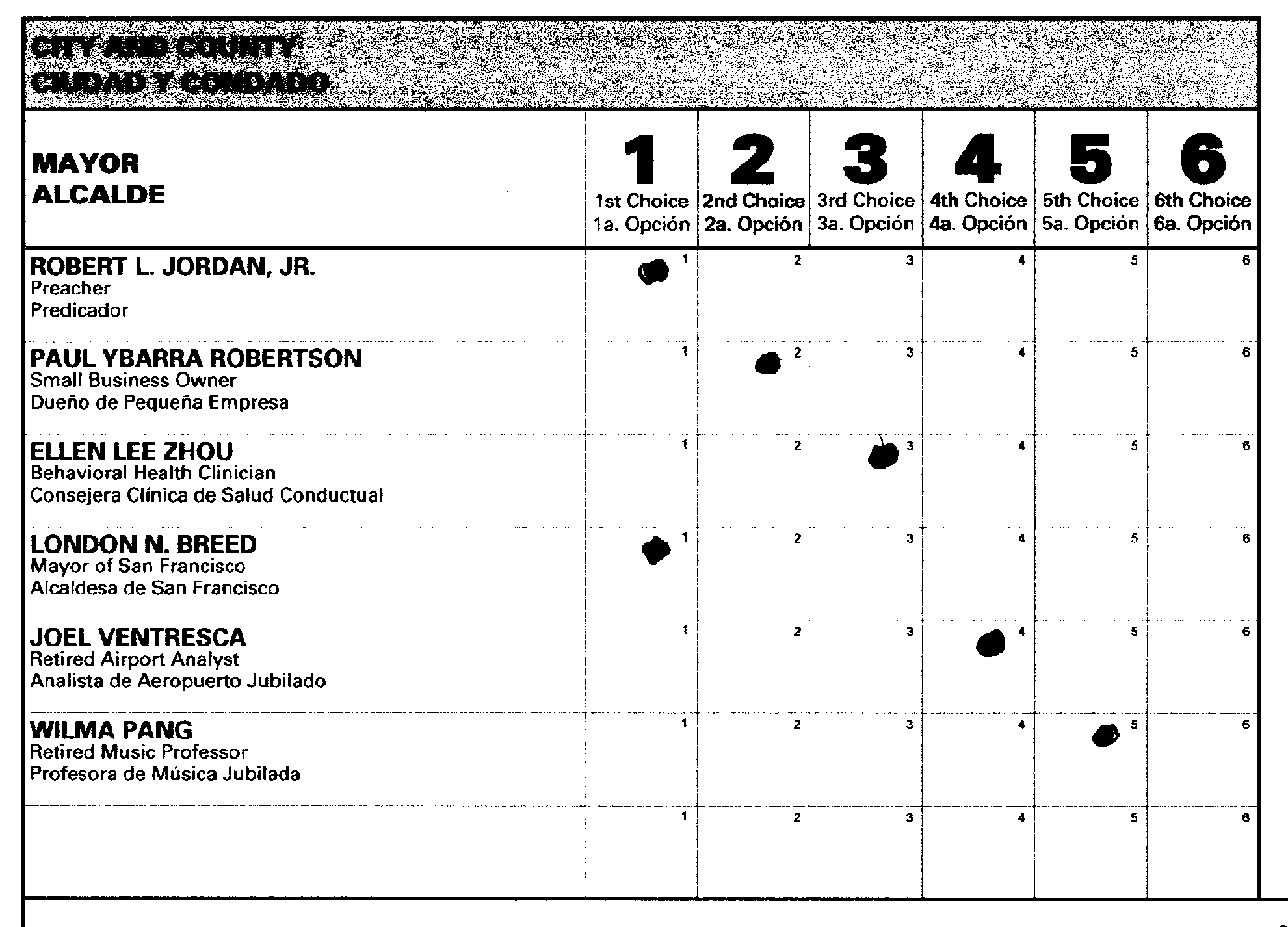}
		\caption{Two top choices}
	\end{subfigure}\hspace{5pt}
	\begin{subfigure}{0.32\textwidth}
		\includegraphics[width=\linewidth]{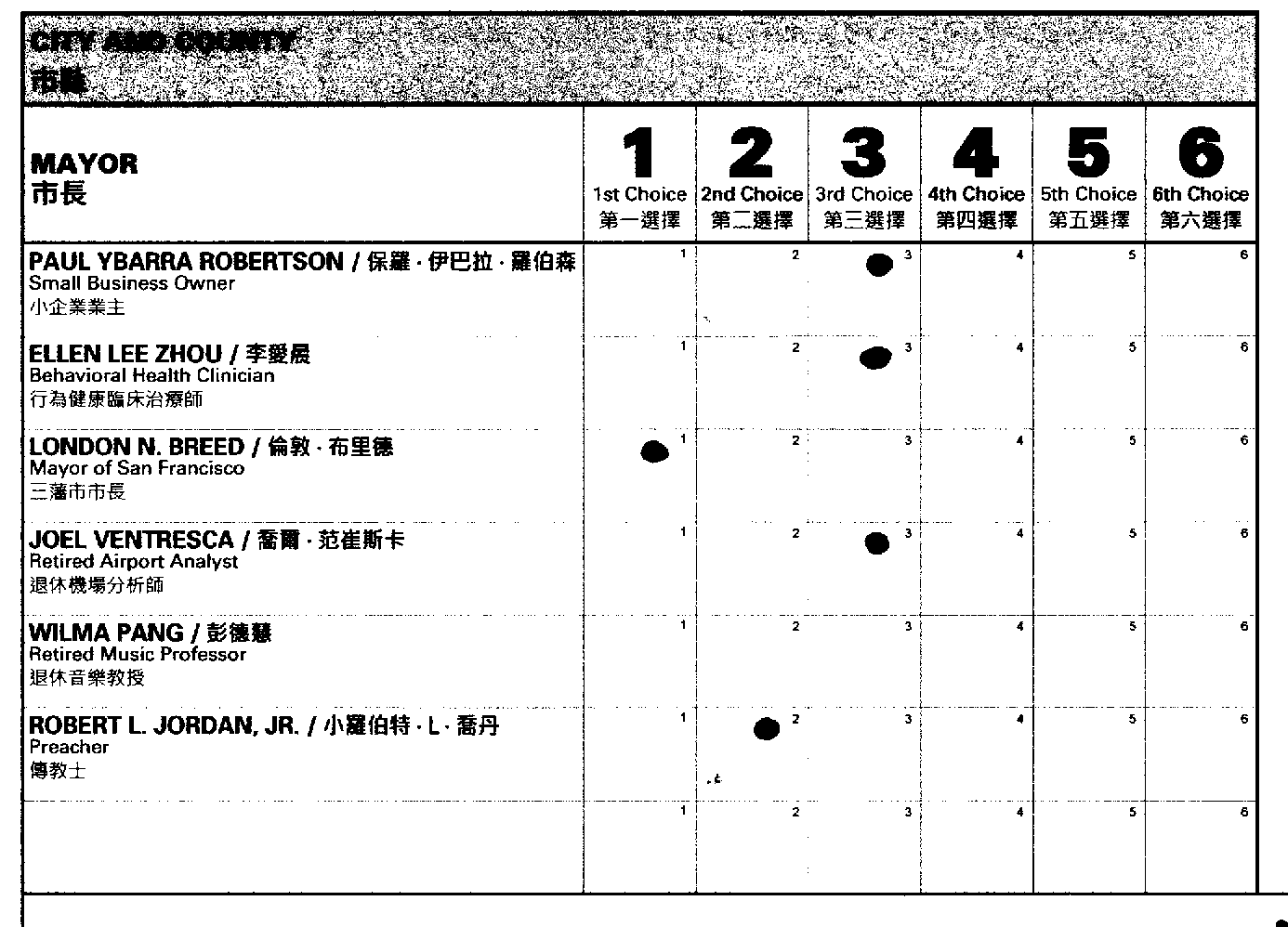}
		\caption{Vetoing a candidate}
		\label{fig:sfballots:veto}
	\end{subfigure}\hspace{5pt}
	\begin{subfigure}{0.32\textwidth}
		\includegraphics[width=\linewidth]{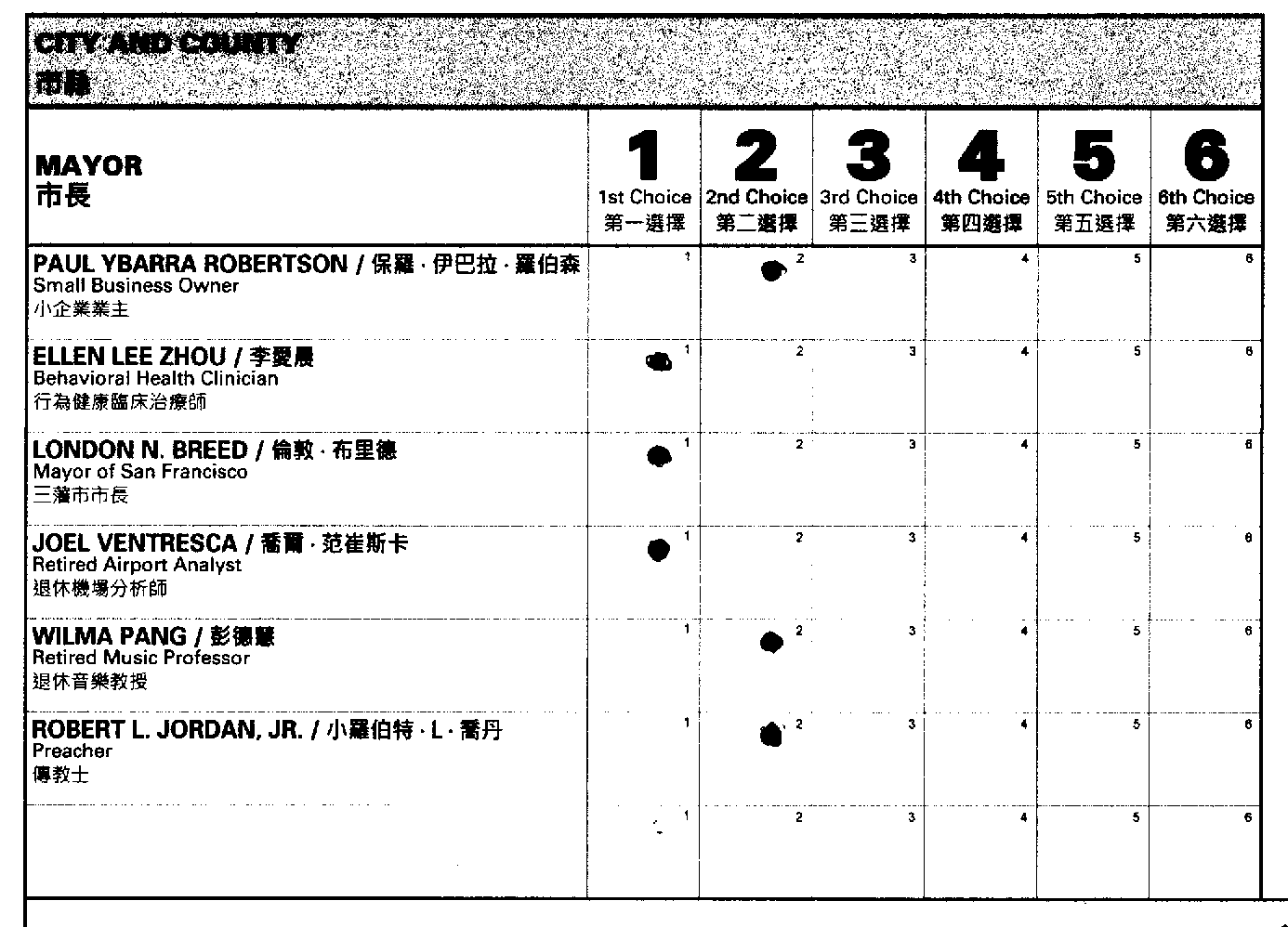}
		\caption{An approval vote}
		\label{fig:sfballots:approval}
	\end{subfigure}
	\vspace{-18pt}
	\caption{Examples of ballots that can be interpreted as weak orders (2019 mayoral election in San Francisco).}
	\label{fig:sfballots}
\end{figure}

(4) \emph{Better alignment with candidate campaign communications.} 
In practice, many voters decide how to vote based on candidate campaigns and advertisements. In Australia (where truncated rankings are not allowed), political parties issue recommended rankings. But in the American context (for example in New York City), most campaigns only ask supporters to rank their candidate \#1, without providing guidance about how to rank other candidates \citep{cormack2023}. Voters who are convinced by several candidates might therefore be interested to rank all of them \#1, which Approval-IRV allows them to do.

(5) \emph{Better incentives to support one's favorite.} 
Under plurality voting, it is common strategic advice to not vote for one's favorite candidate $x$, but rather for one's favorite $y$ among front-runners. One advantage of IRV is that one can honestly vote for $x$, because should $x$ be eliminated, the vote is transferred to $y$. However, there is some risk with this honest approach, because it may be that one's vote is ``stuck'' at $x$ while $y$ gets eliminated (and $x$ also ends up losing). Insincerely ranking $y$ top could get $y$ elected. This phenomenon is sometimes called ``favorite betrayal'' \citep{small2010geometric}, or ``compromise strategic voting'' \citep{green2014strategic,graham2023examination}. Approval-IRV partially mitigates this issue by allowing voters to rank both $x$ and $y$ in top position.%
\footnote{There are still some situations where putting only $y$ on top is a successful manipulation; however we can at least show that under Approval-IRV, \emph{if} the sincere vote for $x$ leads to a win for $y$, then putting both $x$ and $y$ on top still leads to a win for $y$ (a certain monotonicity property), so the latter strategy is quite safe. See \Cref{sec:monotonicity}.}

(6) \emph{A compromise between Ranked Choice Voting and Approval Voting.}
As of 2024, the two most successful strands of the electoral reform movement in America are advocates for IRV (or ``ranked choice voting'') such as \href{https://fairvote.org/}{FairVote \ExternalLink}, and advocates for Approval Voting such as the \href{https://electionscience.org/}{Center for Election Science \ExternalLink}. Approval-IRV is a rule that combines properties of both of these rules, and may in some situations serve as a compromise position.

\medskip
While there are several advantages to moving to weak orders, there are also some issues that need to be considered. For example, Approval-IRV inherits the well-known downsides of linear-order IRV, such as monotonicity violations and failures to elect Condorcet winners. In addition, voter instructions explaining how to fill out a weak-order ballot correctly could become more complicated than current instructions for linear-order ballots, which could cause misunderstandings and offset some of the gains in terms of avoiding invalid ballots. 

\medskip\noindent
\textbf{Multi-winner voting.}\:
IRV also has a version that allows for multi-winner elections, known as Single Transferable Vote (STV), which is used in Ireland, New Zealand, the Australian senate, and elsewhere. Since the number of candidates in a multi-member district can be even higher than for single-winner elections,%
\footnote{In the 1983 Australian senate election, voters in New South Wales had to strictly rank 62 candidates for 10 seats. 11.1\% of ballots were invalid. In 1984, the election law for the senate was reformed to allow truncated rankings and to rank parties instead of candidates. This reduced the rate of invalid votes to 3.5\% [\href{http://psephos.adam-carr.net/countries/a/australia/1983/1983senatensw.txt}{1983 results \ExternalLink}, \href{https://www.abc.net.au/news/2015-09-23/the-origin-of-senate-group-ticket-voting-and-it-didnt-come-from-/9388658}{abc news article \ExternalLink}].}
moving to weak orders makes sense for STV also. We define \emph{Approval-STV} in a natural way, and show that it preserves the proportionality properties of linear-order STV.

\medskip\noindent
\textbf{Prior discussion in the literature.}\:
The advantages of allowing indifferences in an IRV or STV election have long been recognized. In a series of articles in the journal \emph{Voting Matters}, \citet[Section 6]{meek1994}, \citet{warren1996}, and \citet{hill2001} developed a way in which IRV and STV (and in fact, every ranking-based voting rule) can be generalized to weak orders. Their idea was to replace every weak order by several (weighted) ranking votes for all possible ways in which the indifferences can be broken. For example, a voter who has a complete ranking, except for an indifference between $a$ and $b$, would be replaced by a weight-$\frac12$ vote with $a$ ranked above $b$ and a weight-$\frac12$ vote with $b$ above $a$. Similarly, a voter who reports indifference between all candidates would be replaced by $m!$ votes each with weight $\smash{\frac{1}{m!}}$. After this replacement operation, linear-order IRV or STV is applied.

\citet[footnote 8]{aziz2020expanding} note that, as described, this leads to an algorithm that may take exponential time. However, there is an equivalent description of this rule that is polynomial-time and easy to understand: for a voter who currently ranks $t$ candidates on top, the voter assigns $1/t$ points to each of these candidates, and the rule repeatedly eliminates the candidate with the fewest points. Accordingly, we call this rule \emph{Split-IRV}. Note that Split-IRV is different from Approval-IRV; in the example shown in \Cref{fig:approval-irv-example}, $a$ is the winner under Approval-IRV, but under Split-IRV, $a$ is the first alternative to be eliminated (receiving only $\frac12 + \frac13$ points), and $b$ is the final winner. Split-IRV has some intuitive appeal since it encodes the idea that every voter has a ``single'' vote.%
\footnote{The idea of splitting $1$ point uniformly between approved candidates has been discussed in other contexts under the terms ``Satisfaction Approval Voting'' \citep{brams2015satisfaction} or as ``equal and even cumulative voting'' \citep[e.g.,][Section~3.2.2]{bardal2023eecv} and this scoring system is used by Peoria, Illinois, for its city council elections \citep[p.~6]{brockington2003peoria}.}
The multi-winner version, Split-STV, is in fact in practical use: \citet{mollison2023fair} reports that Split-STV ``was first used by the \href{https://www.johnmuirtrust.org/assets/000/001/714/JMT_Articles_of_Association_2021_original.pdf?1624286801\#page=32}{John Muir Trust \ExternalLink} (for Trustee elections) in 1998, and by the \href{https://www.lms.ac.uk/sites/lms.ac.uk/files/files/Single\%20Transferable\%20Vote.pdf}{London Mathematical Society \ExternalLink} in 1999'' and both still use Split-STV today. It is implemented in the \texttt{vote} package for R \citep[p.~682]{r-vote-package}.

To our knowledge, the only previous scholarly discussion of Approval-IRV and Approval-STV is by \citet[Section~18.2]{janson2016}.%
\footnote{Since about 1996, there have been sporadic discussions of Approval-IRV on internet forums, see e.g., the \texttt{election-methods} mailing list (\href{http://lists.electorama.com/pipermail/election-methods-electorama.com/1996-July/098704.html}{1996  \ExternalLink}, \href{https://www.mail-archive.com/election-methods-electorama.com@electorama.com/msg03363.html}{2004  \ExternalLink}), \href{https://electowiki.org/wiki/Single_transferable_vote\%23Ways_of_dealing_with_equal_rankings}{electowiki \ExternalLink}, and reddit (\href{https://www.reddit.com/r/EndFPTP/comments/e5h2uu/equalrank_stv/}{2019 \ExternalLink}). A 2004 \href{http://condorcet.ericgorr.net}{webtool} implements both Approval-IRV and Split-IRV.}
He discovered the possibility of generalizing IRV and STV via approval based on his historical analysis of voting methods used in Sweden in the early 1900s. He called this approach ``Phragm\'en's principle'' after the Swedish mathematician Lars Edvard Phragmén (1863--1937) who around 1903 had proposed a similar approach to extend certain voting rules to work with votes that have two levels of approval. Janson did not analyze the properties of Approval-IRV or Approval-STV, but wrote that ``it would be interesting to compare the two ways to handle weakly ordered lists. It seems that Phragm\'en’s principle [...] might have some advantages over splitting the vote between total orderings as described [by Meek and Hill].''
In this paper, we take up Janson's research program and confirm his prediction: Approval-IRV is a better generalization than Split-IRV.

\medskip\noindent
\textbf{Our results.}\:
We prove two characterization results showing that Approval-IRV is the unique extension of IRV to weak orders that satisfies certain desirable properties. Our characterizations operate within the large class of \emph{elimination scoring rules} which use sequential elimination, at each step eliminating the candidate with the lowest score, where the scores of the candidates are determined by some system of positional scoring rules.%
\footnote{There are no known axiomatic characterizations of elimination scoring rules among the class of all possible voting rules, and obtaining one has been a very long-standing open problem [\citealp{smith1973aggregation} (``we have nothing approaching a satisfactory characterization of point runoff systems''), \citealp{conitzer2009preference} (Conjecture~1); see \citealp{freeman2014axiomatic} for a related result].}
In the world of weak orders, positional scoring rules have a lot of flexibility since they can depend in arbitrary ways on the \emph{order type} of a weak order (i.e., on the number of alternatives in each indifference class). Both Approval-IRV and Split-IRV are examples of elimination scoring rules.

Our first result shows that Approval-IRV is the unique elimination scoring rule satisfying the axioms of \emph{independence of clones} and \emph{respect for cohesive majorities}. These axioms are generalizations to the weak order context of two properties that are characteristic features of IRV, and we feel that they should therefore be satisfied by a reasonable generalization of IRV to weak orders.

{
\renewcommand{\thetheorem}{\ref{thm:clonemajority}}
\begin{theorem}
	Approval-IRV is the unique elimination scoring rule satisfying independence of clones and respect for cohesive majorities.
\end{theorem}
}

The first condition is \citeauthor{tideman1987independence}'s \citeyearpar{tideman1987independence} independence of clones, which requires that if a candidate is cloned by inserting new candidates into the preference profile in such a way that all clones are ranked adjacently by all voters, then this should not affect the outcome: If a winner is cloned, then one of the clones should still be a winner; if a loser is cloned, the identity of the winner should not change at all. This property encodes a kind of resistance to the spoiler effect. \citet{tideman1987independence} showed that this axiom is satisfied by IRV for linear orders. We show that Approval-IRV continues to satisfy it for weak orders, while Split-IRV fails the axiom.

The second axiom is an axiom encoding the idea that a majority's preference should be followed. For IRV with linear orders, this is easy to define: if a majority of voters ranks a candidate in top position, then this candidate should be a winner. For weak orders, the definition is less straightforward. At a minimum, when a majority of voters has the exact same \emph{set} of top choices, then one of those candidates should win. Both Approval-IRV and Split-IRV satisfy this. But Approval-IRV satisfies a stronger axiom, guaranteeing influence to majorities even if they have some disagreement in their top choices. Our axiom of \emph{respect for cohesive majorities} requires that if a majority of voters all rank a candidate $x$ in top position (possibly among others), then the winning candidate must be top ranked by at least one of the voters in this majority. In the example of \Cref{fig:approval-irv-example}, this property applies to the left three voters who all rank $b$ on top, and would require that the winner is one of $a$, $b$, or $d$ (but not $c$ since none of those three voters ranks $c$ top). Split-STV fails this condition. So do many other voting rules, including Condorcet extensions such as Schulze and ranked pairs.

Our second result shows that Approval-IRV is the unique elimination scoring rule which extends IRV in a monotonic way.

{
\renewcommand{\thetheorem}{\ref{thm:monocharacterization}}
\begin{theorem}
	Approval-IRV is the unique elimination scoring rule that agrees with IRV on profiles of linear orders and satisfies indifference monotonicity.
\end{theorem}
}

\begin{wrapfigure}[8]{r}{0.32\textwidth}
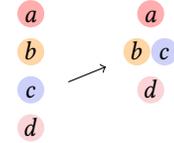

	\centering
	\weakorder{{a},{b},{c},{d}}
	\tikz{
		\path (0,0) -- (0.7,1);
		\draw[->] (0.1,0.8) -- (0.6,1);
	}
	\weakorder{{a},{b,c},{d},{}}
	\vspace{-8pt}
	\caption{Indifference monotonicity: if $c$ is the winner and a voter makes this change, then $c$ stays winning.}
	\label{fig:indifference-monotonicity}
\end{wrapfigure}

This statement may sound surprising, since IRV has long been famous for failing monotonicity axioms \citep{smith1973aggregation}; in particular, if $c$ is the winner and some voter moves up $c$ in their ranking by swapping its place with another candidate $b$, then $c$ may stop being the winner. In a sense, the reason for the monotonicity failure is that $b$ is now ranked lower by that voter, and may therefore be eliminated earlier than previously. We show that Approval-IRV satisfies a weak form of monotonicity: if $c$ is the winner, and a voter moves $c$ up so it is now in the same rank as $b$, then $c$ should remain the winner. We call this property \emph{indifference monotonicity}, since it only applies to the case where a voter introduces an indifference by shifting up some alternative that the voter previously did not have indifferent with any other alternative (see \Cref{fig:indifference-monotonicity}). Approval-IRV satisfies indifference monotonicity. As predicted by our characterization result, Split-IRV fails it.

For our third result, we turn to multi-winner elections. For linear orders, the Single Transferable Vote (STV) provides proportional representation to voters, which has been formalized by an axiom known as \emph{Proportionality for Solid Coalitions} (PSC). \citet{aziz2020expanding} generalized this axiom to weak orders, calling the resulting axiom \emph{generalized PSC}. They introduced a multi-winner voting rule called Expanding Approvals Rule (EAR) which works for weak orders and satisfies this axiom. We generalize STV to Approval-STV using the same ideas used in Approval-IRV, and prove that Approval-STV, thus defined, satisfies generalized PSC.
{
\renewcommand{\thetheorem}{\ref{thm:approval-stv-satisfies-gen-psc}}
\begin{theorem}
	Approval-STV satisfies generalized Proportionality for Solid Coalitions.
\end{theorem}
}
\noindent
On the other hand, Split-STV fails this axiom, even in the simple case of selecting only one winner.

\begin{table}[t]
	\centering
	\begin{tabular}{lcc}
		\toprule
		& Approval-IRV & Split-IRV \\
		\midrule
		\hyperref[sec:clones]{Independence of clones} & \cmark & \xmark \\
		\hyperref[sec:majority]{Respecting cohesive majorities} & \cmark & \xmark \\
		\hyperref[sec:monotonicity]{Indifference monotonicity} & \cmark & \xmark \\
		\bottomrule
	\end{tabular}
	\vspace{5pt}
	\caption{Comparison of properties satisfied by the rules.}
	\vspace{-15pt}
	\label{tbl:axioms}
\end{table}

\section{Preliminaries}
Let $C = \{c_1, \dots, c_m\}$ be a set of $m$ \emph{candidates} or \emph{alternatives} and $N = \{1, \dots, n\}$ be a set of $n$ \emph{voters}. 

\paragraph{Weak orders.}
A \emph{weak order} $\succeq$ is a complete and transitive binary relation over $C$, where we write $a \succ b$ if $a \pref b$ and $b \not\pref a$ (strict preference) and $a \sim b$ if both $a \pref b$ and $b \pref a$ (indifference). 
For sets $A, B \subseteq C$, we say $A \succ B$ if $a \succ b$ for all $a\in A$ and $b \in B$.

\begin{wrapfigure}{r}{0.29\textwidth}
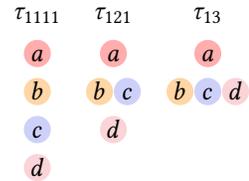

	\centering
	\votermultiplicity{$\tau_{1111}$}{\weakorder{{a},{b},{c},{d}}}\:
	\votermultiplicity{$\tau_{121}$}{\weakorder{{a},{b,c},{d},{}}}\:
	\votermultiplicity{$\tau_{13}$}{\weakorder{{a},{b,c,d},{},{}}}
	\caption{Examples of weak orders with different order types.}
	\label{fig:order-types}
\end{wrapfigure}

We will often write a weak order $\pref$ in the following format: $\{c_1, c_2\} \succ \{c_3\} \succ \{c_4\}$.
Formally, every weak order partitions $C$ into \emph{indifference classes} $C_1, \dots, C_k$ such that $C_1 \succ \dots \succ C_k$ and whenever $x, y \in C_j$ then $x \sim y$.
The \emph{order type} of a weak order $C_1 \succ \dots \succ C_k$ is the ordered list $\tau = (|C_1|,\dots,|C_k|)$ of the sizes of its indifference classes. We usually denote an order-type using the notation $\tau_{|C_1|\dots|C_k|}$ if this causes no ambiguity; for example $\tau_{121}$ denotes the order type $(1,2,1)$. See \Cref{fig:order-types} for examples. 
A weak order with order type $(1,1,\dots,1)$ (i.e., no indifferences) is called a \emph{linear order}, and a weak order with an order type $\tau$ with $|\tau| = 2$ (e.g., $\tau_{23}$) is called \emph{dichotomous}.
A \emph{profile} is a collection of weak orders $P = (\succeq_1,\dots,\succeq_n)$ where $\succeq_i$ is the weak order of voter $i$.

\paragraph{Scoring systems.}
We need the notion of a (positional) scoring system for weak orders. These are used to let weak orders assign scores to alternatives. A \emph{scoring system} is a map $s$ that maps order types $\tau$ to scores $s(\tau) = (s_1, \dots, s_{|\tau|}) \in \mathbb{R}_{\ge 0}^{|\tau|}$ where $s_{|\tau|} = 0$ without loss of generality.
Given a scoring system $s$ and a weak order $C_1 \succ \dots \succ C_k$ with order type $\tau$, an alternative $c \in C_j$ is assigned a score of $S_{\pref, s}(c) = s(\tau)_j$.
Given a profile $P = (\succeq_1,\dots,\succeq_n)$, the \emph{score} of alternative $c$ is $S(c) = S_{P, s}(c) = \sum_{i\in N} S_{\pref_i, s}(c)$. (We often drop the subscript if $P$ and $s$ are clear from the context.)

Examples of scoring systems are the \emph{approval scoring system} where $s(\tau) = (1,0,\dots,0)$ for all order types $\tau$ and the \emph{split scoring system} where $s(\tau) = (1/\tau(1), 0, \dots, 0)$. But the class of scoring systems is broad and includes systems like variants of Borda scores but also pathological systems such as those where weak orders with an odd number of indifference classes hand out approval scores, and other weak orders give out 0 score to all alternatives.%
\footnote{This definition gives the most general way to define positional scoring for weak orders in a neutral way (i.e., that doesn't depend on alternative names). It can also be derived from the Myerson--Pivato theorem \citep{myerson1995axiomatic,pivato2013variable} from which one can deduce that the voting rules for weak orders that satisfy reinforcement are exactly those selecting the highest-scoring alternatives based on a scoring system as we have defined it \citep[see also][]{Mork82a}.}
A scoring system is \emph{monotone} if for all order types $\tau$, we have  $s(\tau)_1 \ge s(\tau)_2 \ge \dots \ge s(\tau)_{|\tau|}$ and where $s(\tau)_{|\tau|} = 0$. We will only consider monotone scoring systems in this paper.%
\footnote{This is not a completely benign restriction, since an elimination scoring rule defined on a non-monotone score system can still satisfy axioms like unanimity \citep{freeman2014axiomatic}, but they seem very unnatural.}

\paragraph{Voting rules}
A \emph{voting rule} is a function $f$ that takes as input a profile $P$ and outputs a non-empty set of tied winners $W \subseteq C$ (usually a singleton). 
In this paper, we will focus on a particular family of voting rules, namely \emph{elimination scoring rules} \citep[Section 4]{smith1973aggregation}. 
The elimination scoring rule associated to a (monotone) scoring system $s$ selects a winning candidate by repeatedly removing a candidate with the lowest $s$-score from the profile, until only one candidate remains, who is the winner. More formally, and carefully handling the possibility of ties, the rule can be defined inductively as follows. Consider a profile $P$, defined on candidate set $C'$. If $C' = \{c\}$, we set $f(P) = \{c\}$. If $|C'| \ge 2$, let $L' = \{ c \in C' : S_{P, s}(c) \le S_{P, s}(d) \text{ for all $d\in C'$} \}$ be the set of lowest-scoring alternatives in $P$. Then we set $f(P) = \bigcup_{c \in L'} f(P|_{C' \setminus \{c\}})$, where $P|_{C' \setminus \{c\}}$ refers to the profile obtained from $P$ by restricting all the weak orders in it to the candidate set $C' \setminus \{c\}$ where $c$ has been deleted.%
\footnote{\label{footnote:put}This way of breaking ties is known as \emph{parallel-universe tie-breaking} \citep[Section 7]{conitzer2009preference}, and leads to an axiomatically well-behaved rule. Deciding whether a candidate is a winner under this tie-breaking is NP-complete \citep[Theorem 6.2]{boehmer2023rank}, but this should not be a problem in political elections with moderately low $m$ since the problem can be solved using an $O(2^m\cdot nm^2)$ time algorithm \citep[Theorem 6.1]{boehmer2023rank}.}

If we restrict ourselves to profiles of linear orders, various elimination scoring rules have already been defined and studied in the literature. The most well-known is Instant Runoff Voting (IRV), that uses plurality scores $s(\tau_{11\dots1}) = (1, 0, \dots, 0)$, but rules like the Baldwin rule and the Coombs rule (based on Borda and veto scores, respectively) have also been studied. 

For profiles of weak orders, we are particularly interested in two elimination scoring rules that generalize IRV in natural ways: \emph{Approval-IRV} is the elimination scoring rule based on the approval scoring system, and \emph{Split-IRV} is the one based on the split scoring system. We can interpret Approval-IRV as each voter assigning a full point to each top-ranked alternative. In  Split-IRV, the voter splits one point evenly among top-ranked alternatives -- for example if there are 3 top alternatives, then each receives a score of $\frac13$. Then repeatedly the lowest scoring alternative is eliminated, until only one alternative remains.

\section{Single-winner voting: Independence of clones and respecting cohesive majorities}

In this section we will axiomatically study elimination scoring rules and prove that Approval-IRV is the only elimination scoring rule to satisfy two axioms characteristic of IRV (when suitably generalized for weak orders): independence of clones and a condition about respecting majorities.

\subsection{Independence of clones}
\label{sec:clones}

Among social choice theorists, IRV stands out from most other ranking-based voting rules because it satisfies the independence of clones axioms \citep{tideman1987independence}. This axiom requires that adding new candidates to an election who are very similar to existing candidates (so that all voters rank them in adjacent positions) should not change the outcome. This is a desirable property in many contexts, including in political elections where some candidates may be running on similar policy platforms. Independence of clones is a way to avoid some forms of the ``spoiler effect'' that plagues elections using plurality voting. Apart from IRV, among standard voting rules only certain Condorcet extensions (such as Schulze's rule, ranked pairs, split cycle, and some tournament solutions) satisfy independence of clones \citep{schulze2011,holliday2023split}, and they are much more complicated than IRV and therefore harder to ``sell'' to voters and politicians.

\begin{figure}[t]
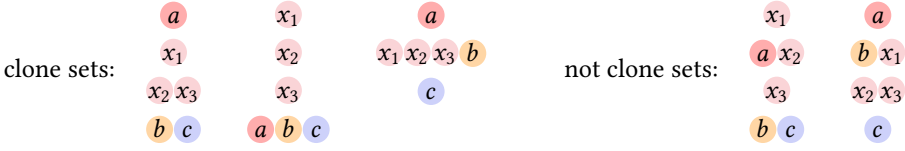

	\centering
	\tikz{\node[minimum height=2cm] {clone sets: \:}}
	\weakorder{{a},{x_1/alt-x},{x_2/alt-x,x_3/alt-x},{b,c}}
	\quad
	\weakorder{{x_1/alt-x},{x_2/alt-x},{x_3/alt-x},{a,b,c}}
	\quad
	\weakorder{{a},{x_1/alt-x,x_2/alt-x,x_3/alt-x,b},{c},{}}
	\qquad
	\tikz{\node[minimum height=2cm] {not clone sets: \:}}
	\weakorder{{x_1/alt-x},{a,x_2/alt-x},{x_3/alt-x},{b,c}}
	\quad
	\weakorder{{a},{b, x_1/alt-x},{x_2/alt-x,x_3/alt-x},{c}}
	\caption{Examples of $X = \{x_1, x_2, x_3\}$ being a clone set or not being a clone set.}
	\label{fig:clone-set-examples}
\end{figure}

Let us formally define independence of clones for weak orders. Given a profile $P$, a set of candidates $X \subseteq C$ is a \emph{clone set} if for every voter $i \in N$ and every candidate $c\not\in X$, we have either
\[
x \succ_i c \text{ for all $x\in X$,} 
\hspace{17pt}\text{or}\hspace{17pt} 
x \sim_i c \text{ for all $x\in X$,}
\hspace{17pt}\text{or}\hspace{17pt}
c \succ_i x \text{ for all $x\in X$.}
\]
See \Cref{fig:clone-set-examples} for examples and non-examples of clone sets.
We can now define the axiom:

\begin{definition}%
\label{def:independence-of-clones}
A voting rule $f$ satisfies \emph{independence of clones} if for all profiles $P$ with clone set $X \subseteq C$, letting $\hat P$ be the profile obtained by removing all but one candidate $\hat x$ from $X$, it holds that
\begin{enumerate}
    \item for every $c \notin X$, we have $c\in f(P)$ if and only if $c \in f(\hat P)$, and
    \item $\hat x\in f(\hat P)$ if and only if there exists $x \in X$ such that $x \in f(P)$.
\end{enumerate}
\end{definition}

Informally, this property states that adding or removing clones should not alter the result of the election. In particular, non-clones either stay winning or stay losing. The only thing that might change is that if a clone is winning, then another clone can win instead. 
\Cref{def:independence-of-clones} is equivalent to the definition for weak orders given by \citet{schulze2011} and by \citet[Section 5.3.2]{holliday2023split}. 
Restricted to dichotomous orders, it is equivalent to the version of the axiom defined by \citet{brandl2022approval}. Restricted to linear orders, it is equivalent to the original definition of \citet{tideman1987independence}. However, it is stronger than the way \citet[p.~186]{tideman1987independence} proposed to define independence of clones for weak orders (in his definition, it is not allowed for a clone to be ranked equally to a non-clone, so the third weak order shown in \Cref{fig:clone-set-examples} would not qualify).

Approval-IRV satisfies independence of clones. We can show this by adapting the standard proofs that linear-order IRV satisfies the axiom \citep{tideman1987independence,freeman2014axiomatic}, using the fact that cloning an alternative does not change the score of any non-clone alternative in any round of Approval-IRV. We defer the detailed proof to the appendix (\Cref{app:AVIRVclones-proof}).

\begin{restatable}{theorem}{AVIRVclones} \label{thm:AVIRVclones}
    Approval-IRV is independent of clones.
\end{restatable}

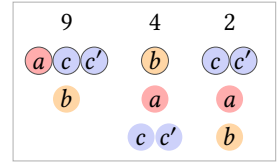
\begin{wrapfigure}{r}{0.26\textwidth}
	\centering
	\begin{tikzpicture}
		\node at (0,0) [draw=black!30] (step1) {	
			\votermultiplicity{9}{\weakorder[border]{{a,c,c'},{b},{}}}
			\votermultiplicity{4}{\weakorder[border]{{b},{a},{c,c'}}}
			\votermultiplicity{2}{\weakorder[border]{{c,c'},{a},{b}}}
		};
	\end{tikzpicture}
	\vspace{-8pt}
	\caption{Split-IRV fails independence of clones.}
	\label{fig:split-fails-clone}
\end{wrapfigure}

However, Split-IRV fails independence of clones. A simple counter-example is shown in \Cref{fig:split-fails-clone}, where $a$ is eliminated with a score of $3$ while other alternatives have a score of $4$; then finally the clones $c$ and $c'$ win. However, when removing the clone $c'$ of $c$ from the profile, we get a profile where the score of $a$ is $4.5$, the score of $b$ is $4$, and the score of $c$ is $6.5$. Thus, $b$ gets eliminated, after which $a$ wins the majority vote against $c$. Thus, the winner changed from $c$/$c'$ to $a$, which is not allowed by independence of clones.
Still, there are elimination scoring rules other than Approval-IRV that satisfy independence of clones. One example is the rule that is like Approval-IRV, but with scoring vector $s(\tau) = (\frac{1}{2}, 0, \dots, 0)$ for order types $\tau$ such that $\tau(1) \ge 2$ or such that $|\tau| = 2$. Note that this rule still generalizes IRV on linear orders and approval voting on dichotomous orders.

\subsection{Respecting cohesive majorities}
\label{sec:majority}

\begin{wrapfigure}[9]{r}{0.27\textwidth}
	\centering
	\begin{tikzpicture}
		\node at (0,0) [draw=black!30] (step1) {	
			\votermultiplicity{47\%}{\weakorder{{a,b},{c},{d},{}}}
			\votermultiplicity{4\%}{\weakorder{{a},{b},{c},{d}}}
			\votermultiplicity{25\%}{\weakorder{{c},{b},{d},{a}}}
			\votermultiplicity{24\%}{\weakorder{{d},{b},{c},{a}}}
		};
	\end{tikzpicture}
	\vspace{-7pt}
	\caption{A problem with electing majority alternatives.}
	\label{fig:majority-alternative}
\end{wrapfigure}
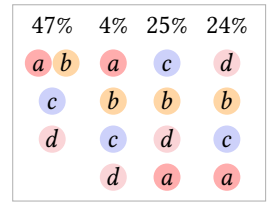
Another characteristic property of linear-order IRV (often known as the \emph{majority criterion}) is that if a majority of voters places some candidate in top position, then that candidate wins. There are several ways to define such a condition for weak orders. A weak version would be \emph{respect for unanimous majorities} \citep{brandl2022approval} which would say that if a majority of voters rank the same set $T$ of alternatives in top position, then only alternatives in $T$ should be winners. This is satisfied by both Approval-IRV and Split-IRV. 
An axiom that is intuitive but undesirably strong is \emph{select some majority alternative} which demands that if some alternative $c \in C$ is a majority alternative (i.e., is in the top indifference class of a majority), then all winning alternatives must be majority alternatives (noting that there could be several). To see why this might be undesirable, consider the profile in \Cref{fig:majority-alternative} where this axiom would demand that $a$ is the winner. However, alternative $b$ is strictly preferred over $a$ by $49\%$ of voters, while only $4\%$ of voters have the opposite strict preference. Thus, there is a good argument that $b$ should be the winner, and indeed it is the winner under Approval-IRV (and it is also the Condorcet winner). In contrast, Split-IRV selects $a$.

\begin{table}[t]
	\centering
	\scalebox{0.8}{
	\begin{tabular}{lcc}
		\toprule
		& Approval-IRV & Split-IRV \\
		\midrule
		Respect for unanimous majorities & \cmark & \cmark \\
		Respect for cohesive majorities & \cmark & \xmark \\
		Select some majority alternative & \xmark & \xmark \\
		\bottomrule
	\end{tabular}}
	\vspace{5pt}
	\caption{Comparison of majority properties satisfied by the rules.}
	\vspace{-15pt}
	\label{tbl:majority-axioms}
\end{table}

We propose an axiom called ``respect for cohesive majorities'' that logically lies between these two axioms (see \Cref{tbl:majority-axioms}). It says that if there is a majority of voters who rank some alternative $c$ on top (so they are ``cohesive''), possibly among others, then the winning alternative must be ranked top by at least one member of that majority. 

\renewcommand{\top}{\mathrm{top}}
Given a profile $P$, write $\top_i = \{c \in C : c \pref_i d \text{ for all $d\in C$}\}$ for the top alternatives of voter $i$.
\begin{definition}
	\label{def:majority}
	A voting rule $f$ \emph{respects cohesive majorities} if for all profiles $P$ and all subsets of voters $S \subseteq N$ such that $|S| > \frac{n}{2}$ and $\bigcap_{i\in S} \top_i \neq \emptyset$, we have $f(P) \subseteq \bigcup_{i \in S} \top_i$.
\end{definition}

Intuitively, if a group of $\frac{n}{2} + t$ voters is a cohesive majority that jointly ranks $c$ on top, then the group has a ``default claim'' that $c$ should be elected. If a rule wants to override this claim, the axiom says that it must choose some alternative $x$ that at least $t$ voters in that group think is at least as good as $c$ (and that they hence rank on top). Making this choice justifiably overrides the claim because at most $\frac{n}{2}$ voters remain in the group, who do not form a majority on their own.

\Cref{fig:split-fails-majority} shows an example profile where a majority of voters (19 out of 37) have top sets $\{a,b,c\}$, $\{a,b\}$, and $\{a,c\}$. They are cohesive as they agree on $a$. Therefore, the axiom demands that the winning candidates are either $a$, $b$, or $c$, but not $d$. On the shown profile, $d$ is the winner under Split-IRV, which therefore fails our axiom. Approval-IRV selects $a$, consistently with our axiom.

Respect for cohesive majorities is inspired by proportionality axioms, and is in fact a special case of ``generalized PSC'' axiom (with $|T|=k=1$) which we discuss in \Cref{sec:multi-winner}. Our axiom also implies the ``weak defensive strategy criterion'' \citep{ossipoff2000}, which says that if a majority of voters strictly prefers $a$ over $b$, then there should be votes that the majority can submit that ensure that $b$ loses. The submitted votes should be \emph{sincere} (they only differ from the true preferences by including extra indifferences). Respect for cohesive majorities implies this property, since each member of the majority can report all candidates that they strictly prefer to $b$ as their top alternatives.

\begin{figure}[t]
	\begin{tikzpicture}
		\node at (0,0) [draw=black!30] (step1) {	
			\votermultiplicity{9}{\weakorder[border]{{a,b,c},{d},{}}}
			\votermultiplicity{5}{\weakorder[border]{{a,b},{d},{c}}}
			\votermultiplicity{5}{\weakorder[border]{{a,c},{d},{b}}}
			\votermultiplicity{8}{\weakorder[border]{{b,c,d},{a},{}}}
			\votermultiplicity{10}{\weakorder[border]{{d},{a,b,c},{}}}
		};
	\end{tikzpicture}
	\caption{Split-IRV violates respect for cohesive majorities because it eliminates $a$, then $b$ and $c$, and elects $d$.}
	\label{fig:split-fails-majority}
\end{figure}
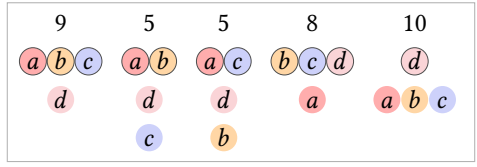

As we mentioned, Split-IRV fails to respect cohesive majorities. 
In addition, all standard Condorcet extensions fail the axiom, as the following result implies (see \citealp{zwicker2016} for definitions):

\begin{proposition}
	\label{prop:c2-majority}
	No C2 voting rule (i.e., one that depends only on the pairwise majority margins $|\{ i \in N : a \succ_i b \}| - |\{ i \in N : b \succ_i a \}|$, like Schulze or ranked pairs) respects cohesive majorities.
\end{proposition}
\begin{proof}
	Consider the following profile:
	\begin{center}
		\begin{tikzpicture}
			[vertex/.style={inner sep=0pt},
			 margin/.style={inner sep=0.8pt, fill=white,transform shape,scale=0.85}]
			\node at (-0.7,0) [draw=black!30] (step1) {	
				\weakorder{{a},{c},{b}}
				\weakorder{{a, b},{c},{}}
				\weakorder{{a, b},{c},{}}
				\weakorder{{c},{b},{a}}
				\weakorder{{c},{b},{a}}
			};
			\node at (3.9,0) {with weighted majority graph};
			\node[vertex] at (7.177,0.5) (a) {\alternative{a}}; %
			\node[vertex] at (7.754,-0.5) (c) {\alternative{c}}; %
			\node[vertex] at (6.6,-0.5) (b) {\alternative{b}};
			\draw (a) edge[->] node[margin] {$1$} (c);
			\draw (c) edge[->] node[margin] {$1$} (b);
			\draw (b) edge[->] node[margin] {$1$} (a);
		\end{tikzpicture}
	\end{center}
	Without loss of generality, because the voting rule is C2, we may assume $c$ is a winner in this profile (otherwise we can rename alternatives in the profile but retain the same weighted majority graph). However, a majority of voters places $a$ in top position and none of these voters puts $c$ in top position. Therefore, by respect for cohesive majorities, $c$ must not be a winner, a contradiction.
\end{proof}

On the other hand, Approval-IRV does satisfy the axiom.%
\footnote{In fact, the proof can be adapted to show that Approval-IRV satisfies generalized PSC (\Cref{def:gen-psc}) which is stronger than respect for cohesive majorities (it is a type of ``mutual majority criterion''). 
This result does not follow immediately from \Cref{thm:approval-stv-satisfies-gen-psc} because there are slight differences between Approval-IRV and Approval-STV for $k = 1$, see \Cref{rem:no-majority-stopping}.}

\begin{restatable}{theorem}{AVIRVweakmaj}
Approval-IRV respects cohesive majorities.
\end{restatable}
\begin{proof}
	Let $S$ be a subset of more than $\frac{n}{2}$ voters who all rank candidate $c$ top. We need to show that the winner under Approval-IRV is ranked top by some voter in $S$. If $c$ is the winner, we are done. Otherwise, consider the time just before $c$ gets eliminated. At that time, a majority of voters (namely at least $S$) ranks $c$ on top, so it has approval score more than $\frac{n}{2}$. Yet because it is about to be eliminated, this is the lowest score. Therefore, for every remaining candidate $x$, we see that $x$ is ranked top by a majority. Since two strict majorities must intersect, it follows that there is some voter $i \in S$ who ranks $x$ top at this time, and because $c$ is still present, $i$ must have ranked $x$ top from the start. Therefore, at this time all remaining alternatives are top alternatives for someone in $S$, and hence this will also be true of the eventual winner.
\end{proof}

\subsection{Characterization}

We have seen that Approval-IRV satisfies independence of clones and respects cohesive majorities, two properties that are characteristic of IRV in the linear-order context. Split-IRV fails both axioms. We will now prove that Approval-IRV is in fact characterized by these two properties within the class of elimination scoring rules.%
\footnote{\citet[Theorem 1]{freeman2014axiomatic} previously characterized linear-order IRV to be the only linear-order elimination scoring rule satisfying independence of clones. However, we cannot use their result since they interpreted IRV as a social welfare function that outputs a \emph{ranking} of alternatives (the elimination order). In this setting, independence of clones is a much stronger property. Indeed, \citet{freeman2014axiomatic} write that they ``do not know whether other nontrivial runoff scoring rules would satisfy the property'' when looking at social choice functions, because their ``proofs relied heavily on being able to alter \emph{some} position in the ranking'', while in our characterization we need to reason about the set of final winners.}
The proof is quite involved, since we need to characterize the scoring vector for every possible order type. It appears in \Cref{app:clonemajority-proof}.

\begin{restatable}{theorem}{clonemajority} \label{thm:clonemajority}
When there are at least $m \ge 4$ alternatives, Approval-IRV is the only elimination scoring rule that satisfies independence of clones and respects cohesive majorities.
\end{restatable}
\begin{proof}[Proof sketch]
	We start by showing that all linear order types $\tau_{11\dots1}$ have score vector $(1,0,\dots,0)$. Then, most of the proof is spent on showing that the score vector of all order types on 3 and 4 alternatives must be $(1,0, \dots,0)$. Respect for cohesive majorities allows us to deduce that only the first score value can be positive, and it also imposes some bounds on that value (e.g., we exhibit a profile showing that every scoring system where $\tau_{21}$ gives less than $\frac59$ points will exhibit a violation of the majority axiom). We can then repeatedly use independence of clones to deduce inequalities between the scores of an order type for 4 alternatives (such as $\tau_{31}$) and an order type for 3 alternatives (such as $\tau_{21}$). An interesting feature of this proof is that we need to consider the cases of $m = 3$ and $m = 4$ simultaneously, because when only 3 alternatives exist, Split-IRV satisfies both axioms; it only starts failing them once a fourth alternative is available.
	
	Once we know that all order types with 3 or 4 alternatives (as well as $\tau_{212}$) have score vector $(1,0, \dots,0)$, we then use a variety of induction steps to deduce the same result for all dichotomous order types. Finally, another induction characterizes the score vector for all order types.
\end{proof}

The proof depends on constructing many families of profiles that witness violations of one of the two axioms in the characterization. We found it helpful to use linear programs to obtain such profiles for particular non-approval scoring systems, and then generalize these examples. To do this, we guess (by iterating) the elimination order of the profile (for clones, the two profiles), and then build an LP that has a continuous variable for each possible weak order, indicating the fraction of the profile(s) made up by voters with that weak order. Since the scoring system and the elimination order are fixed, we can encode the behavior of the elimination scoring rule as linear constraints.

The axioms in the characterization are independent. An elimination scoring rule that respects cohesive majorities but fails independence of clones is the one that's like Approval-IRV except that $s(\tau_{21}) = (\frac12, 0)$ (i.e., it flips to Split-IRV when it gets down to $m=3$). An elimination scoring rule that satisfies independence of clones but fails to respect majorities is given at the very end of \Cref{sec:clones}. 
We leave for future work the question of whether there exists a natural rule satisfying both axioms (but not being an elimination scoring rule).

\section{Single-winner voting: Indifference monotonicity}
\label{sec:monotonicity}

Our second characterization of Approval-IRV has a somewhat different flavor: we are going to show that it is the unique monotonic way to extend IRV to weak orders. To make this claim precise, we need to be careful since IRV fails most monotonicity properties (which is in fact true for all elimination scoring rules \citep[Theorem 2]{smith1973aggregation}). However, we identify a natural notion of ``indifference monotonicity'' that is satisfied by Approval-IRV.

Suppose $c \in f(P)$ is a winner given some profile $P$. Monotonicity requires that if we change the profile $P$ to make $c$ look stronger, then $c$ should still be the winner. However, if we swap $c$ with an alternative $d$ above it (thereby making $c$ stronger), then under an elimination scoring rule we might now have caused an earlier elimination of $d$ which might lead to $c$ losing. In a sense, the reason for the failure is that not only did we make $c$ stronger, but we also made $d$ weaker. Our notion of indifference monotonicity considers a very restricted class of change that makes $c$ stronger without negatively affecting any other candidate. 

Formally, let's say that a \emph{$c$-hover} is the following transformation from one weak order to another:
\[
C_1 \succ \dots \succ C_j \succ \{c\} \succ C_{j+2} \succ \dots \succ C_k
\quad\longmapsto\quad
C_1 \succ \dots \succ C_j \cup \{c\} \succ C_{j+2} \succ \dots \succ C_k
\]
A $c$-hover starts from a weak order in which $c$ is in a singleton indifference class, and ends in the weak order where $c$ has joined the indifference class just above it (see \Cref{fig:indifference-monotonicity} in the introduction). Note that a $c$-hover cannot be applied to a weak order where $c$ is indifferent with another alternative. 
Our indifference monotonicity axiom applies only to changes corresponding to $c$-hovers.
\begin{definition}
	\label{def:indifference-monotonicity}
	A voting rule $f$ is \emph{indifference monotonic} if for every profile $P$ and every $c \in f(P)$, whenever $\hat P$ is obtained from $P$ by applying $c$-hovers to some votes in $P$, we have $c \in f(\hat P)$.
\end{definition}

\noindent
We now show that Approval-IRV satisfies indifference monotonicity.

\begin{restatable}{theorem}{AVIRVmono} \label{thm:AVIRVmono}
	Approval-IRV is indifference monotonic.
\end{restatable}

\begin{proof}
	Let $f$ be Approval-IRV, let $w \in f(P)$, and let $\hat P$ be obtained from $P$ by applying some $w$-hovers. Suppose that in $P$, Approval-IRV eliminates candidates in the order $c_1, \dots, c_{m-1}, w$. We will show that this is also a valid elimination order in $\hat P$, which implies that $w \in f(\hat P)$, as required.
	
	Suppose for a contradiction that this was not the case, and let round $1 \le t \le m-1$ be the first time when in $\hat P$, Approval-IRV cannot eliminate candidate $c_t$ (which can be eliminated in round $t$ under $P$). Let us compare the scores of candidates at this point in $P$ and $\hat P$ after the elimination of candidates $c_1, \dots, c_{t-1}$. 
	Note that, by definition of $w$-hover, every voter has the same current top indifference class under $P$ and $\hat P$, except that some voters may additionally have $w$ in their top indifference class under $\hat P$. (Some voters' top indifference class might be $\{w\}$ under both profiles, but under $\hat P$ it might be $\{w\}$ because $w$ was hovered into a bigger indifference class whose other members have by now been eliminated.)
	Thus, all candidates have the same scores under $P$ and $\hat P$, except that $w$ may have a higher score under $\hat P$. But under $P$, the score of $c_t$ was lowest (and thus weakly lower than the score of $w$), and so the same is true in $\hat P$, contradicting our assumption that $c_t$ could not be eliminated at this time by Approval-IRV.
\end{proof}

We can now state our second characterization of Approval-IRV. Here, we say that a voting rule $f$ is \emph{consistent with IRV} if for every profile $P$ of linear orders, $f(P)$ is the set of IRV winners. 

\begin{restatable}{theorem}{monocharacterization} \label{thm:monocharacterization}
	Approval-IRV is the unique elimination scoring rule that is consistent with IRV on profiles of linear orders and satisfies indifference monotonicity.
\end{restatable}
\begin{proof}[Proof Sketch]
	Consistency with IRV implies that linear orders must have score vector $(1,0, \dots,0)$. Observe that any order type $\tau$ can be obtained from a linear order by successively applying candidate-hovers, for example $\tau_{111111} \to \tau_{21111} \to \tau_{3111} \to \tau_{312}$. Using indifference monotonicity, this allows us to inductively deduce that every order type has score vector $(1,0, \dots,0)$ by constructing counterexample profiles ruling out all other score vectors.
	The full proof appears in \Cref{app:monocharacterization-proof}.
\end{proof}

The properties of this characterization result are logically independent. Split-IRV is an elimination scoring rule that is consistent with IRV but fails indifference monotonicity. An elimination scoring rule that satisfies indifference monotonicity (but is not consistent with IRV) can be constructed by extending the Baldwin rule instead of IRV. For this, take the elimination scoring rule which associates to each order type $\tau$ the Borda-style scoring vector $s(\tau) = (m-1, m-1-\tau(1), m-1-\tau(1)-\tau(2), \dots, m-1 - \sum_{i=1}^{k-1}\tau(i))$. One can use the same proof as in \Cref{thm:AVIRVmono} to show that this rule is indifference monotonic. 
Finally, there are rules that satisfy both properties of the characterization but that are not elimination scoring rules, such as the rule that returns the result of IRV on profiles of linear orders, and the whole set of candidates $C$ when the profile is not linear.

\section{Multi-winner voting}
\label{sec:multi-winner}

In a multi-winner election, our goal is to select a committee of $k$ candidates from $C$. The analog of IRV for multi-winner elections is called the Single Transferable Vote (STV), which is defined in a way that provides proportional representation (giving groups of voters an amount of representation in the committee that is proportional to their size). STV is used for political elections in Scotland, Ireland, and New Zealand, among other places.

STV can be extended to work for weak orders just like how we defined Approval-IRV. This gives rise to a rule we call Approval-STV. (We will not describe STV for linear orders separately, since it is just the restriction of Approval-STV to the case when everyone submits a linear order.)
Approval-STV works by assigning each voter a \emph{budget} of 1 monetary unit (often referred to as the \emph{weight} of the voter). Voters will spend their budgets on electing candidates that they rank highly, and because voters start out with equal budgets, this produces a proportional outcome. The cost (or \emph{quota}) of electing a candidate is $q = n/(k+1)$, where $n$ is the number of voters and $k$ is the number of seats.%
\footnote{The quota $q = n/(k+1)$ is known as the \emph{Droop quota} which we will use throughout. However, everything we say holds analogously for the Hare quota $q = n/k$, with strict inequalities appropriately replaced by weak inequalities. For the Droop quota, note that we require that supporters of a candidate have \emph{strictly more} than $q$ money, but we only charge the supporters \emph{exactly} $q$ money. This is more elegant and robust than taking $q = n/(k+1) + \epsilon$ or rounding the quota up.}
Repeatedly, Approval-STV looks for a candidate that appears in the top indifference class of voters who together have a remaining budget of strictly more than $q$. If such a candidate exists, one is selected and added to the committee $W$, and the voters placing it in their top indifference class pay $q$ for it (with this cost divided between them); then this winning candidate is removed from the profile. If no such candidate exists, some candidate is \emph{eliminated} and removed from the profile, and we go to the next iteration.
The method is described in pseudocode in \Cref{fig:approval-stv-code}.

\begin{figure}[ht]
	\centering
	\begin{minipage}{0.7\textwidth}
		\RestyleAlgo{plain}
		\newcommand{\setvars}[3]{%
			\makebox[7.5pt][l]{$#1$} 
			$\gets$ 
			\makebox[54pt][l]{#2} 
			(#3)}
		\newcommand{\remaining}{R}
		\newcommand{\supporters}{\mathsf{supp}}
		\begin{algorithm}[H]
			\SetAlgoNoEnd
			\setvars{b_i}{$1$ for all $i \in N$}{the budget of each voter} \\
			\setvars{\remaining}{$C$}{remaining candidates} \\
			\setvars{W}{$\emptyset$}{selected candidates} \\
			\While{$\remaining \neq \emptyset$}
			{   
				$\supporters(c) \gets \{i\in N : c \pref_i \remaining\}$ for all $c \in \remaining$ \\
				\eIf{there is $c \in \remaining$ with $\sum_{i\in \supporters(c)} b_i > q$}
				{
					select one such candidate $c$ \\
					reduce the budgets $(b_i)_{i \in \supporters(c)}$ by a total amount of $q$ \\
					$W \gets W \cup \{c\}$ \\
					$\remaining \gets \remaining \setminus \{c\}$ \\
				}
				{
					select some candidate $c \in R$ to eliminate \\
					$\remaining \gets \remaining \setminus \{c\}$
				}
			}
			\Return $W$
		\end{algorithm}
	\end{minipage}
	\vspace{-10pt}
	\caption{Approval-STV}
	\label{fig:approval-stv-code}
\end{figure}

Clearly this rule elects at most $k$ candidates, because after $k$ selections, $kq$ money has been spent. In the beginning, the total amount of money was $(k+1)q$, so only $q$ money is left, and therefore there can't be any additional candidates whose supporters have a remaining budget of \emph{strictly more} than $q$. The rule is in fact guaranteed to select at least $k$ (and therefore exactly $k$) candidates; this can be deduced from \Cref{thm:approval-stv-satisfies-gen-psc} below (with $S = N$, $T = C$, and $\ell = k$).

Our description of Approval-STV above and in \Cref{fig:approval-stv-code} left several details vague: it did not specify (1) how to select the next winning candidates should several candidates be eligible, (2) how to divide the cost $q$ of the candidate among its supporters, and (3) how to decide which candidate to eliminate. The reason we left these vague is that our main result (that Approval-STV satisfies an axiom providing proportional representation) holds no matter how these three issues are decided. However, it is worth noting how STV is commonly used, for example in Scotland: (1) we select the candidate whose supporters have the highest amount of remaining budget, (2) the cost is divided by multiplying the remaining budget of each supporter by a common factor (``Gregory method''), and (3) the candidate whose supporters have the lowest amount of remaining budget is eliminated. However, there are other possible choices, including mirroring the payment scheme of the Method of Equal Shares which has been proposed for approval-based multi-winner elections and for participatory budgeting \citep{peters2020laminar,peters2021pb}.%

\begin{remark}
\label{rem:no-majority-stopping} Note that Approval-IRV is the same thing as Approval-STV for $k=1$ with Hare quota, but not with Droop quota. For instance, In the example shown in \Cref{fig:approval-irv-example}, more than $q = n/2$ (the Droop quota for $k = 1$) voters rank $b$ on top. Thus, Approval-STV (with Droop quota) would not do any elimination and instead immediately select $b$. Approval-IRV performs eliminations and ends up selecting alternative $a$. For the reason we do it this way, see \Cref{fig:majority-alternative} and the surrounding discussion, and also note that Approval-STV fails independence of clones even for $k=1$ (since a majority alternative could stop being a majority alternative if cloned).
\qed
\end{remark}

One can also define \emph{Split-STV}. It is the rule that replaces weak orders by weighted linear orders breaking indifferences in all possible ways, and then applies linear-order STV \citep{warren1996}. Equivalently, one can adapt the algorithm in \Cref{fig:approval-stv-code} to say that if a voter supports $t$ candidates, then the voter is only willing to contribute $b_i/t$ (instead of $b_i$) to elect a supported candidate.

STV is a rule providing proportional representation. \citet[pp.~282--283]{dummett1984voting} formalized this using a property saying that groups of voters should have a representation in the winning committee proportional to the group's size. He argued via an example that STV satisfies his property. \citet[footnotes 1 and 2]{tideman1995single} termed this axiom \emph{proportionality for solid coalitions} (PSC) and gave a proof sketch that STV satisfies it. \citet{woodall1994properties,woodall1997monotonicity} refers to it as ``Droop proportionality''. \citet[Appendix C]{aziz2020expanding} give a formal proof that STV satisfies PSC.

\definecolor{T1}{HTML}{cfe8ef}
\definecolor{T2}{HTML}{cfe8ef}
\definecolor{T3}{HTML}{cfe8ef}
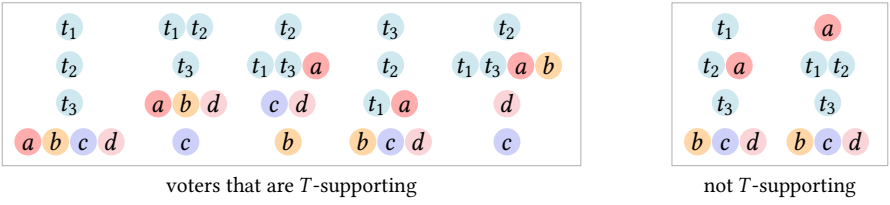
\begin{figure}[t]
	\centering
	\begin{tikzpicture}
		\node at (0,0) [draw=black!30] (step1) {	
			\weakorder{{t_1/T1},{t_2/T2},{t_3/T3},{a,b,c,d}}
			\weakorder{{t_1/T1,t_2/T2},{t_3/T3},{a,b,d},{c}}
			\weakorder{{t_2/T2},{t_1/T1,t_3/T3,a},{c,d},{b}}
			\weakorder{{t_3/T3},{t_2/T2},{t_1/T1,a},{b,c,d}}
			\weakorder{{t_2/T2},{t_1/T1,t_3/T3,a,b},{d},{c}}
		};
		\node[transform shape, scale=0.82] at (0,-1.35) {voters that are $T$-supporting};
		
		\node at (6.5,0) [draw=black!30] (step1) {	
			\weakorder{{t_1/T1},{t_2/T2, a},{t_3/T3},{b,c,d}}
			\weakorder{{a},{t_1/T1,t_2/T2},{t_3/T3},{b,c,d}}
		};
		\node[transform shape, scale=0.82] at (6.5,-1.35) {not $T$-supporting};
	\end{tikzpicture}
	\caption{An illustration of the definition of generalized PSC. For $T = \{t_1, t_2, t_3\}$, the 5 voters in the left box are $T$-supporting, since every alternative in $T$ is ranked weakly higher than every alternative not in $T$. The closure of $T$ with respect to the 5 voters is $\{t_1, t_2, t_3, a, b\}$, since some voters rank $a$ and/or $b$ on the same level as an alternative from $T$. The voters in the right box are not $T$-supporting because of where they rank $a$.}
	\label{fig:psc}
\end{figure}

To define the axiom, for a set $T \subseteq C$ of candidates, let us say that a weak order $\pref$ is \emph{$T$-supporting} if $T \pref C\setminus T$, i.e., if for every $t\in T$ and every $c \in C \setminus T$, we have $t \pref c$. A group $S \subseteq N$ of voters is $T$-supporting if every voter in $S$ is $T$-supporting. For linear orders, PSC is defined as follows.

\begin{definition}[Proportionality for solid coalitions]
	Let $P$ be a profile of linear orders. Let $q = n/(k+1)$. A committee $W \subseteq C$ of size $k$ satisfies \emph{proportionality for solid coalitions (PSC)} if for all $S \subseteq N$ with $|S| > \ell \cdot q$ and all $T\subseteq C$ with $|T| \ge \ell$ such that  $S$ is $T$-supporting, we have $|W \cap T| \ge \ell$.
\end{definition}

\citet{aziz2020expanding} propose a way to generalize PSC to weak orders. Consider a group $S$ that is $T$-supporting. Note that since we have weak orders, it is not necessary that every voter in $S$ \emph{strictly} prefers $T$ above all other alternatives, see \Cref{fig:psc}. Instead, being $T$-supporting also allows there to be one ``mixed'' indifference class that contains both alternatives from within and from outside $T$, provided that all indifference classes above the mixed one only contain $T$-alternatives, and all indifference classes below it do not contain any $T$-alternatives. 

For a $T$-supporting group $S$, let us define the \emph{closure} of $T$ with respect to $S$ as
\newcommand{\closure}{\mathsf{closure}}
\[
\closure_S(T) = \{ c \in C: \text{for some $i \in S$ and some $t \in T$, we have $c \pref_i t$} \}.
\]
That is, the closure of $T$ contains the alternatives in $T$ together with all alternatives that appear in a ``mixed'' indifference class, looking at the voters in $S$. The caption of \Cref{fig:psc} contains an example. Generalized PSC requires that sufficiently many alternatives are elected from the closure of $T$.%
\footnote{\citet[Section 4]{brillpeters2023} propose a stronger proportionality axiom than generalized PSC that also works for weak orders, called rank-PJR+. Their axiom is satisfied by the Expanding Approvals Rule of \citet{aziz2020expanding}, but it is failed by STV even for linear orders, and so it is also failed by Approval-STV.}
\begin{definition}[Generalized PSC, \textup{\citealp{aziz2020expanding}}]
	\label{def:gen-psc}
	Consider a profile of $n$ weak orders, and let $q = n/(k+1)$. A committee $W \subseteq C$ of size $k$ satisfies \emph{generalized PSC} if for all $S \subseteq N$ with $|S| > \ell \cdot q$ and all $T\subseteq C$ with $|T| \ge \ell$ such that  $S$ is $T$-supporting, we have $|W \cap \closure_S(T)| \ge \ell$.
\end{definition}

Generalized PSC for committee size $k = 1$ and for sets $T$ with $|T|=1$ is equivalent to respect for cohesive majorities (\Cref{def:majority}). Therefore, Split-STV fails generalized PSC even for $k = 1$, by the example shown in \Cref{fig:split-fails-majority}.%
\footnote{A much weaker way of defining PSC for weak orders is to define ``$T$-supporting'' as meaning $T \succ C\setminus T$, so that no voter is indifferent between alternatives in $T$ and outside $T$ (so only the two left-most examples in \Cref{fig:psc} count). In that case $\closure_S(T) = T$, so PSC just requires $|W \cap T| \ge \ell$. This definition is satisfied by both Approval-STV and Split-STV.}
On the other hand, Approval-STV satisfies the axiom.

\begin{theorem}
	\label{thm:approval-stv-satisfies-gen-psc}
	Approval-STV satisfies generalized PSC.
\end{theorem}
\begin{proof}
	Let $W$ be the committee elected by Approval-STV.
	Let $S \subseteq N$ be a group of voters with $|S| > \ell q$ and let $T \subseteq C$ with $|T| \ge \ell$ be such that $S$ is $T$-supporting. Write $U = \closure_S(T)$.
	We have to prove that $|W \cap U| \ge \ell$.
	
	Throughout the execution of Approval-STV, let $s = |W \cap U|$ denote the current `satisfaction' of $S$ and let $r = |R|$ denote the current number of `remaining' (neither selected nor eliminated) candidates in $T$. We claim that the following invariant is true throughout the execution:
	\[
	s + r \ge \ell.
	\]
	Certainly this holds at the start because $r = |T| \ge \ell$. At the end of the execution, we have $r = 0$ (because at the end, all candidates have been either selected or eliminated), and hence the claim implies $s \ge \ell$, proving the theorem.
	To prove the claim, we will show that every `action' (selection or elimination of a candidate) of Approval-STV preserves the correctness of the claim.
	
	Whenever a candidate outside $T$ is selected, then $s$ may go up and $r$ will stay constant, so the claim stays correct. When a candidate in $T$ is selected, then $s$ goes up by 1 and $r$ goes down by 1, so $s + r$ stays constant, and the claim stays correct. If a candidate outside $T$ is eliminated, then $s + r$ stays constant, and the claim stays correct. If a candidate in $T$ is eliminated, but currently we have $s + r > \ell$, then $s$ stays constant and $r$ goes down by $1$, so the claim will hold true after elimination.
	
	So the only case to worry about is that a candidate from $T$ is eliminated at a time when $s + r = \ell$. We will prove that this cannot happen. For contradiction, suppose it does happen.
	Because a candidate in $T$ is about to be eliminated, at least one candidate in $T$ has always been present thus far. Therefore, by definition of $U$ as the closure of $T$, up to now, the top indifference class of every voter in $S$ has been a subset of $U$. Hence, thus far, voters in $S$ have only spent money on candidates in $U$. The current satisfaction is $s$. Therefore, the voters in $S$ can have paid at most $sq$ money until now. The voters in $S$ started out with more than $\ell q$ money. Hence, the voters in $S$ together have more than $(\ell - s)q = rq$ money left.
	
	We prove that at least one of the remaining $r$ candidates in $T$ has supporters who together have more than $q$ money left. It follows by definition of Approval-STV that in this step, a candidate is selected instead of eliminated, which means that the current case is impossible. Label the remaining candidates in $T$ as $c_1, \dots, c_r$. Suppose for a contradiction that for each of them, the voters who place that candidate in their top indifference class have at most $q$ money left. Write $S_1 \subseteq S$ for all the voters in $S$ who currently place $c_1$ top. Write $S_2$ for all voters in $S \setminus S_1$ who currently place $c_2$ top, and so on. Note that for each $j \in [r]$, the voters in $S_j$ have at most $q$ money left. However, because $T \pref_i C \setminus T$ for all $i \in S$, it is the case that every voter in $S$ places at least one of the remaining candidates in $T$ in their top indifference class. Hence $S = S_1 \cup \dots \cup S_r$. But the $S_j$ are pairwise disjoint, and so it follows that $S$ has at most $rq$ money left, a contradiction.
\end{proof}

This proof also shows that linear-order STV satisfies classic PSC, in an arguably clearer way than existing proofs in the literature.

\section{Experiments}
\label{sec:expe}
In this last section, we will experimentally compare the generalizations of IRV and STV to weak orders using synthetic and real data. %
We tested our rules on a variety of datasets, including synthetic data (sampled from various models such as impartial culture, mixture of Mallows and Euclidean), and real data (experiments on voting methods conducted during French presidential elections~\citep{bouveret_2018_1199545,delemazure_2024_10568799}, and the actual votes in elections that took place in Dublin, Ireland, from Preflib~\citep{MaWa13a}). Details on these datasets can be found in \Cref{app:experiments-datasets}.

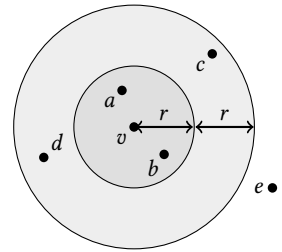
\begin{wrapfigure}[11]{r}{0.27\textwidth}
	\centering
	\begin{tikzpicture}[transform shape, scale=0.89]
		\def\r{0.9}

		\draw[fill=gray!13] (0,0) circle (2*\r cm);
		\draw[fill=gray!25] (0,0) circle (\r cm);

		\fill (0,0) circle (2pt) node[below left=-1pt] {$v$};

		\fill (-0.2*\r,0.6*\r) circle (2pt) node[below left=-1pt] {$a$};
		\fill (0.5*\r,-0.45*\r) circle (2pt) node[below left=-2pt] {$b$};

		\fill (1.3*\r,1.2*\r) circle (2pt) node[below left=-1pt] {$c$};
		\fill (-1.5*\r,-0.5*\r) circle (2pt) node[above right=-0.5pt] {$d$};

		\fill (2.3*\r,-1*\r) circle (2pt) node[left] {$e$};

		\draw[<->, thick] (0.02,0) -- node[midway, above=-1pt, sloped] {$r$} (\r-0.02,0);
		\draw[<->, thick] (\r+0.02,0) -- node[midway, above=-1pt, sloped] {$r$} (2*\r-0.01,0);
		
	\end{tikzpicture}
	\vspace{-5pt}
	\caption{A voter with preferences $\{a,b\} \succ \{c,d\} \succ \{e\}$.}
	\label{fig:radius-model}
	
\end{wrapfigure}
All datasets concern profiles of linear order, so we need a method to turn them into weak orders. 
We used two such methods. 
The ``coin-flip'' method with parameter $p\in [0,1]$ works as follows: in an order $\succ$, for each pair of consecutively ranked candidates $a$ and $b$, we add a tie between them (and thus put them in the same indifference class) with probability $p$. 
For example, for the linear order $a \succ b \succ c \succ d$ we throw 3 independent coins, one for each occurrence of the ``$\succ$'' symbol, and replace a strict preference by an indifference when the coin comes up heads (which happens with probability $p$). If the coins come up tails, heads, tails, the resulting weak order is $\{a\} \succ \{b, c\} \succ \{d\}$.
The ``radius'' method is specific to Euclidean models, in which voters $v$ and candidates $c$ are placed in random locations $p(v), p(c) \in \mathbb{R}^d$ in Euclidean space. The method is parameterized by a radius $r \ge 0$, which from the perspective of voter $v$ divides the candidates into sets $C_k = \{(k-1)r \le \|p(c)-p(v)\| <kr\}$. This produces the weak order $C_1 \succ C_2 \succ \dots$ for voter~$v$. \Cref{fig:radius-model} illustrates this model with an example.

\subsection{Evaluation of the winner} \label{sec:expe-eval}

\begin{figure}
	\centering
	\includegraphics[width=1\linewidth]{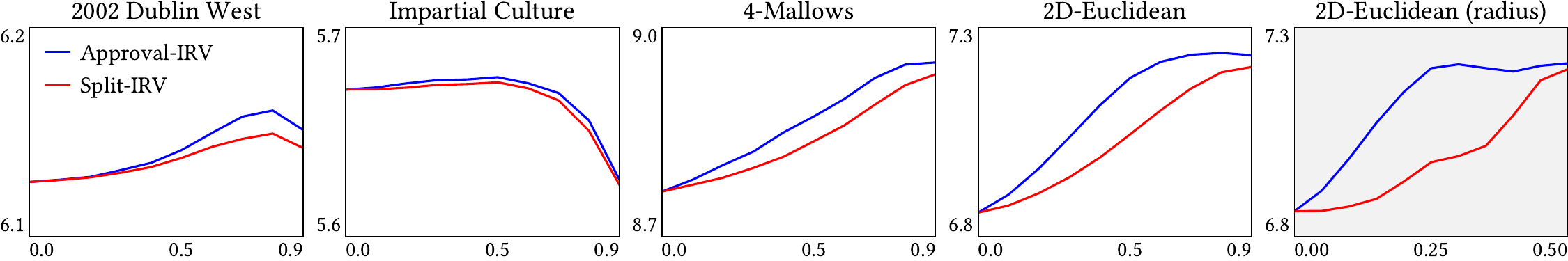}
	\vspace{-20pt}
	\caption{Average Borda score of the winner (normalized by dividing by $n$) for various datasets.}
	\label{fig:exp-borda-main}
\end{figure}

First, we compare the single-winner rules: Approval-IRV and Split-IRV. We randomly sample 10\,000 instances from each of our datasets (for real data, we sample voters at random with replacement). For the coin-flip method, we tested values of $p$ between $0$ (introducing no indifferences) and $0.9$; we exclude $p = 1$ since this leads to complete indifference. For the radius method, we tested values of $r$ between $0$ (no indifferences) and $0.5$. Each instance has $n=500$ voters and $m=10$ candidates, except for real data where we keep the original number of candidates $m$ (between $9$ and $14$).

In \Cref{fig:exp-borda-main}, we evaluate the winning candidate chosen by the two rules in terms of its Borda score, where the Borda score is computed with respect to the original profile of linear orders. While one might expect that introducing more indifferences (higher $p$ or $r$) would lead our rules to make lower-quality decision, in fact the opposite is the case. In almost all datasets, the Borda score of the candidate selected by linear-order IRV in the original profile ($p = 0$ or $r = 0$) is lower than the Borda score of the candidate selected by Approval-IRV and Split-IRV, on average over the instances. (The main exception is impartial culture.) We believe that the reason for this is that linear-order IRV depends mostly on plurality scores, whereas the weak orders allow the rules to identify candidates that are frequently ranked in high positions. \Cref{fig:exp-borda-main} also shows that the candidate selected by Approval-IRV has higher Borda score than Split-STV, consistently across datasets.

\begin{wrapfigure}{r}{0.46\textwidth}
	\includegraphics[width=\linewidth]{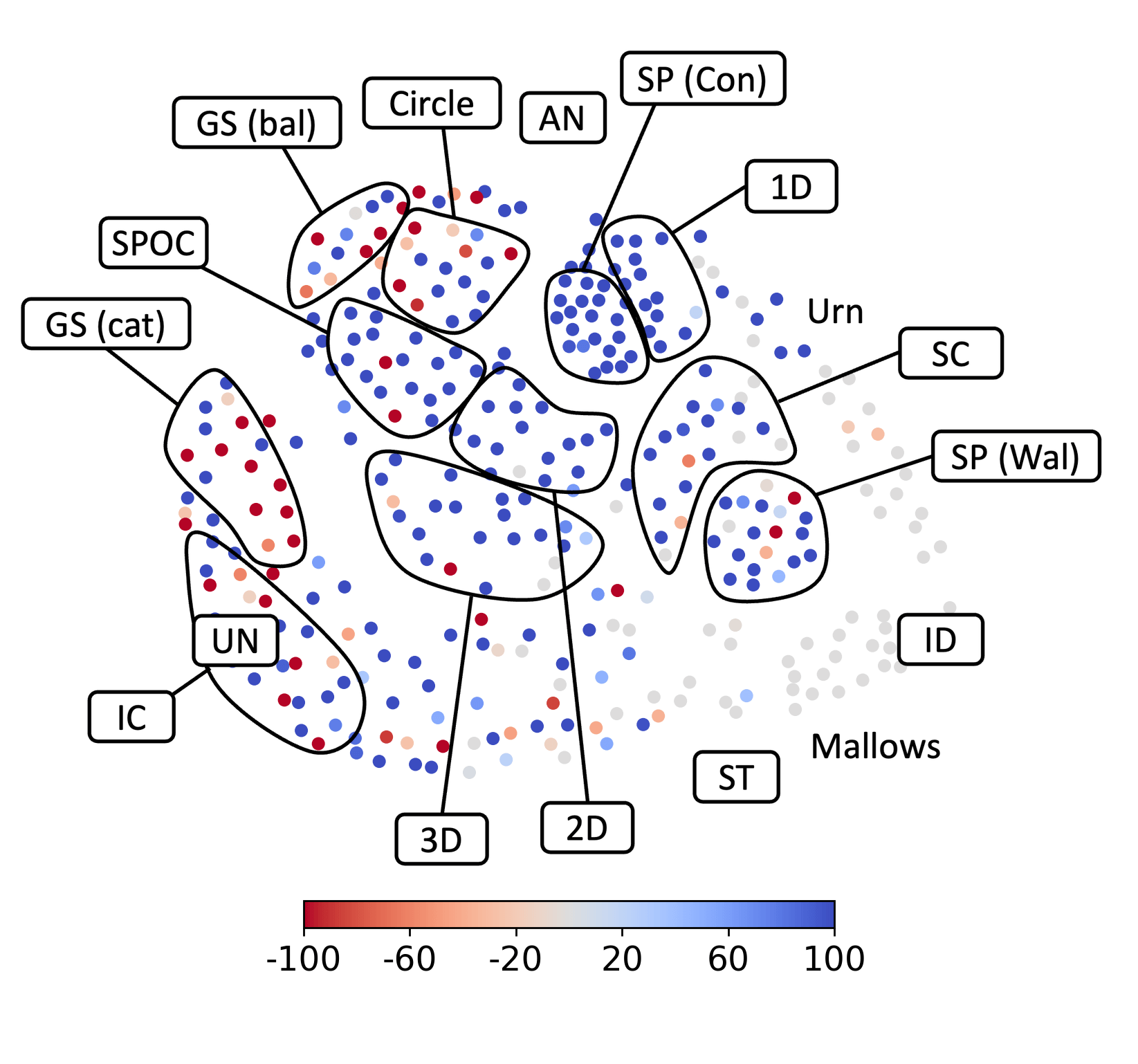}
	\vspace{-28pt}
	\caption{Map of elections, showing the difference in Borda score between the Approval-IRV and Split-IRV winner in the coin-flip model, with blue dots indicating that Approval-IRV selected a winner with higher Borda score.}
	\label{fig:exp-map-1}
\end{wrapfigure}
We checked that this observation is consistent over many random distributions over profiles, using the \emph{map of elections} framework \citep[Figure 1(a), 10$\times$50 isomorphic swap]{boehmer2022}. \Cref{fig:exp-map-1} shows a dot for each profile of the map, which is colored according to the difference in Borda score between the winner of Approval-IRV and Split-IRV, where we used the coin-flip method to transform the linear-order profile into a weak order one. For each dot, we averaged over the parameter $p = 0.1, 0.2, \dots, 0.9$ of the coin-flip method and over 50 random sampled weak order profiles for each $p$. In order to highlight the differences, we set the color scale to range from $-100$ to $100$, though some datasets have a difference higher than 2\,000. (In \Cref{app:map}, we show the same figure with a broader scale from $-800$ to $800$.) We observe that the Borda score of the Approval-IRV winner is generally higher than the one of the Split-IRV winner (blue dots), especially for structured preferences like single-peaked ones. For profiles close to those drawn from impartial culture, the difference is less pronounced (similar to what we see in \Cref{fig:exp-borda-main}, and some datasets show an advantage for Split-IRV. For profiles with little difference between voters (those close to identity) show almost no difference between the two rules, as there is often a clear winner that both rules will select.

\begin{wrapfigure}{r}{0.46\textwidth}
	\includegraphics[width=\linewidth]{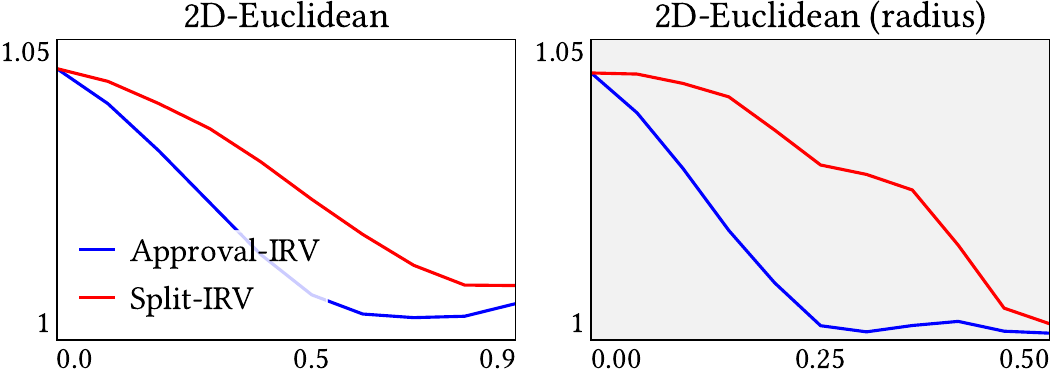}
	\vspace{-21pt}
	\caption{Average distortion of the winner.}
	\label{fig:exp-distortion-main}
\end{wrapfigure}
\Cref{fig:exp-distortion-main} confirms these findings with respect to the average \emph{distortion}, a quality measure of the selected candidate in terms of Euclidean distance instead of Borda scores (lower is better, see \Cref{app:experiments-results} for the definition). For large $p$ or $r$, Approval-IRV frequently selects the optimum candidate (distortion $=1$).
Finally, \Cref{fig:exp-similarity-main} shows the probability that the rules select the same candidate as linear-order IRV on the original profile. We see that Split-IRV is more ``similar'' to IRV than is Approval-IRV, consistently across datasets.
Additional discussion and more results (including about Condorcet winners) are shown in \Cref{app:experiments-results}.

\begin{figure}
	\centering
	\vspace{-5pt}
	\includegraphics[width=1\linewidth]{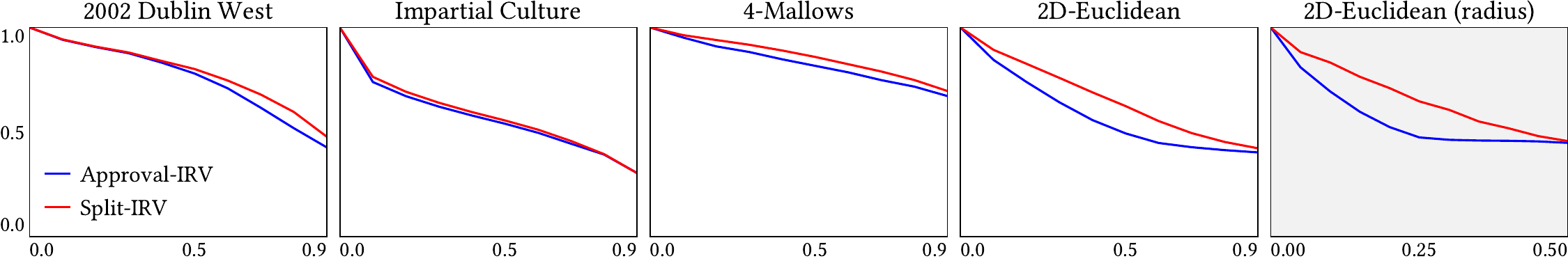}
	\vspace{-20pt}
	\caption{Frequency of agreement between the rule and linear-order IRV for various datasets.}
	\label{fig:exp-similarity-main}
\end{figure}

\subsection{Multi-winner} \label{sec:expe-multiwinner}

We now compare the multi-winner rules: Approval-STV and Split-STV. We also study the \emph{Expending Approvals Rules} (EAR) \citep{aziz2020expanding}. Approval-STV and EAR both satisfy the generalized PSC axiom for proportional representation, while Split-STV fails it. In our experiments, we wish to visually distinguish these rules, to allow us to judge their proportionality. We follow the experimental setup introduced by \citet{elkind2017multiwinner}, based on random 2D-Euclidean profiles. In our experiments, we used $n=100$ voters, $m=100$ candidates and $k=10$ seats.

\begin{figure}
	\footnotesize
	\setlength{\tabcolsep}{2pt}
	\begin{subfigure}{0.47\textwidth}
		\begin{tabular}{cccc}
			& Split-STV & Approval-STV & EAR \\
			\rot{\hspace{0.45cm}$r=0.1$} & \includegraphics[width=0.3\textwidth]{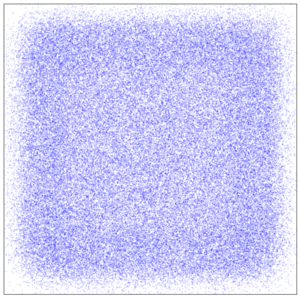}& \includegraphics[width=0.3\textwidth]{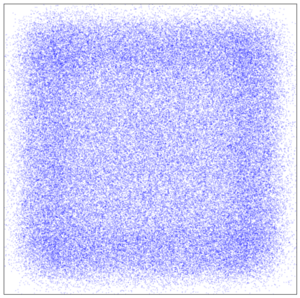}&\includegraphics[width=0.3\textwidth]{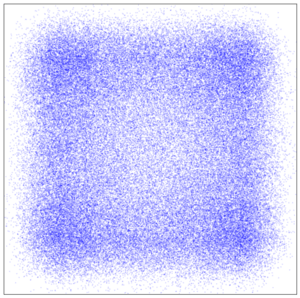} \\
			\rot{\hspace{0.45cm}$r=0.2$} & \includegraphics[width=0.3\textwidth]{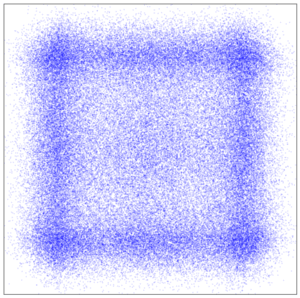}& \includegraphics[width=0.3\textwidth]{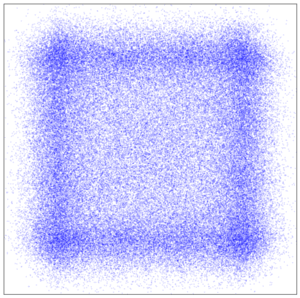}&\includegraphics[width=0.3\textwidth]{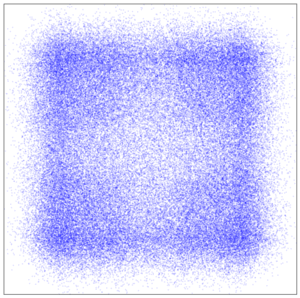} \\
			\rot{\hspace{0.45cm}$r=0.3$} & \includegraphics[width=0.3\textwidth]{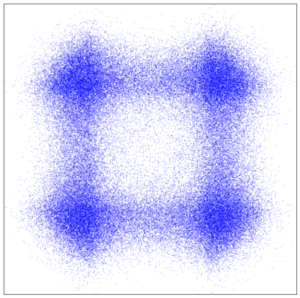}& \includegraphics[width=0.3\textwidth]{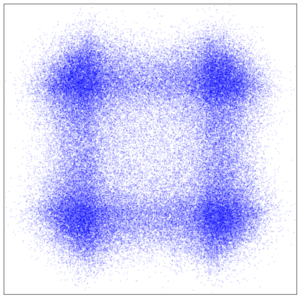}&\includegraphics[width=0.3\textwidth]{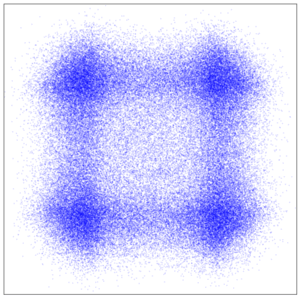} \\
		\end{tabular}
		\caption{Unit square}
	\end{subfigure}
	\quad
	\begin{subfigure}{0.47\textwidth}
		\begin{tabular}{c c c c}
			& Split-STV & Approval-STV & EAR \\
			\rot{\hspace{0.45cm}$r=0.1$} & \includegraphics[width=0.3\textwidth]{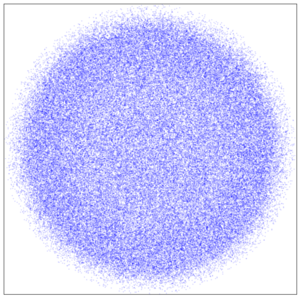}& \includegraphics[width=0.3\textwidth]{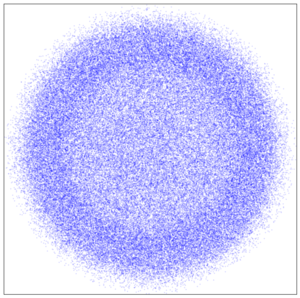}&\includegraphics[width=0.3\textwidth]{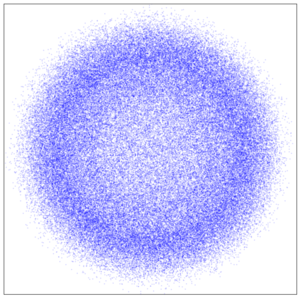} \\
			\rot{\hspace{0.45cm}$r=0.2$} & \includegraphics[width=0.3\textwidth]{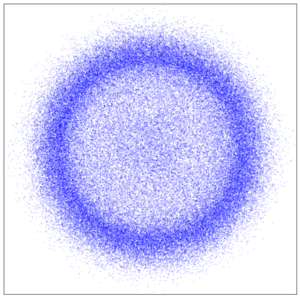}& \includegraphics[width=0.3\textwidth]{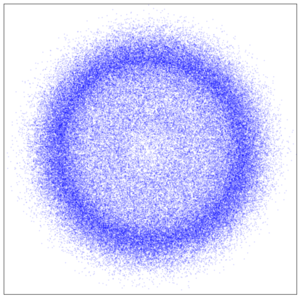}&\includegraphics[width=0.3\textwidth]{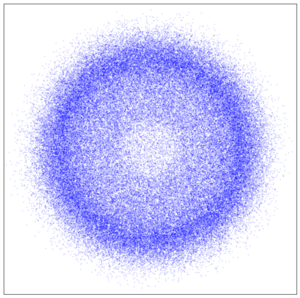} \\
			\rot{\hspace{0.45cm} $r=0.3$} & \includegraphics[width=0.3\textwidth]{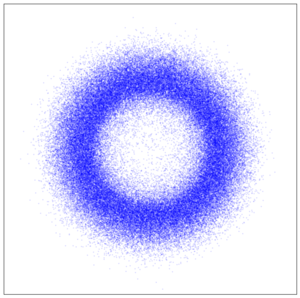}& \includegraphics[width=0.3\textwidth]{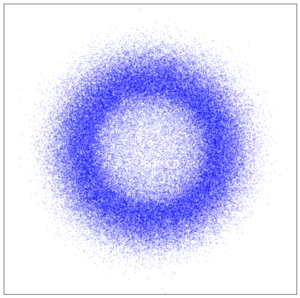}&\includegraphics[width=0.3\textwidth]{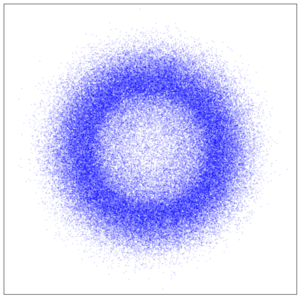} \\
		\end{tabular}
		\caption{Unit disc}
	\end{subfigure}	
	\caption{Comparing multi-winner rules on 2D-Euclidean profiles.}
	\label{tabfig:expe-multiwinner}
\end{figure}

We considered two models: sampling the voter and candidate positions from the unit square, or from the unit disc. For each model, we sample 10\,000 profiles. We convert them to weak order profiles using the radius method ($r = 0.1, 0.2, 0.3$), and compute the winning committees of the three rules. \Cref{tabfig:expe-multiwinner} shows the positions of the $10$ winning candidates, superimposing the results from all profiles. Intuitively, we expect a rule to be more proportional if the distribution of winners is closer to the distribution of voters (here, the uniform distribution over the unit square or disc). 

Our first observation is that the more indifference we introduce (i.e., a larger radius $r$), the further is the distribution of the positions of winning candidates from the distribution of voter positions. For $r=0.3$, winning candidates are concentrated roughly midway between the edges and the center of the distribution. A second observation is that for $r=0.1$, Split-STV appears most faithful to the voter positions, and EAR the least faithful; but when $r=0.3$, the order is reversed, as EAR is the most faithful and Split-STV the least faithful. In both cases, Approval-STV is between the two.

\section{Conclusion and Future Work}

We have studied generalizations of IRV and STV to weak orders and have given formal arguments why the ``Approval'' generalization behaves better than the ``Split'' generalization. However, there are additional avenues for comparing these two rules. For example, we have not looked into strategic aspects, and the two rules may well differ in how often voters can obtain a better outcome by misrepresenting their preferences. Another direction would be to study the utilitarian welfare provided by these rules, for example in the metric distortion model, where the worst-case performance of linear-order IRV has been studied \citep{anshelevich2018approximating,anagnostides2022dimensionality}, though the model would need to be generalized to weak orders.

Assuming there are no ties (see \Cref{footnote:put}), computing Approval-IRV can clearly be done in polynomial time, though just like for linear-order IRV the task cannot be efficiently parallelized \citep{mapreduce}, because counting needs to proceed in rounds. Since some jurisdictions (notably Australia) count IRV elections by hand using paper ballots, it makes sense to optimize the computation in terms of the number of times each ballot needs to be handled, which \citet{ayadi2019single} analyze using a query complexity model. It may be worth performing such an analysis for Approval-IRV. It seems that this rule can still be easily counted by hand, as follows: For each candidate, make a stack of ballot papers that vote for that candidate. If a ballot votes for $t \ge 2$ candidates, assign the ballot to any one of the $t$ stacks, and additionally take $t-1$ \emph{tokens} (e.g., specially colored pieces of paper) and add one to each of the other $t-1$ stacks. Determine the canddidate $c$ with the smallest stack and eliminate it, by throwing away all the tokens in the stack and reassigning the ballot papers in the stack. To reassign, if a ballot paper is indifferent between $c$ and some other non-eliminated candidate $d$, then add the ballot paper to the stack of $d$ in exchange for one of the tokens in that stack. Otherwise, we need to pass to the next-highest indifference class of the ballot, assign the ballot paper to one of the candidates in that indifference class, and add tokens to other stacks as appropriate. Note that Split-IRV does not admit a similar protocol, since we need to update the scores of all top-ranked candidates after each elimination.

As to whether it would be desirable to move from linear-order IRV to Approval-IRV, future work could address this question from several angles. User studies could explore whether voters understand how weak order ballots work, and whether they understand Approval-IRV. Generally, approval-based voting methods suffer from an intuitive (mis)impression that some voters get ``more votes'' than others by ranking several candidates first, and it would be interesting whether this objection can be answered in an intuitively persuasive way. Models of political ideology (such as strategic candidacy models) could be used to understand the impact on the political landscape of a move to Approval-IRV. Finally, it would be interesting to provide a theoretical foundation for our experimental finding that Approval-IRV selects higher-quality candidates (with respect to Borda scores or distortion) than linear-order IRV.

\subsection*{Acknowledgments}
We thank Jérôme Lang and Jannik Peters for useful feedback, and Chris Dong, Klaus Nehring, and Bill Zwicker for helpful discussions. This work was funded in part by the French government under management of Agence Nationale de la Recherche as part of the ``Investissements d’avenir'' program, reference ANR-19-P3IA-0001 (PRAIRIE 3IA Institute).

\bibliographystyle{ACM-Reference-Format}

\newpage
\appendix 

\addtocontents{toc}{\protect\setcounter{tocdepth}{1}}

\section{Delayed Proofs}
\label{app:proofs}

\subsection{Proof of \Cref{thm:AVIRVclones}}
\label{app:AVIRVclones-proof}

\AVIRVclones*

\begin{proof}
We show that Approval-IRV satisfies independence of clones. Write $f$ for Approval-IRV.

For a profile $P$, we will use the following shorthand notation: For $\ell \in C$, we write $P - \ell$ for the profile $P_{C \setminus \{\ell\}}$ with $\ell$ deleted. For a set $X \subseteq C$, we write $P - X$ for the profile $P|_{C\setminus X}$ with the alternatives in $X$ deleted. Similarly, for $x \in X$, we write $P - X + x$ for the profile $P|_{(C \setminus X) \cup \{x\}}$ with all alternatives in $X$ except for $x$ deleted.

The following lemma connects the scores of the alternatives in the profile $P$ and in the profile $P - X + x$, where $X$ is a clone set. It is the key to the proof working, and other rules like Split-IRV do not have the same property, explaining why the proof does not work for them.

\begin{lemma} \label{lem:clone-scores}
	Let $P$ be a profile defined on alternative set $C$ with clone set $X\subseteq C$. Let $x \in X$. Then
	\begin{itemize}
		\item every $c \in C \setminus X$ has the same approval score in $P$ and $P - X + x$, and
		\item the approval score of $x$ in $P -X +x$ is at least as high as the approval score of every clone alternative $x' \in X$ in $P$.
	\end{itemize}
\end{lemma}
\begin{proof}
	
	For the first point, observe that for $c \in C \setminus X$, $c$ is ranked in the top indifference class of a voter $i$ in $P$ iff $c \pref_i d$ for all $d \in C$ iff (by definition of clone set) $c \pref_i d$ for all $d \in (C \setminus X) \cup \{x\}$ iff it is ranked in the top indifference class of $i$ in $P - X + x$.
	
	For the second point, fix $x \in X$ and let $x' \in X$. Then if $x'$ is ranked in the top indifference class of a voter $i$ in $P$, then $x' \pref_i d$ for all $d \in C$ and in particular for all $d \in C \setminus X$. Thus, by definition of clone set, we also have $x \pref_i d$ for all $d \in C \setminus X$, and hence $x$ is ranked in the top indifference class of $i$ in $P - X + x$. So the number of voters with $x'$ in their top indifference class in $P$ is weakly lower than the number with $x$ in their top indifference class in $P - X + x$.
\end{proof}

By induction on $m$, we prove the following statement:

\begin{quote}
	For every profile $P$ defined on a set $C$ of $m$ alternative including a non-empty clone set $X \subseteq C$, the following hold: [$H$ stands for hypothesis]
	\begin{enumerate}
		\item[\quad$H_1(P,X)$:] for all $c \in C \setminus X$, $c \in f(P)$ if and only if $c \in f(P - X + x)$ for all $x \in X$.
		\item[\quad$H_2(P,X)$:] we have $f(P) \cap X \neq \emptyset$ if and only if $x \in f(P - X + x)$ for all $x \in X$.
	\end{enumerate}
\end{quote}

Note that the statement is trivially true if $|X| = 1$ since then $P - X + x = P$. The statement is also obvious when $|X| = |C|$ since then $P - X + x$ is a profile in which only 1 alternative exists. Now, the base case $m = 2$ is easy to see, since then either $|X| = 1$ or $|X| = 2 = |C|$.

So let $m \ge 3$, assume we have shown the statement for $m - 1$, and let $P$ be a profile with $m$ alternatives $C$ including clone set $X \subseteq C$ with $2 \le |X| \le m - 1$.

Let us first note a simple fact that follows because Approval-IRV is a neutral rule (invariant under renaming alternatives). For every non-clone alternative $c \in C \setminus X$, we have
\[
c \in f(P - X + x) \text{ for some $x \in X$} \iff c \in f(P - X + x) \text{ for all $x \in X$},
\]
and we have that
\[
x \in f(P - X + x) \text{ for some $x \in X$} \iff x \in f(P - X + x) \text{ for all $x \in X$}.
\]
This just follows because for two clones $x, x' \in X$, by definition of clone sets, the profiles $P - X + x$ and $P - X + x'$ are identical up to the permutation that exchanges $x$ and $x'$. These equivalences mean that we can use the inductive hypotheses in the ``all $x$'' version but only need to prove them in the ``some $x$'' version.

We first prove $H_1(P, X)$. Let $c \in f(P	) \setminus X$ be a non-clone alternative that wins in $P$. We need to show that $c \in f(P - X + x)$ for some $x \in X$.
Note that by definition of elimination scoring rules, $c \in f(P)$ means that there is an alternative $\ell$ with lowest score in $P$ such that $c \in f(P - \ell)$.
\begin{itemize}
	\item Consider first the case that $\ell$ is not a clone alternative, $\ell \not \in X$. Take any $x \in X$. By \Cref{lem:clone-scores}, $\ell$ is also a lowest-scoring alternative in $P - X + x$. Thus by definition of elimination scoring rules, $f((P - X + x) - \ell) \subseteq f(P - X + x)$, and hence it suffices to show that $c \in f((P - X + x)- \ell) = f((P - \ell) - X + x)$. But this follows from $H_1(P-\ell, X)$ because $c \in f(P - \ell)$.
	\item Consider next the case that $\ell \in X$, and take any $x \in X \setminus \{\ell\}$, which exists since $|X| \ge 2$. Then applying $H_1(P - \ell, X \setminus \{\ell\})$ to $c \in f(P - \ell)$, we get that $c \in f((P - \ell) - (X \setminus \{\ell\}) + x) = f(P - X + x)$ where the last equality follows because the two profiles are the same since $\ell \in X$.
\end{itemize}

Conversely, suppose that $c \in f(P - X + x)$ for all $x \in X$. 
\begin{itemize}
	\item Suppose that there exists a clone alternative $x' \in X$ which is a lowest-scoring alternative in $P$. Noting that $|X| \ge 2$, choose any other clone $x \in X \setminus \{x'\}$ and note that $c \in f(P - X + x)$ by assumption. Since $x'$ is lowest-scoring in $P$, by definition of elimination scoring rules, $f(P - x') \subseteq f(P)$. By $H_1(P - x', X \setminus \{x'\})$, it follows from $c \in f(P - X + x) = f((P - x') - (X \setminus \{x'\}) + x)$ that $c \in f(P - x')$ and hence $c \in f(P)$.
	\item Otherwise, only non-clone alternatives are lowest-scoring in $P$. Then by \Cref{lem:clone-scores}, the same is true in $f(P - X + x)$. Since $c \in f(P - X + x)$ and $|X| \le m - 1$, there must be a lowest-scoring alternative $\ell \not \in X$ such that $c \in f((P - X + x) - \ell) = f((P - \ell) - X + x)$. By $H_1(P - \ell, X)$, it follows that $c \in f(P - \ell)$. Because $\ell$ must also be lowest-scoring in $P$ (due to \Cref{lem:clone-scores}), we have $f(P - \ell) \subseteq f(P)$, and hence $c \in f(P)$.
\end{itemize}

We next prove $H_2(P, X)$, using analogous reasoning. Suppose that $f(P) \cap X \neq \emptyset$. Let $x \in f(P) \cap X$ be a winning clone alternative. By definition of elimination scoring rules, $x \in f(P)$ means that there is an alternative $\ell$ with lowest score in $P$ such that $x\in f(P - \ell)$. 
We show that $x \in f(P - X + x)$.
\begin{itemize}
	\item Consider first the case that $\ell$ is not a clone alternative, $\ell \not \in X$. 
	By \Cref{lem:clone-scores}, $\ell$ is also a lowest-scoring alternative in $P - X + x$. Thus by definition of elimination scoring rules, $f((P - X + x) - \ell) \subseteq f(P - X + x)$, and hence it suffices to show that $x \in f((P - X + x)- \ell) = f((P - \ell) - X + x)$. But this follows from $H_2(P-\ell, X)$ because $x \in f(P - \ell)$.
	\item Consider next the case that $\ell \in X$. Clearly $\ell \neq x$ since $x \in f(P - \ell)$.
	By $H_2(P - \ell, X \setminus \{\ell\})$, since $x \in f(P - \ell)$, we get that $x \in f((P - \ell) - (X \setminus \{\ell\}) + x) = f(P - X + x)$ where the last equality follows because the two profiles are the same since $\ell \in X$.
\end{itemize}

Conversely, suppose that $x \in f(P - X + x)$ for all $x \in X$. We need to show that $f(P) \cap X \neq \emptyset$.
\begin{itemize}
	\item Suppose that there exists a clone alternative $x' \in X$ which is a lowest-scoring alternative in $P$. 
	Noting that $|X| \ge 2$, choose any other clone $x \in X \setminus \{x'\}$ and note that $x \in f(P - X + x)$ by assumption. 
	Since $x'$ is lowest-scoring in $P$, by definition of elimination scoring rules, $f(P - x') \subseteq f(P)$. By $H_2(P - x', X \setminus \{x'\})$, it follows from $x \in f(P - X + x) = f((P - x') - (X \setminus \{x'\}) + x)$ that $f(P - x') \cap (X \setminus \{x'\}) \neq \emptyset$ and hence also $f(P) \cap X \neq \emptyset$.
	\item Otherwise, only non-clone alternatives are lowest-scoring in $P$. Then by \Cref{lem:clone-scores}, the same is true in $f(P - X + x)$. Since $x \in f(P - X + x)$ and $|X| \le m - 1$, there must be a lowest-scoring alternative $\ell \not \in X$ such that $x \in f((P - X + x) - \ell) = f((P - \ell) - X + x)$. By $H_2(P - \ell, X)$, it follows that $f(P - \ell) \cap X \neq \emptyset$. Because $\ell$ must also be lowest-scoring in $P$ (due to \Cref{lem:clone-scores}), we have $f(P - \ell) \subseteq f(P)$, and hence also $f(P) \cap X \neq \emptyset$. \qedhere
\end{itemize}
\end{proof}

\subsection{Proof of \Cref{thm:clonemajority}}
\label{app:clonemajority-proof}

\clonemajority*
\begin{proof}
	
	Let $f$ be an elimination scoring rule satisfying independence of clones and respect for cohesive majorities. We want to show that for all order types $\tau$, their scoring vector is $s(\tau) = (1,0,\dots,0)$. We first prove it for linear order types $\tau = (1,1,\dots,1,1)$, then for all other order types. Because there is only one linear order type of length $m$, we denote $s(m) = s(\tau)$ for $|\tau| = m$. We know that $s(m)_1 > 0$ otherwise $s(m) = (0,0,\dots,0,0)$ and $a$ is not the only winner in the profile with one ranking where $a$ is first $\{a\} \succ \dots$, contradicting respect for cohesive majorities. Thus, we can assume without loss of generality that for all $m$, $s(m)_1 = 1$ (as there is only one linear order type for each number of candidates).  In \Cref{lem:linearclones}, we prove that $s(\tau) = (1,0,\dots,0)$ for all linear orders $\tau$.
	
	\begin{lemma} \label{lem:linearclones}
		Let $f$ be an elimination scoring rule satisfying independence of clones and respect for cohesive majorities. Then any linear order type $\tau$ is associated to the scoring vector $s(\tau) = (1,0,\dots,0)$.
	\end{lemma}
	\begin{proof}
		We prove it by induction on the number of candidates $m$. For $m=2$ it is true because the only order type $\tau = (1,1)$ has scoring vector $(1,0)$. For $m=3$, we will use two order types: $\tau_{111} = (1,1,1)$ and $\tau_{21} = (2,1)$. Denote $x$ and $y$ such that $s(\tau_{111}) = (1,x,0)$ and $s(\tau_{21}) = (y,0)$ with $x \in [0,1]$ and $y \ge 0$. We want to prove that $x=0$.
		
		We will prove (1) $y \le 1$, (2) if $y < 1$ then $x=0$ and (3) if $y=1$ then $x=0$. This will prove that $x=0$. Note that in this lemma we will not determine
		the exact value of $y$.
		
		Let us prove that $y \le 1$. Assume for contradiction that $y > 1$. Let $q \in \mathbb N$ with $q > \frac{1}{y-1}$ and consider the profile $P$ with $q+1$ orders $\{a\} \succ \{b\} \succ \{c\}$ and $q$ orders $\{b,c\} \succ \{a\}$. By respect for cohesive majorities, $a$ should be the winner. But the scores are $S(a) = q+1$ and $S(b) \ge S(c) = qy$. Because $q > \frac{1}{y-1}$, $a$ is eliminated first, a contradiction.
		
		Let us prove that if $y < 1$, then $x=0$. Assume for contradiction that $y < 1$ and $x > 0$. Let $q \in \mathbb N$ with $q > \frac{1}{x}$ and $q > \frac{1}{1-y}$ and consider the profile $P$:
		\begin{align*}
			q &: \{c\} \succ \{b\} \succ \{a\}  \\
			q &: \{a,b\} \succ \{c\} \\
			1 &: \{a\} \succ \{c\} \succ \{b\}
		\end{align*}
		In this profile, $a$ is ranked top by more than half of the voters, so $f(P) \subseteq \{a,b\}$ by respect for cohesive majorities. However, the scores are $S(a) = qy + 1$, $S(b) = qy + qx$ and $S(c) = q+x$. Because $q > \frac{1}{x}$ we have $S(b) > S(a)$ and because $q > \frac{1}{1-y}$, $S(c) > S(a)$. Thus, $a$ is eliminated first. The scores are now $S(c) = q+1$ and $S(b) = q$, so $c$ wins, a contradiction. Thus, if $y < 1$, then $x=0$.
		
		Let us prove that if $y = 1$, then $x=0$ again. Assume for contradiction that $y = 1$ but $x > 0$. Let $q \in \mathbb N$ with $q > \frac{1}{x}$ and $q > 2$ and consider the profile $P$:
		\begin{align*}
			q+1 &: \{b,c\} \succ \{a\} \\
			q &: \{a\} \succ \{c\} \succ \{b\}  \\
			1 &: \{a,b\} \succ \{c\} \\
			1 &: \{a\} \succ \{b\} \succ \{c\}
		\end{align*}
		In this profile, $a$ is ranked top by more than half of the voters, so $f(P) \subseteq \{a,b\}$ by respect for cohesive majorities. However, the scores are $S(a) = q+2$, $S(b) = q+2 + x > S(a)$ and $S(c) = q+1 + qx$. Since $q > \frac{1}{x}$, we have $S(c) > S(a)$ and $a$ is eliminated first. Then the scores are $S(c) = q$ and $S(b) = 2$, so $c$ wins, which is a contradiction. 
		
		Therefore, we know that $x=0$ and $s(\tau_{111}) = (1,0,0)$. Let us now prove it for $m = 4$ (i.e. $\tau = (1,1,1,1)$). Assume for contradiction that $s(4)_2 > 0$. Let $q \in \mathbb N$ with $q > \frac{1}{s(4)_2}$ and consider the following profile:
		\begin{align*}
			q+1 &: \{b\} \succ \{c\} \succ \{a\} \succ \{a'\} \\
			q &: \{c\} \succ \{a\} \succ \{a'\} \succ \{b\} \\
			q &: \{a\} \succ \{a'\} \succ \{c\} \succ \{b\} \\
			q &: \{a'\} \succ \{a\} \succ \{c\} \succ \{b\}
		\end{align*}
		In this profile, the scores are $S(c) \ge q + (q+1)s(4)_2$, $S(b) = q+1$ and $S(a) \ge S(a') \ge q + qs(4)_2$. Because $q > \frac{1}{s(4)_2}$, we have $S(c) > S(b)$ and $S(a') > S(b)$, so $b$ is eliminated first. The scoring vectors are now $(1,0,0)$, so the scores are $S(c) = 2q+1$ and $S(a) = S(a') = q$ so one of $a$ or $a'$ is eliminated, and $c$ wins the majority vote in the final round. Note that in this profile $a$ and $a'$ are clones, so $c$ should also win in the profile in which we remove the clone $a'$ of $a$. However, in this profile, the scores are $S(b) = q+1$, $S(c) = q$ and $S(a) = 2q$, so $c$ is eliminated first, a contradiction. This proves that $s(4)_2 =0$ and thus $s(4) = (1,0,0,0)$.
		
		We know prove by induction that for all $m$, $s(m) = (1,0,\dots,0)$. Assume it is true for $m \ge 4$, let us prove it for $m+1$. Assume for contradiction that $s(m+1)_2 > 0$. Let $q \in \mathbb N$ with $q > \frac{1}{s(m+1)_2}$ and consider the profile $P$ on $C = \{a,a',b,c_1, \dots,c_{m-2}\}$:
		\begin{itemize}
			\item For each $i \in [1,m-2]$, $q$ linear orders $\{c_i\} \succ \{b\} \succ \{a\} \succ \{a'\} \succ \{c_{i+1}\}  \succ \dots \succ \{c_{i+m-3}\}$ where the subscripts of $c_i$ have to be taken modulo $m-2$ (for instance, the linear order starting with $c_2$ ends with $c_1$).
			\item $q-1$ linear orders $\{b\} \succ \{a\} \succ \{a'\} \succ \{c_1\} \succ \dots \succ \{c_{m-2}\}$.
			\item $q$ linear orders $\{a\} \succ \{a'\} \succ \{b\} \succ \{c_1\} \succ \dots \succ \{c_{m-2}\}$.
			\item $q$ linear orders $\{a'\} \succ \{a\} \succ \{b\} \succ \{c_1\} \succ \dots \succ \{c_{m-2}\}$.
		\end{itemize}
		Note that in this profile, each $c_i$ is ranked last in at least $q$ orders. The scores are $S(a) \ge S(a') = q + qs(m+1)_2 + qs(m+1)_3 + q(m-2)s(m+1)_4$, $S(b) = (q-1) + q(m-2)s(m+1)_2 + 2qs(m+1)_3$ and $S(c_i) \le q + (3q-1)s(m+1)_4 + q(m-4)s(m+1)_5$.
		
		We know that $s(m+1)_2 \ge s(m+1)_3 \ge s(m+1)_4 \ge s(m+1)_5$, so we can make the following observation for all candidates $c_i$:
		\begin{align*}
			S(b)
			&\ge(q-1) + q(m-2)s(m+1)_2 + 2qs(m+1)_3 \\
			&= q-1 + qs(m+1)_2+ q(m-3)s(m+1)_2 +2qs(m+1)_3 \\
			&\ge q-1 + qs(m+1)_2 + q(m-3)s(m+1)_3 +2qs(m+1)_3 \\
			&= q-1 +  qs(m+1)_2 + q(m-1)s(m+1)_3 \\
			&= q-1 +  qs(m+1)_2 + q(m-4)s(m+1)_3 + 3qs(m+1)_3 \\
			&\ge  q-1 +  qs(m+1)_2 + q(m-4)s(m+1)_5 + 3q s(m+1)_4\\
			&\ge  q-1 +  qs(m+1)_2 + q(m-4)s(m+1)_5 + (3q-1) s(m+1)_4 \\
			&> q + q(m-4)s(m+1)_5 + (3q-1)s(m+1)_4 \tag{since $q > \frac{1}{s(m+1)_2}$} \\
			&\ge S(c_i).
		\end{align*}
		
		Moreover, we have:
		
		\begin{align*}
			S(a) \ge S(a')
			&\ge q + q(m-2)s(m+1)_4 + (q-1)s(m+1)_3 + qs(m+1)_2 \\
			&= q + q(m-4)s(m+1)_4 + 2qs(m+1)_4 + (q-1)s(m+1)_3 + qs(m+1)_2 \\
			&\ge q + q(m-4)s(m+1)_4 + (3q-1)s(m+1)_4 + qs(m+1)_2 \\
			&> q + q(m-4)s(m+1)_5 + (3q-1)s(m+1)_4 \tag{since $s(m+1)_2 > 0$} \\
			&\ge S(c_i).
		\end{align*}
		
		Therefore, one of the $c_i$ is eliminated first. By induction hypothesis, all scoring vectors are now $(1,0,\dots,0)$. The scores are now $S(a) = S(a') = q$, $S(b) = 2q-1$ and $S(c_i) = q$ for all $c_i$ that are not eliminated. Moreover, the score of any candidate other than $b$ is upper bounded by $q$ as long as $b$ is not eliminated, except for $a$ (or $a'$), which obtains a score of $2q$ once its clone $a'$ (or $a$) is eliminated. Therefore, the $c_i$ are successively eliminated until $a$ (or $a'$) and $b$ remain. The scores are $S(a) = 2q$ and $S(b) = q-1 + q(m-2)$. Because $m \ge 4$, $S(b) > S(a)$ and $b$ wins.
		
		Note that in this profile, $a$ and $a'$ are clones. Therefore, by independence of clones, $b$ should also win in the profile in which we remove the clone $a'$ of $a$. However, in this profile, all scoring vectors are $(1,0,\dots,0)$ by our induction hypothesis. Thus, the scores are $S(a) = 2q$, $S(c_i) = q$ for all $c_i$ and $S(b) = q-1$. Thus, $b$ is eliminated first, a contradiction.

		This proves that $s(m+1)_2 =0$. Because we know that $s(m+1)_i \le s(m+1)_2$ for all $i \ge 3$, we have $s(m+1) = (1,0,\dots,0)$ and the induction hypothesis concludes the proof.
	\end{proof}

	We now prove the result for all other order types. For this, we proceed in several steps. We first show that some specific order types with length $|\tau| \le 3$ or defined for $m \le 5$ candidates have approval scores as their scoring vector. Then, we proceed by induction to prove the result for all other order types.

	In this proof, if a ballot is indifferent between all candidates, i.e., its order type is of length $|\tau| = 1$, we say that $s(\tau) = (0)$ and no candidates get points. Moreover, in every profile, when we use some $q\in \mathbb N$, we assume that $q$ is large enough (e.g., $q > 100$), such that what matters most is the coefficient of $q$ (for instance, we will consider that $2q > q + 5$).

	\textbf{Step 1.} We first focus on order types $\tau_{21}$ and $\tau_{31}$ and show that (1) $s(\tau_{21})_1 \le 1$, (2) $s(\tau_{21})_1 \in [0, \frac12] \cup \{1\}$, (3) $s(\tau_{21})_1 = s(\tau_{31})_1$ and  (4) $s(\tau_{21})_1 \ge \frac{5}{9}$. Combining these gives $s(\tau_{21}) = s(\tau_{31}) = (1,0)$.
	
	\textbf{Step 1.1.} We prove that $s(\tau_{21})_1 \le 1$ using respect for cohesive majorities. Assume for a contradiction that $s(\tau_{21})_1 > 1$.  Let $q \in \mathbb N$ be such that $q > \frac{1}{s(\tau_{21})_1-1}$ and consider the profile $P$ with $q$ orders $\{c,b\} \succ \{a\}$ and $q+1$ linear orders $\{a\} \succ \{c\} \succ \{b\}$. Here, $a$ is in the top indifference class of more than half voters, but the scores are $S(a) = q+1$ and $S(c) = S(b) = qs(\tau_{21})_1$. Since $q > \frac{1}{s(\tau_{21})_1-1}$, $S(c) = S(b) > S(a)$ and $a$ is eliminated first, which contradicts respect for cohesive majorities.

	\textbf{Step 1.2.} We prove that $s(\tau_{21})_1 \in [0, \frac12] \cup \{1\}$ using respect for cohesive majorities. Assume for a contradiction that $\frac12 < s(\tau_{21})_1 < 1$ and let $q \in \mathbb N$ be such that $q > \max(\frac{1}{1-s(\tau_{21})_1}, \frac{1}{2s(\tau_{21})_1-1})$. Consider the following profile $P$:
	\begin{align*}
		q+1 &: \{a\} \succ \{c\} \succ \{b\} \\
		3q  &: \{a,b\} \succ \{c\} \\
		2q  &: \{b,c\} \succ \{a\} \\
		2q  &: \{c\} \succ \{b\} \succ \{a\} 
	\end{align*}
	In this profile, $a$ is in the top indifference class of more than half of the voters. Respect for cohesive majorities imposes that $f(P) \subseteq \{a,b\}$. The scores are $S(a) = q+1 + 3qs(\tau_{21})_1$, $S(b) = 5qs(\tau_{21})_1$ and $S(c) = 2q+ 2qs(\tau_{21})_1$. Since $q > \frac{1}{1-s(\tau_{21})_1}$, we have $S(a) < S(c)$ and since $q > \frac{1}{2s(\tau_{21})_1-1}$ we have $S(a) < S(b)$. Thus, $a$ is eliminated first and $c$ wins the majority vote against $b$, which contradicts respect for cohesive majorities.
	
	\textbf{Step 1.3.} We prove that  $s(\tau_{21})_1 = s(\tau_{31})_1$ using independence of clones. First assume that $s(\tau_{21})_1 < s(\tau_{31})_1$. Let $q,q' \in \mathbb N$ such that $qs(\tau_{21})_1 <q' <qs(\tau_{31})_1$. Consider the following profile $P$:
	\begin{align*}
		q    &: \{c,c',a\} \succ \{b\} \\
		2q'  &: \{b\} \succ \{a\} \succ \{c'\} \succ \{c\} \\
		q'+2 &: \{c\} \succ \{c'\} \succ \{b\} \succ \{a\} \\
		q'+1 &: \{c'\} \succ \{c\} \succ \{b\} \succ \{a\} \\
		q'   &: \{a\} \succ \{b\} \succ \{c\} \succ \{c'\}
	\end{align*}
	In this profile, the scores are $S(c) > S(c') > S(a) = q'+qs(\tau_{31})_1 > 2q'$ and $S(b) = 2q'$. Thus, $b$ is eliminated first. The new scores are $S(c) > S(c') = q'+1$ and $S(a) = 3q'$. Thus, $c'$ and $c$ are eliminated, and $a$ wins. Now, observe that $c$ and $c'$ are clones in $P$. Therefore, independence of clones imposes that $a$ should also be winning in the profile $P'$ in which we remove $c'$. In $P'$, the scores are $S(c) > S(a) = qs(\tau_{21})_1 + q' < 2q'$ and $S(b) = 2q'$. Therefore, $a$ is eliminated first, which contradicts independence of clones.  
	
	We now assume that $s(\tau_{21})_1 > s(\tau_{31})_1$. Let $q,q' \in \mathbb N$ such that $qs(\tau_{21})_1 >q' > qs(\tau_{31})_1$. Using the same profiles $P$ and $P'$ as above, we have that $a$ is eliminated first in $P$, but in $P'$, $b$ is eliminated first, and $a$ is the winner. This contradicts again independence of clones. We conclude that $s(\tau_{21})_1 = s(\tau_{31})_1$.
	
	\textbf{Step 1.4.} We prove that $s(\tau_{21})_1 \ge \frac{5}{9}$. Assume for a contradiction that $s(\tau_{21})_1 < \frac{5}{9}$ and consider the following profile $P$.
	\begin{align*}
		4  &: \{a,b\}\succ \{d\} \succ\{c\} \\
		5  &: \{a,c\} \succ \{d\} \succ \{b\} \\
		13 &: \{a,b,c\} \succ \{d\} \\
		5  &: \{b,d\} \succ \{c\} \succ \{a\} \\
		4  &: \{c,d\} \succ \{b\}\succ \{a\} \\
		1  &: \{b\} \succ \{d\} \succ \{c\} \succ \{a\} \\
		1  &: \{c\} \succ \{d\} \succ \{b\} \succ \{a\} \\
		10 &: \{d\} \succ \{c\} \succ \{b\} \succ \{a\}
	\end{align*}
	In this profile, $a$ is on the top indifference class of more than half of the voters, therefore by respect for cohesive majorities we have $f(P) \subseteq \{a,b,c\}$. The scores are $S(a) = 9s(\tau_{211})_1 + 13s(\tau_{31})_1$, $S(c) \ge S(b) = 9s(\tau_{211})_1 + 13s(\tau_{31})_1 + 1 > S(a)$ and $S(d) \ge 10 + 9s(\tau_{211})_1$. Note that we are only lower bounding as we do not know if $(\tau_{211})_2 = 0$. Since $s(\tau_{31})_1 = s(\tau_{21})_1 < \frac{5}{9} < \frac{10}{13}$, we have that $S(a) < S(d)$. Therefore, $a$ is eliminated first. The new scores are $S(b) = 18s(\tau_{21})_1 + 5$, $S(c) =  17s(\tau_{21})_1 + 6$ and $S(d) = 10 + 9s(\tau_{21})_1$, since $s(\tau_{21})_1 < \frac{5}{9}$, we have $S(b) < S(d)$, and $S(b) < S(c)$. So $b$ is eliminated. Finally, $d$ wins the majority vote against $c$. This contradicts respect for cohesive majorities.
	
	Thus, since $s(\tau_{21})_1 \in [0, \frac12] \cup \{1\}$ and $s(\tau_{21})_1 \ge \frac{5}{9} > \frac{1}{2}$, we can conclude that $s(\tau_{21}) = (1,0)$. Moreover, $s(\tau_{31}) = (1,0)$.

	\textbf{Step 2.} We now prove that $s(\tau_{211})_1 = 1$.
	
	\textbf{Step 2.1.} Let us prove that $s(\tau_{211})_1 \le 1$. Assume that $s(\tau_{211})_1 > 1$ and take $q \in \mathbb N$ such that $q > \frac{1}{s(\tau_{211})_1 - 1}$. Consider the following profile $P$:
	\begin{align*}
		q+4  &: \{a\} \succ \{d\}\succ \{c\} \succ \{b\} \\
		3q   &: \{a,b,c\} \succ \{d\}\\
		2q+1 &: \{b,d\} \succ \{c\}\succ\{a\} \\
		2q-1 &: \{c,d\} \succ \{b\}\succ\{a\}
	\end{align*}
	In this profile, $a$ is in the top indifference class of more than half of the votes, so by respect for cohesive majorities we should have $f(P) \subseteq \{a,b,c\}$. The scores are $S(a) = 4q+4$, $S(b) > S(c) \ge 3q + (2q-1)s(\tau_{211})_1 > 5q -1 >  S(a)$ and $S(d) = 4qs(\tau_{211})_1$. Note that we are only lower bounding $S(c)$ and $S(b)$ as we do not know if $s(\tau_{211})_2 = 0$. Since $q > \frac{1}{s(\tau_{211})_1 - 1}$, $S(d) > S(a)$. Thus, $a$ is eliminated first. The new scores are $S(b) > S(c) =  5q-1$ and $S(d) =5q+4$. Thus, $c$ is eliminated next, and $d$ wins the majority vote against $b$, which contradicts respect for cohesive majorities.
	
	\textbf{Step 2.2.} Let us prove that $s(\tau_{211})_1 \ge 1$ with independence of clones. Assume that $s(\tau_{211})_1 < 1$. Take $q \in \mathbb N$ such that $q > \frac{1}{1 - s(\tau_{211})_1}$ and consider the profile $P$:
	\begin{align*}
		q   &: \{a,b\} \succ \{c\} \succ \{c'\} \\
		3   &: \{a\} \succ \{c\} \succ \{c'\} \succ \{b\} \\
		4   &: \{b\} \succ \{c\} \succ \{c'\} \succ \{a\} \\
		q+2 &: \{c,c',a\} \succ \{b\} \\
		q+2 &: \{c,c',b\} \succ \{a\}
	\end{align*}
	In this profile, $S(b) > S(a) = qs(\tau_{211})_1 + q + 5$ and $S(c) \ge S(c') \ge 2q+4$. Since $q > \frac{1}{1 - s(\tau_{211})_1}$, we have $S(c) > S(a)$. Thus, $a$ is eliminated first. The new scores are $S(b) = q+4$, $S(c) =  q+5$ and $S(c') = q+2$ so $c'$ is eliminated next, then $c$ wins the majority vote against $b$. Now observe that $c$ and $c'$ are clones in $P$, so by independence of clones, $c$ should also win in the profile $P'$ without $c'$. However, in $P'$, the scores are $S(b) > S(a) = 2q+5$ and $S(c) = 2q+4$, so $c$ is eliminated first. This contradicts independence of clones, and proves that $s(\tau_{211})_1 = 1$.
	
	\textbf{Step 3.} We now prove that $s(\tau_{12}) = s(\tau_{22}) = s(\tau_{13}) = (1,0)$.
	
	\textbf{Step 3.1.} We first show that $s(\tau_{12})_1 = 1$. First, we show that $s(\tau_{12})_1 \ge 1$ using respect for cohesive majorities. Take $q \in \mathbb N$ such that $q > \frac{1}{1-s(\tau_{12})_1}$ and consider the profile $P$ with $q$ orders $\{a\} \succ \{b,c\}$ and $q-1$ orders $\{b,c\} \succ \{a\}$. In this profile, $a$ is ranked first in more than half of the votes, but the scores are $S(a) = qs(\tau_{12})_1$ and $S(b) = S(c) = q$. By hypothesis on $q$, $S(b) = S(c) > S(a)$ and $a$ is eliminated first, a contradiction.

	Assume now that $s(\tau_{12})_1 > 1$. Take $q \in \mathbb N$ such that $ q > \frac{1}{s(\tau_{12})_1-1}$ and consider the following profile $P$:
	\begin{align*}
		q+2 &: \{a,b\} \succ \{c\} \\
		q   &: \{c\} \succ \{a,b\} \\
		4   &: \{b,c\} \succ \{a\} \\
		3   &: \{a\} \succ \{c\} \succ \{b\}
	\end{align*}
	In this profile, $a$ is in the top indifference class of more than half of the votes, thus by respect for cohesive majorities $f(P) \subseteq \{a,b\}$. The scores are $S(a) = q+5$, $S(b) = q+6$ and $S(c) = 4 + qs(\tau_{12})_1$. Since $q > \frac{1}{s(\tau_{12})_1-1}$, $S(c) > S(a)$, thus $a$ is eliminated first, and $c$ wins the majority vote against $b$. This contradicts respect for cohesive majorities. Therefore, we necessarily have $s(\tau_{12}) = (1,0)$.
	
	\textbf{Step 3.2.} We now show that $s(\tau_{22})_1 = 1$. Assume that $s(\tau_{22})_1 < 1$. Take $q \in \mathbb N$ such that $q > \frac{1}{1-s(\tau_{22})_1}$ and consider the following profile $P$:
	\begin{align*}
		q+2 &: \{b,b'\} \succ \{c\} \succ \{a\} \\
		q   &: \{a,c\} \succ \{b,b'\} \\
		3   &: \{a\} \succ \{b\} \succ \{b'\} \succ \{c\} \\
		4   &: \{c\} \succ \{b\} \succ \{b'\} \succ \{a\}
	\end{align*}	
	In this profile, $S(c) > S(a) = qs(\tau_{22})_1+3$ and $S(b) = S(b') = q+2$. Since $q > \frac{1}{1-s(\tau_{22})_1}$, $S(b) > S(a)$ and $a$ is eliminated first, then $b'$ is eliminated, and $b$ wins the majority vote against $c$. Since $b$ and $b'$ are clones, this means that in the profile $P'$ in which we remove $b'$, $b$ should still win. If we remove $b'$, the scores are $S(c) > S(a) = q+3$ and $S(b) = q+2$, so $b$ is eliminated first. This contradicts independence of clones.
	
	Assume now that $s(\tau_{22})_1 > 1$. Take $q \in \mathbb N$ such that $q > \frac{1}{s(\tau_{22})_1-1}$ and consider the following profile $P$:
	\begin{align*}
		q   &: \{b,b'\} \succ \{a,c\} \\
		q-2 &: \{a,c\} \succ \{b\} \succ \{b'\} \\
		3   &: \{a\} \succ \{b\} \succ \{b'\} \succ \{c\} \\
		4   &: \{c\} \succ \{b\} \succ \{b'\} \succ \{a\}
	\end{align*}
	In this profile, $S(c)  > S(a) = q+1$ and $S(b) = S(b') \ge qs(\tau_{22})_1$. Since $q > \frac{1}{s(\tau_{22})_1-1}$, $S(b) > S(a)$ and $a$ is eliminated, then $b'$ is eliminated and $b$  wins the majority vote against $c$. Since $b$ and $b'$ are clones, by independence of clones, $b$ should also win in the profile $P'$ in which we remove $b'$. If we remove $b'$, the scores are $S(c) > S(a) = q+1$ and $S(b) = q$, so $b$ is eliminated first. This contradicts independence of clones. Thus, we necessarily have $s(\tau_{22}) = (1,0)$.

	\textbf{Step 3.3.} We prove that $s(\tau_{13}) = (1,0)$. Assume first that $s(\tau_{13})_1 < 1$. Let $q > \frac{1}{1-s(\tau_{13})_1}$ and $P$ be the profile with $q$ orders $\{a\} \succ \{b,c,d\}$ and $q-1$ orders $\{b,c,d\} \succ \{a\}$. By respect for cohesive majorities, $a$ should be the winner, but in this case $a$ is eliminated first (as $s(\tau_{31}) = (1,0)$). This proves that $s(\tau_{13})_1 \ge 1$. Now assume that  $s(\tau_{13})_1 > 1$. Let $q \in \mathbb N$ be such that $q > \frac{1}{s(\tau_{13})_1 - 1}$ and consider the following profile $P$:
	\begin{align*}
		4   &: \{a\} \succ \{d\}\succ \{b\} \succ \{c\} \\
		1   &: \{a,b\}\succ \{c,d\} \\
		q+2 &: \{a,b,c\}\succ \{d\}  \\
		6   &: \{b,c,d\} \succ \{a\} \\
		q   &: \{d\} \succ \{b,c,d\}
	\end{align*}
	In this profile, $a$ is in the top indifference class of more than half of the votes, so respect for cohesive majorities implies that $f(P) \subseteq \{a,b,c\}$. The scores are $S(a) = q+7$, $S(b) = q+9$, $S(c) = q+8$ and $S(d) = qs(\tau_{13})_1+6$. Because of hypothesis on $q$, $S(d) > S(a)$ so $a$ is eliminated first. The scores are now $S(d) = q+4$, $S(b) = q+3$ and $S(c) = q+2$. $c$ is eliminated and $d$ wins the majority vote against $b$. This contradicts respect for cohesive majorities. Therefore, $s(\tau_{13}) = (1,0)$.
	
	\textbf{Step 4.} We show that $s(\tau_{211})_2 = 0$. Assume for a contradiction that $s(\tau_{211})_2 > 0$. Take $q \in \mathbb N$ such that $q > \frac{1}{s(\tau_{211})_2}$ and consider the profile $P$:
	\begin{align*}
		q    &: \{a,b\} \succ \{c\} \succ \{d\} \\
		q    &: \{a,b\} \succ \{d\} \succ \{c\} \\
		4    &: \{a\} \succ \{c\} \succ\{d\} \succ\{b\} \\
		2q-2 &: \{c,d\} \succ \{a,b\} \\
		5    &: \{b,c,d\} \succ \{a\}
	\end{align*}
	In this profile, $a$ is in the top indifference class of more than half of the vote, so by respect for cohesive majorities we should have $f(P) \subseteq \{a,b\}$. The scores are $S(a) = 2q+4$, $S(b) = 2q+5$, $S(c) = S(d) = 2q+3 + qs(\tau_{211})_2$. Since $q > \frac{1}{s(\tau_{211})_2}$, $S(d) > S(a)$. Thus $a$ is eliminated first. The new scores are $S(b) = 2q$, $S(c) = 2q+2$ and $S(d) = 2q-2$ so $d$ is eliminated next, and $c$ wins the majority vote against $b$. This contradicts respect for cohesive majorities. We conclude that $s(\tau_{211})_2 = 0$. By combining this with Step 2, we obtain that $s(\tau_{211}) = (1,0,0)$.
	
	\textbf{Step 5.} From the previous four steps, we obtained that $\tau_{21}$, $\tau_{31}$, $\tau_{12}$, $\tau_{22}$ and $\tau_{211}$ are associated to approval score vectors. We continue to focus on specific order types before the induction step. In this step, we prove the result for $\tau_{112}$, $\tau_{212}$ and $\tau_{121}$.

	\textbf{Step 5.1.} We show that $s(\tau_{112}) = (1,0,0)$. We first prove that $s(\tau_{112})_1 \ge 1$. For this, assume by contradiction that $s(\tau_{112})_1 < 1$. Take $q \in \mathbb N$ such that $q > \frac{1}{1-s(\tau_{112})_1}$ and consider the profile $P$ with $q$ orders $\{a\} \succ \{b\} \succ \{c,d\}$ and $q-1$ orders $\{b,c,d\} \succ \{a\}$. Because of hypothesis on $q$, we have $S(b) \ge S(c) = S(d) > S(a)$, therefore $a$ is eliminated first. However, $a$ is in the top indifference class of more than half of the voters, so this contradicts respect for cohesive majorities.
	
	We now prove that $s(\tau_{112})_1 \le 1$. For this, assume by contradiction that $s(\tau_{112})_1 > 1$. Take $q \in \mathbb N$ such that $q > \frac{1}{s(\tau_{112})-1}$ and consider the profile $P$:
	\begin{align*}
		4   &: \{a\} \succ \{d\} \succ \{b\} \succ \{c\} \\
		1   &: \{a,b\} \succ \{c\} \succ \{d\} \\
		q+2 &: \{a,b,c\} \succ \{d\} \\
		6   &: \{b,c,d\} \succ \{a\} \\
		q   &: \{d\}\succ \{c\} \succ \{a,b\}
	\end{align*}
	Then, we can use exactly the same reasoning as in Step 3.3 but with $\tau_{112}$ instead of $\tau_{13}$ to obtain a contradiction.

	Finally, we prove that $s(\tau_{112})_2 = 0$. Assume that $s(\tau_{112})_2 > 0$. Take $q \in \mathbb N$ such that $q > \frac{1}{s(\tau_{112})_2}$ and consider the profile $P$:
	\begin{align*}
		q    &: \{a\} \succ \{b\} \succ \{c,d\} \\
		q    &: \{a\} \succ \{c\} \succ \{b,d\} \\
		q    &: \{a\} \succ \{d\} \succ \{c,b\} \\
		3q-1 &: \{b,c,d\} \succ \{a\}
	\end{align*}
	In this profile $a$ is in the top indifference class of more than half of the vote, thus it should win the election by respect for cohesive majorities. The scores are $S(a) = 3q$ and $S(b) = S(c) = S(d) = 3q-1 + qs(\tau_{112})_2$. By our hypothesis on $q$, $a$ is eliminated first, which contradicts respect for cohesive majorities. We conclude that $s(\tau_{112}) = (1,0,0)$.
	
	\textbf{Step 5.2.} We now show that $s(\tau_{212}) = (1,0,0)$.
	
	We first prove that $s(\tau_{212})_1 \ge 1$. Assume that $s(\tau_{212})_1 < 1$. Take $q \in \mathbb N$ such that $q > \frac{1}{1-s(\tau_{212})_1}$ and consider the profile $P$:
	\begin{align*}
		q   &: \{a,b\} \succ \{d\} \succ \{c,c'\} \\
		q+3 &: \{d,c,c'\} \succ \{a,b\} \\
		4   &: \{a\} \succ \{c\} \succ \{c'\} \succ \{b\} \succ \{d\} \\
		5   &: \{b\} \succ \{c\} \succ \{c'\} \succ \{a\} \succ \{d\}
	\end{align*}
	In this profile, the scores are $S(b) > S(a) = qs(\tau_{212})_1+4$, $S(d) \ge S(c) = S(c') = q+3$. Since $q > \frac{1}{1-s(\tau_{212})_1}$, $a$ is eliminated first, then $c'$ and $d$ are eliminated, and $c$ wins the majority vote against $b$. Note that $c$ and $c'$ are clones in this profile, so $c$ should also win in the profile $P'$ in which we remove $c'$. In $P'$, we have $S(b) > S(a) = q+4$ and $S(c) = S(d) = q+3$, so $c$ and $d$ are eliminated first. This contradicts independence of clones. Therefore,  $s(\tau_{212})_1 \ge 1$.
	
	We now prove that $s(\tau_{212})_1 \le 1$. Assume that $s(\tau_{212})_1 > 1$. Take $q \in \mathbb N$ such that $q > \frac{1}{s(\tau_{212})_1-1}$ and consider the profile $P$:
	\begin{align*}
		q   &: \{c,c'\} \succ \{d\} \succ \{a,b\} \\
		q+2 &: \{d\} \succ \{c\} \succ \{c'\} \succ \{a\} \succ \{b\} \\
		q+1 &: \{a\} \succ \{c\} \succ \{c'\} \succ \{b\} \succ \{d\} \\
		q+2 &: \{b\} \succ \{c\} \succ \{c'\} \succ \{a\} \succ \{d\} 
	\end{align*}
	In this profile, the scores are $S(d) \ge S(b) > S(a) = q+1$ and $S(c) = S(c') = qs(\tau_{212})_1$. Because of the hypothesis on $q$, $S(c) > S(a)$, so $a$ is eliminated first. Then, $c'$, $b$ and $d$ are eliminated and $c$ is the winner. By independence of clones, $c$ should also be the winner if we remove its clone $c'$. However, in this case, because $s(\tau_{112}) = (1,0,0)$, $c$ is eliminated first with score $S(c) =q$. This contradicts independence of clones. Therefore, $s(\tau_{212})_1 = 1$.
	
	We finally prove that $s(\tau_{212})_2 = 0$. Assume that $s(\tau_{212})_2 > 0$. Take $q \in \mathbb N$ such that $q > \frac{1}{s(\tau_{212})_2}$ and consider the profile $P$:
	\begin{align*}
		q   &: \{c,c'\} \succ \{d\} \succ \{a,b\} \\
		q-2 &: \{d\} \succ \{a\} \succ \{c\} \succ \{c'\} \succ \{b\} \\
		q-1 &: \{a\} \succ \{c\} \succ \{c'\} \succ \{b\} \succ \{d\} \\
		q   &: \{b\} \succ \{a\} \succ \{c\} \succ \{c'\}  \succ \{d\} 
	\end{align*}
	In this profile, the scores are $S(c) = S(c') = S(b) =  q$, $S(a) = q-1$,  and $S(d) = q- 2+ qs(\tau_{212})_2 $. By our hypothesis on $q$, we have $S(d) > S(a)$ and $a$ is eliminated first. By independence of clones, this implies that if we remove the clone $c'$ of $c$, $a$ should not be the winner. However, in this new profile $P'$ without $c'$,  the scores are $S(c) = S(b) = q$, $S(a) = q-1$ and $S(d) = q-2$, so $d$ is eliminated first. Then $S(c) = S(b) = q$ and $S(a) = 2q-3$, so $c$ or $b$ is eliminated next. In both cases, $a$ wins the majority vote, which contradicts independence of clones. We can conclude that $s(\tau_{212}) = (1,0,0)$.
	
	\textbf{Step 5.3.} We now focus on $\tau_{121} = (1,2,1)$. Assume first that $s(\tau_{121})_1 < 1$. Let $q \in \mathbb N$ with $q > \frac{1}{1-s(\tau_{121})_1}$ and consider the profile $P$ with $q$ orders $\{a\} \succ \{b,c\} \succ \{d\}$ and $q-1$ orders $\{b,c,d\} \succ \{a\}$. In this profile, $a$ is on top of more than half of the votes, so it should win by respect for cohesive majorities. However, the scores are $S(a) = qs(\tau_{211})_1$ and $S(b) = S(c) \ge S(d) = q-1$. By our hypothesis on $q$, $S(d) > S(a)$, and $a$ is eliminated first, a contradiction.

	Assume now that $s(\tau_{121})_1 > 1$. Let $q \in \mathbb N$ with $q > \frac{1}{s(\tau_{121})_1-1}$ and consider the profile $P$.
	\begin{align*}
		q+1 &: \{d\} \succ \{b\} \succ \{a,a'\} \\
		q   &: \{b\} \succ \{a,a'\} \succ \{d\} \\ 
		q+2 &: \{a,a'\} \succ \{b\} \succ \{d\} 
	\end{align*}
	In this profile, the scores are $S(a) = S(a') > S(d) = q+1$ and $S(b) = qs(\tau_{121})_1$. By hypothesis on $q$, we have $S(b) > S(d)$, and $d$ is eliminated first. Then $a'$ (or $a$) is eliminated and $b$ wins the majority vote against $a$ (or $a'$). Note that $a$ and $a'$ are clones in $P$, so $b$ should also be a winner in the profile $P'$ without $a'$. In $P'$, the scores are $S(a) > S(d) = q+1$ and $S(b) = q$, so $b$ is eliminated first. This contradicts independence of clones.
	
	We now prove that $s(\tau_{121})_2 = 0$.  Let $q \in \mathbb N$ with $q > \frac{1}{s(\tau_{121})_2}$ and consider the following profile $P$.
	\begin{align*}
		q    &: \{a\} \succ \{b,c\} \succ \{d\} \\
		q    &: \{a\} \succ \{b,d\} \succ \{c\} \\
		2q-1 &: \{c,d,b\} \succ \{a\}
	\end{align*}
	In this profile, $a$ is in the top indifference class of more than half of the voters, so it should be the sole winner by respect for cohesive majorities. The scores are $S(b) > S(d) = S(c) = 2q-1 + qs(\tau_{121})_2$ and $S(a) = 2q$. By hypothesis on $q$, $a$ is eliminated first. This contradicts respect for cohesive majorities and proves that $s(\tau_{121}) = (1,0,0)$.	
	
	\textbf{Step 6.} We now proceed to induction steps. First, we focus on dichotomous orders and show that for all dichotomous orders $\tau = (k,k')$ we have $s(\tau) = (1,0)$. We prove it by induction on $m = k + k'$. We know this is true for 
	$\tau_{21}$, $\tau_{12}$, $\tau_{13}$, $\tau_{31}$, and $\tau_{22}$,
	so this is true for $m= 3$ and $m=4$.  Assume by induction that it is true up to some $m \ge 4$. Let $\tau = (k,k')$ with $k+k' = m+1 \ge 5$. This means that either $k \ge 3$ or $k' \ge 3$.
	
	\textbf{Step 6.1.} Assume first that $k \ge 3$, we will show that $s(\tau)_1 = 1$. Assume for a contradiction that $s(\tau)_1 < 1$. Take $q \in \mathbb N$ such that $q > \frac{1}{1-s(\tau)_1}$ and consider the profile $P$:
	\begin{itemize}
		\item $q$ orders $\{c_1, \dots, c_{k-2},a, a'\} \succ \{b_1, \dots, b_{k'}\}$.
		\item For all $i \in [2, k-2]$, $1$ linear order $\{c_i\} \succ \{c_1\} \succ \dots \succ \{a\} \succ \{a'\}$.
		\item $1$ linear order $\{a\} \succ \{a'\} \succ \{c_1\} \succ \dots $ 
		\item $1$ linear order $\{a'\} \succ \{a\} \succ \{c_1\} \succ \dots $ 
		\item For all $j \in [1,k']$, $q-1$ linear orders $\{b_j\} \succ \{c_1\} \succ \dots  \succ \{a\} \succ \{a'\}$
	\end{itemize}
	In this profile, the scores are $S(a) = S(a') = qs(\tau)_1 + 1$, $S(c_1) = qs(\tau)_1$, $S(c_i) = qs(\tau)_1 + 1$ for $i \in [2,k-2]$ and $S(b_j) = q-1$ for $j \in [1,k']$. Since $q > \frac{1}{1-s(\tau)_1}$, $q-1 > qs(\tau)_1$ and $c_1$ is eliminated first. We do not care who wins except that it is not $c_1$. Note that in this profile, $a$ and $a'$ are clones, so in the profile $P'$ without $a'$, $c_1$ should not be a winner. In $P'$, by induction hypothesis all order types have scoring vector $(1,0,\dots,0)$. Thus, the scores are $S(a) = q+1$, $S(c_1) = q$, $S(c_i) = q+1$ for $i \in [2,k-2]$ and $S(b_j) = q-1$ for $j \in [1,k']$. Therefore, $b_j$ are successively eliminated. After this, the score of $c_1$ is $S(c_1) = q + (q-1)k'$ and the score of all other candidates is $q+1$. Since $k' > 0$, they are successively eliminated and $c_1$ is the winner. This contradicts independence of clones.
	
	Assume now for contradiction that $s(\tau)_1 > 1$. Take $q \in \mathbb N$ such that $q > \frac{1}{s(\tau)_1-1}$ and consider the profile $P$ described above, but with $q+1$ of each order of the last category instead of $q-1$. In this profile, the scores are $S(a) = S(a') = qs(\tau)_1 + 1$, $S(c_1) = qs(\tau)_1$, $S(c_i) = qs(\tau)_1 + 1$ for $i \in [2,k-2]$ and $S(b_j) = q+1$ for $j \in [1,k']$. Since $q > \frac{1}{s(\tau)_1-1}$, $S(c_1) > S(b_j)$ for all $b_j$, and some $b_j$ is eliminated first. From this point, all scoring vectors are $(1,0,\dots,0)$ by induction hypothesis. The score of $c_1$ is now $S(c_1) = 2q+1$ and the score of any other candidate is $q+1$. As long as $c_1$ is not eliminated, no candidate can have a score higher than $q+2$, so $c_1$ survives until the end and is the winner of the election. By independence of clones, it should also be the winner in $P'$ in which we remove the clone $a'$ of $a$. In this profile, all scoring vectors are $(1,0,\dots,0)$ by induction hypothesis. The scores are $S(c_1) = q$ and $q+1$ for all other candidates. Thus, $c_1$ is eliminated first, which contradicts independence of clones. 
	This implies that $s(\tau)_1 = 1$.
	
	\textbf{Step 6.2.} If $k \le 2$ but $k' \ge 3$, we use a similar reasoning. Let $q \in \mathbb N$, and consider the following profile $P$:
	\begin{itemize}
		\item $q$ orders $\{c_1, \dots, c_{k}\} \succ \{b_1, \dots, b_{k'-2},a,a'\}$.
		\item $q$ linear orders $\{c_1\} \succ \{a\} \succ \{a'\} \succ \dots $.
		\item If $k = 2$, $q+1$ linear orders $\{c_2\} \succ \{c_1\} \succ \dots \succ \{a\} \succ \{a'\}$.
		\item $2q+2 $ linear orders $\{a\} \succ \{a'\} \succ \{c_1\} \succ \dots $ 
		\item $2q+2$ linear orders $\{a'\} \succ \{a\} \succ \{c_1\} \succ \dots $ 
		\item For all $j \in [1,k'-2]$, $2q\pm1$ linear orders $\{b_j\} \succ \{c_1\} \succ \dots  \succ \{a\} \succ \{a'\}$
	\end{itemize}
	If we assume $s(\tau)_1 < 1$, we take $q > \frac{1}{1-s(\tau)_1}$ and the profile in which each order of the last category appears $2q-1$ times. In this profile, the scores are $S(c_1) = q + qs(\tau)_1$, $S(b_j) = 2q-1$ for all $j \in [1,k'-2]$, $S(a) = S(a') =2q+2$ and if $k = 2$, $S(c_2) > S(c_1)$. Because of hypothesis on $q$, $S(b_j) > S(c_1)$ for all $b_j$ and $c_1$ will be eliminated first. In this profile, $a$ and $a'$ are clones, so in the profile $P'$ in which we remove $a'$, $c_1$ should not be a winner. In $P'$, the scores are $S(c_1) = 2q$, $S(a) = 2q+2$, $S(b_j) = 2q-1$ for all $j \in [1,k'-2]$ and if $k = 2$, $S(c_2) = 2q+1$. Thus, the candidates $b_j$ are eliminated, then $c_2$ if $k = 2$. At least two candidates are eliminated since $m \ge 4$, and at least one of them is a $b_j$ since $k' \ge 3$, so the score of $c_1$ is now $S(c_1) \ge 2q + (q+1) + (2q-1) = 5q$, and the score of $a$ is $S(a) = 4q+4$. $c_1$ is the winner, which contradicts independence of clones. 
	
	If we now assume that $s(\tau)_1 > 1$, we take $q > \frac{1}{s(\tau)_1 - 1}$ and the profile $P$ in which each order of the last category appears $2q+1$ times. This time, some $b_j$ is eliminated first, giving $2q+1$ additional points to $c_1$, and all the other candidates $b_j$ and $c_2$ (if $k = 2$) are eliminated successively, then $a'$ is eliminated, and the scores are $S(a) = 4q+4$ and $S(c_1) \ge 5q+1$ for the same reasons as above, thus $c_1$ is the winner. However, if we remove the clone $a'$ of $a$, $c_1$ is eliminated first with the lowest score of $S(c_1) = 2q$. This contradicts independence of clones. We conclude that $s(\tau) = (1,0)$ for all dichotomous order types $\tau$.
	
	\textbf{Step 7.} We now proceed to the ultimate step. Note that for all remaining order types $\tau$ where $k = |\tau|$ is the length of the order, one of the following is true: (1) $\tau(1) \ge 3$, (2) $\tau(k) \ge 3$, (3) there is some $j \notin \{1,k\}$ such that $\tau(j) \ge 2$ or (4) there is some  $j \notin \{1,k-1\}$ such that $\tau(j) = \tau(j+1) = 1$. Note that any order type of size $k = |\tau| \ge 4$ satisfies either (3) or (4). Moreover, we already prove the results for dichotomous order types. The only order types of size $k=3$ which do not satisfy any of the above conditions are $\tau_{111}, \tau_{211}, \tau_{112}$ and $\tau_{212}$, and we already prove the results for all of these. Moreover, we also proved the result for $\tau_{121}$ and thus for all order types on $m \ge 4$ candidates.
	We prove the result for the remaining order types by induction on the number of candidates $m$. Let us assume that the result is true up to some $m \ge 4$, let's prove it is true for $m+1$. Let $\tau$ be an order type for $m+1$ candidates. We will use a similar idea as in the previous steps to prove that $s(\tau) = (1,0,\dots,0)$. Let $q \in \mathbb N$ and $P$ be the following profile on $C = \{a,a', b,d , c_1, \dots c_{m-3}\}$:
	\begin{itemize}
		\item $q$ orders of the type $\tau$. We have $b$ in the top indifference class ($b \in C_1$) and $d$ in the last indifference class ($d \in C_k$). $a$ and $a'$ are in a position to be clones, which is possible because $\tau$ satisfies one of the 4 conditions detailed above. We either put them both in indifference class $j \in [1,k]$ (if $\tau$ satisfies one of the first three conditions) or we put $a$ alone in indifference class $j \notin \{1,k-1\}$ and $a'$ alone in indifference class $j+1$ (if $\tau$ satisfies the fourth condition).
		\item $q$ linear orders $\{b\} \succ \{d\} \succ \dots \succ \{a\} \succ \{a'\}$.
		\item $2q\pm1$ linear orders $\{d\} \succ \{b\} \succ \dots \succ \{a\} \succ \{a'\}$.
		\item For all $c_j$ that are in the top indifference class of orders of type $\tau$, $q+2$ linear orders $\{c_j\} \succ \{b\} \succ \dots \succ \{a\} \succ \{a'\}$.
		\item For all $c_j$ that are not in the top indifference class of the orders of type $\tau$, $2q+2$ linear orders $\{c_j\} \succ \{b\} \succ \dots \succ \{a\} \succ \{a'\}$.
		\item If $a$ and $a'$ are in the top indifference class of the orders of type $\tau$, $2$ orders of this kind, otherwise, $2q+2$ linear  orders $\{a\} \succ \{a'\} \succ \{b\} \succ \dots$
		\item If $a$ and $a'$ are in the top indifference class of the orders of type $\tau$, $2$ orders of this kind, otherwise $2q+2$ linear  orders $\{a'\} \succ \{a\} \succ \{b\} \succ \dots$.
	\end{itemize}
	
	\textbf{Step 7.1.} We assume first that $s(\tau)_1 > 1$. Then, we take the profile $P$ with $2q+1$ orders $\{d\} \succ \{b\} \succ \dots \succ \{a\} \succ \{a'\}$. Moreover, we take $q > \frac{1}{s(\tau)_1-1}$. The scores are $S(a) \ge S(a') \ge 2q+2$, $S(b) = q + qs(\tau)_1$, $S(d) = 2q+1$ and for all $c_j$, $S(c_j) \ge 2q+2$. Because of hypothesis on $q$, $S(b) > S(d)$. Moreover, $S(a) > S(d)$ and $S(c_j) > S(d)$ for all $j \in [1, m-3]$. Thus, $d$ is eliminated first. From this point all order types have scoring vector $(1,0,\dots,0)$ by induction hypothesis. The new score of $b$ is $S(b) = 4q+1$. As long as $b$ is not eliminated, the score of any other candidate is at most $2q+2$, except for $a$ which get a score $4q+4$ once its clone $a'$ is eliminated. In the end, only $a$ and $b$ remain, with $S(b) \ge q + q + 2q+1 + (m-3)(q+2) \ge 5q+3$ (since $m \ge 4$) and $S(a) = 4q+4$. Thus, $b$ wins the majority vote. By independence of clones, $b$ should also win in the profile without the clone $a'$ of $a$. By induction hypothesis, all orders have scoring vector $(1,0,\dots,0)$ in this profile. The scores are $S(b) = 2q$, $S(d) = 2q+1$, $S(c_j) = 2q+2$ for all $j \in [1,m-3]$ and $S(a) = 4q+4$. Thus, $b$ is eliminated first, which contradicts independence of clones.
	
	\textbf{Step 7.2.} We now assume that $s(\tau)_1 < 1$. Then, we take the profile $P$ with $2q-1$ orders $\{d\} \succ \{b\} \succ \dots \succ \{a\} \succ \{a'\}$. Moreover, we take $q > \frac{1}{1-s(\tau)_1}$. In this profile, the scores are $S(a) \ge S(a') \ge 2q+2$, $S(b) = qs(\tau)_1 + q$, $S(d) = 2q-1$ and $S(c_i) \ge qs(\tau)_1 + q+2 > S(b)$. Because of hypothesis on $q$, we have that $S(d) > S(b)$. Thus, $b$ is eliminated first. By independence of clones, this implies that in the profile $P'$ without the clone $a'$ of $a$, $b$ should not be a winner. By induction hypothesis, all orders have scoring vector $(1,0,\dots,0)$ in this profile, and thus the first candidate eliminated is $d$ with score $S(d) = 2q-1$. The score of $b$ is then $S(b) = 4q-1$, the score of $a$ is $S(a) = 4q+4$ and all other candidates have score at most $2q+2$ until $b$ is eliminated. Thus, the $c_j$ are successively eliminated until only $a$ and $b$ remain. The scores are $S(a) = 4q+4$ and $S(b) \ge q+q+(2q-1)+(m-3)(q+2) \ge 5q+1$ (since $m \ge 4$). $b$ wins the majority vote, which contradicts independence of clones. Therefore $s(\tau)_1 = 1$.
	
	\textbf{Step 7.3.} We finally show that $s(\tau)_2 = 0$. Assume that $s(\tau)_2 > 0$. Let $q > \frac{1}{s(\tau)_2}$. If $\tau(1) = 1$, consider the profile $P$ with $C = \{a,c_1,\dots,c_m\}$:
	\begin{itemize}
		\item For each $i \in[1,m]$, $q$ orders of the type $\tau$ with $C_1 = \{a\}$ as the top indifference class and $C_2 = \{c_i, \dots,c_{i+\tau(2)-1}\}$ as the second indifference class (subscripts should be considered modulo $m$).
		\item $mq-1$ dichotomous orders $\{c_1, \dots, c_m\} \succ \{a\}$.
	\end{itemize}
	In this profile, $a$ appears at the top of more than half of the votes, so by respect for cohesive majorities it should be the sole winner. The scores are $S(a) = mq$ and $S(c_i) \ge mq-1 + qs(\tau)_2 > mq$ by hypothesis on $q$. Thus, $a$ is eliminated first, which contradicts respect for cohesive majorities. 
	
	Now, if $\tau(1) \ge 2$, consider the following profile $P$ on $C = \{a,a',d,b, c_1, \dots,c_{m-3}\}$:
	\begin{itemize}
		\item $q$ orders of type $\tau$ with $a$ and $a'$ in the top indifference class $C_1$, $b$ in the second indifference class $C_2$ and $d$ in the last indifference class $C_k$ (which is not the second one because $\tau$ is not dichotomous).
		\item $q-2$ linear orders $\{b\} \succ \{d\} \succ \{a\} \succ \{a'\} \succ \dots$
		\item $q-1$ linear orders $\{d\} \succ \{b\} \succ \{a\} \succ \{a'\} \succ \dots$
		\item For all $c_i$ that are not in the top indifference class $C_1$ of the orders with type $\tau$, $q$ linear order $\{c_i\} \succ \{d\} \succ \{b\} \succ \{a\} \succ \{a'\} \succ \dots$
	\end{itemize}
	In this profile, the scores are $S(a) = S(a') = q$, $S(b) = q-2 + qs(\tau)_2$, $S(d) = q-1$ and $S(c_i) \ge q$ for $i \in [1,m-3]$. By our hypothesis on $q$, $S(b) > S(d)$ and $d$ is eliminated first. By our induction hypothesis, all scoring vectors are now $(1,0,\dots,0)$. The score of $b$ is now $S(b) = 2q-3$ and the score of any other candidate is at most $q$ until $b$ is eliminated. Thus, all candidates are successively eliminated until only $b$ remains and wins. By independence of clones, $b$ should also win in the profile $P'$ in which we remove the clone $a'$ of $a$. In this profile, $b$ has the lowest score $S(b) = q-2$, so it is eliminated first, which contradicts independence of clones. Therefore, $s(\tau)_2 = 0$.
	
	We proved that $s(\tau)_1 = 1$ and $s(\tau)_2 = 0$, and by definition $s(\tau)_1 \ge s(\tau)_2 \ge s(\tau)_3 \ge \dots \ge s(\tau)_k \ge 0$, so $s(\tau) = (1,0,\dots,0)$. The induction concludes that for all order types $\tau$, we have $s(\tau) = (1,0,\dots,0)$, which means that $f$ is actually Approval-IRV.
\end{proof}

\subsection{Proof of \Cref{thm:monocharacterization}}
\label{app:monocharacterization-proof}

\monocharacterization*
\begin{proof}
	\Cref{thm:AVIRVmono} already shows that this rule is indifference monotonic. The consistency with IRV is clear. Let's now show that no other elimination scoring rule satisfies these two axioms.

	For this, we prove that for all order types $\tau = (\tau(1), \dots, \tau(k))$, the associated scoring vector is $s(\tau) = (1,0,\dots,0)$.
	
	First, we use consistency with IRV to show that it is the case for all linear orders. Because there is only one possible linear order type $\tau = (1, \dots, 1)$ for each number of candidates $m$, we denote $s(m) = s(\tau)$ for $|\tau| = m$. We know that $s(m)_1 > 0$ otherwise $s(m) = (0,0,\dots,0,0)$ and $a$ is not the only winner in the profile with one ranking where $a$ is first $\{a\} \succ \dots$. Thus, we can assume without loss of generality that for all $m$, $s(m)_1 = 1$ (as there is only one linear order type for each number of candidates). We now prove that for all $m$, $s(m)_2 = 0$.
	
	We prove it by induction on $m$. It is clearly true for $m=2$ as $\tau = (1,1)$ so $s(2) = (1,0)$. Assume it is true for $m \ge 2$, we prove it for $m+1$.
	
	Assume for contradiction that $s(m+1)_2 > 0$. Let $q \in \mathbb N$ with $q > \frac{1}{s(m+1)_2}$ and consider the profile $P$ on  $C = \{a,c_1,\dots,c_m\}$:
	\begin{itemize}
		\item For each $i \in [1,m]$, $q$ linear orders $\{c_i\} \succ \{a\} \succ \{c_{i+1}\} \succ \dots \succ \{c_{i+m-1}\}$, where the index of the $c_i$ have to be taken modulo $m$.
		\item $q-1$ linear orders $\{a\} \succ \{c_1\} \succ \dots \succ \{ c_m\}$.
	\end{itemize}
	In this profile, observe that each $c_i$ appears at least $q$ times last in a linear order. With IRV, $a$ is eliminated in the first round, as it has score $S(a) = q-1$ while all other candidates have score $s(c_i) = q$. However, here the scores are $S(a) = (q-1) + mqs(m+1)_2$, and $S(c_i) \le q + (m-1)qs(m+1)_2 + q\cdot 0$. We only upper bound because we know that $s(m+1)_i \le s(m+1)_2$ for all $i \ge 3$. However, because we assumed $q > \frac{1}{s(m+1)_2}$, we have $S(a) > S(c_i)$ for all $c_i$. Thus, one $c_i$ is eliminated, assume $c_1$ without loss of generality. By induction hypothesis, all scoring vectors are of the form $(1,0,\dots,0)$ from this point. The new score of $a$ is $S(a) = 2q-1$ and the score of all other candidates is upper bounded by $q$ until $a$ is eliminated. Therefore, $a$ is never eliminated, and wins the election, a contradiction. This shows that $s(m+1)_2 = 0$. Since $s(m+1)_i \le s(m+1)_2$ for all $i \ge 3$, then the scoring vector of the linear order on $m+1$ candidates is $s(m+1) = (1,0,\dots,0)$.

	We can now show that this is true for all order types. The proof is done by induction on the number of candidates $m$. It is clearly true for $m = 2$ as the only possible order is linear. Assume that it is true up to $m-1$ and let us show it for $m \ge 3$.
	
	To show it is the case for all order types on some number of candidates $m \ge 3$, we will do another induction, this time on the number of candidates that are not alone in their indifference class (i.e., on $p = \sum_{i : \tau(i) > 1} \tau(i)$). For $p = 0$, the order is linear, so we already know that its scoring vector is $(1,0, \dots, 0)$. Assume it is true up to some $p \ge 0$, and let us prove this is also true for $p+1$. 
	
	Assume by contradiction that there is an order type $\tau = (\tau(1), \dots, \tau(k))$ with $\sum_{i : \tau(i) > 1} \tau(i) = p+1$ but the scoring vector is such that $s(\tau)_1 \ne 1$ or $s(\tau)_2 > 0$. Let $j$ be the minimal index such that $\tau(j) > 1$. Define the order type $\tau' = (\tau(1), \dots, \tau(j)-1, 1, \tau(j+1), \dots, \tau(k))$. Equivalently, we have $\tau'(i) = \tau(i)$ if $i < j $, $\tau'(j) = \tau(j)-1$, $\tau'(j+1) = 1$ and $\tau'(i) = \tau(i-1)$ for $i > j+1$. We have $\sum_{i : \tau'(i) > 1} \tau'(i) \le \sum_{i : \tau(i) > 1} \tau(i) - 1 \le p$, so by induction on $p$ its associated scoring vector is $s(\tau') = (1,0,\dots,0)$.
	
	Let us assume first that $s(\tau)_1 < 1$.  Take $q \in \mathbb N$ such that $q > \frac{3}{1-s(\tau)_1}$, and consider the following profile $P'$ with candidate set $C = \{a,b,d, c_1,\dots,c_{m-3}\}$.
	
	\begin{itemize}
		\item $q$ orders with order type $\tau'$ such that $a \in C_{j+1}$ (alone in its indifference class), $b \in C_1$, and $d \notin C_1$ (this is possible since $\tau$ has at least two indifference classes, so $\tau'$ has at least three).
		\item $q+2$ linear orders $\{b\} \succ \{d\} \succ \{a\} \succ \dots$.
		\item $2q$ linear orders $\{a\} \succ \{d\} \succ \dots$.
		\item $2q-1$ linear orders $\{d\} \succ \{a\} \succ \dots$.
		\item For all $j \in [1,m-3]$, $2q$ linear orders $\{c_j\} \succ \{d \} \succ \{ a \} \succ \dots$.
	\end{itemize}
	
	Since $\tau'$ satisfies the inductive hypothesis on $p$ and the other orders are linear orders, all order types in this profile are associated with scoring vectors of the form $(1,0,\dots,0)$. Then, the score of every $c_j$ is $S(c_j) \ge 2q$. The score of $a$ is $S(a) = 2q$, the score of $b$ is $S(b) = 2q+2$ and the score of $d$ is $S(d) = 2q-1$. Therefore, $d$ is eliminated first. By induction on $m$, for all the following steps all scoring vectors are also of the form $(1,0,\dots,0)$. Thus, the new score of $a$ is $S(a) = 4q-1$, the one of $b$ is $S(b) = 2q+2$ and the score of all other candidates is $S(c_i) \le 3q$. After each elimination, the score of $a$ increases and the one of all other candidates is always upper bounded by $3q$ as long as $a$ is not eliminated. Thus, $a$ is the winner of this election.
	
	Now, consider the profile $P$ in which we applied $a$-hover transformation to every order of type $\tau'$, and thus obtained orders of type $\tau$ (by merging $a$ with the indifference class above it). By indifference monotonicity, $a$ is still a winner in $P$. In $P$, the score of $b$ is $S(b) = q+2+qs(\tau)_1$, the score of $d$ is  $S(d) \ge 2q-1$, the score of $a$ is $S(a) \ge 2q$ and the score of all other candidates $c_j$ is $S(c_j) \ge 2q$. Because $q > \frac{3}{1-s(\tau)_1}$, we have $S(d) > S(b)$, so $b$ is eliminated first. Therefore, $d$ get ranked first in $q+2$ additional orders. By induction on $m$, for all the following steps all scoring vectors are of the form $(1,0,\dots,0)$. The new score of $d$ is $S(d) \ge 3q+1$, while the score of all other candidates is at most $3q$ as long as $d$ is not eliminated. Thus, $d$ is the sole winner instead of $a$. This shows by contradiction that $s(\tau)_1 \ge 1$.
	
	Assume now that $s(\tau)_1 > 1$. Take $q\in \mathbb N$ such that $q > \frac{1}{s(\tau)_1-1}$ and denote $P'$ the following profile.
	
	\begin{itemize}
		\item $q$ orders with order type $\tau'$ with $a \in C_{j+1}$ (alone in its indifference class), $b \in C_1$ and $d$ at the bottom of the order: If $j < k$, then $d \in C_{k+1}$ and if $j=k$, $d \in C_k$ (in the second case, $C_{k+1} = \{a\}$).
		\item $q-2$ linear orders $\{b\} \succ \{a\} \succ \{d\} \succ\dots$.
		\item $2q$ linear orders $\{a\} \succ\{b\} \succ \dots$.
		\item $2q-1$ linear orders $\{d\} \succ \{b\}\succ \{a\} >\dots$.
		\item For all $j \in [1,m-3]$, $2q$ linear rankings $\{c_j\} \succ \{b\} \succ \{a\} \succ \dots$.
	\end{itemize}
	
	In this profile, all scoring vectors are $(1,0,\dots,0)$. The score of $a$ is $S(a) = 2q$, the score of $b$ is $S(b) = 2q-2$, the score of $d$ is $S(d) = 2q-1$ and the score of all $c_i$ is $S(c_i) \ge 2q$. Thus, $b$ is eliminated first. Now, by induction on $m$ all scoring vectors are $(1,0,\dots,0)$. The score of $a$ is $S(a) \ge 3q-2$, the score of $d$ is $S(d) \le 3q-3$ and the score of all other $c_j$ is  $S(c_i) \le 3q$. Moreover, since $d$ is in the last possible position in the first $q-2$ orders (originally with order type $\tau'$), the score of $d$ is necessarily smaller than the one of all $c_j$. Thus, $d$ is eliminated next, and the score of $a$ is now at least $S(a) \ge 5q-3$. The score of all other candidates $c_j$ is upper bounded $3q$ as long as $a$ is not eliminated. The candidates $c_j$ get successively eliminated and $a$ is the winner.
	
	Consider the profile $P$ obtained from $P'$ by applying $a$-hover transformations to all orders with type $\tau'$ (thus giving an order with order type $\tau$). By indifference monotonicity, $a$ is also a winner in this profile. Note that in the rankings with order type $\tau$, $d$ is always in the last indifference class, therefore getting score $0$ from the voters with this order type. In $P$, the score of $b$ is $S(b) = qs(\tau)_1 + q-2$, the score of $d$ is $S(d) = 2q-1$, the score of $a$ is $S(a) \ge 2q$ and the score of all other candidates $c_j$ is $S(c_j) \ge 2q$. Since $q > \frac{1 }{s(\tau)_1-1}$, $S(b) > S(d)$, thus $d$ is eliminated first. By induction on $m$, we can now assume that all order types are associated with scoring vector $(1,0,\dots,0)$. Therefore, $b$ has score $S(b) = 4q-3$, and $a$ and all $c_j$ all have score upper bounded by $3q-2$ as long as $b$ is not eliminated. Moreover, the elimination of $a$ or any $c_j$ increases the score of $b$. Therefore, $b$ wins instead of $a$ in $P$. This shows by contradiction that $s(\tau)_1 = 1$. 
	
	We now show that $s(\tau)_2 = 0$. Assume by contradiction that $s(\tau)_2 > 0$. Take $q \in \mathbb N$ such that $q > \frac{1}{\alpha_2}+2$ and consider the following profile $P'$.
	
	\begin{itemize}
		\item $q-2$ orders with type $\tau'$ with $a \in C_{j+1}$ (alone in its indifference class), $b \in C_1$ and $d$ in the second indifference class ($a$ excluded): If $j \ne 1$, then $d \in C_2$ and otherwise $d \in C_3$ (because $C_2 = \{a\}$).
		\item $q+1$ linear orders $\{b\} \succ \{d\} \succ  \{a \} \succ \dots$
		\item $2q$ linear orders $\{a\} \succ \{d\} \succ  \dots$
		\item $2q-2$ linear orders $\{d\} \succ \{a\} \succ \dots$
		\item For all $j \in [1,m-3]$, $2q$ linear orders $\{c_j\} \succ \{d\} \succ \{a\} \succ \dots$.
	\end{itemize}
	
	In this profile, all order types are associated with scoring vectors of the form $(1,0,\dots,0)$. The score of $a$ is $S(a) = 2q$, the score of $b$ is $S(b) = 2q-1$, the score of $d$ is $S(d) = 2q-2$ and the score of all $c_j$ is $S(c_j) = 2q$. Thus, $d$ is eliminated first. By induction on $m$, all order types are now associated with scoring vectors of the form $(1,0,\dots,0)$. The score of $a$ is now $S(a) = 4q-2$, the score of $b$ is $S(b) = 2q-1$, and all other candidates have score lower than $3q-2$ but greater than $2q$. $b$ is eliminated next, then the candidates $c_j$. After each elimination (of $b$ or some $c_j$), the score of $a$ increases, and the score of all other candidates remains bounded by $3q-2$ as long as $a$ is not eliminated. Therefore, $a$ is the winner in $P'$.
	
	If we apply $a$-hover transformations (by moving $a$ in the indifference class above) for all orders with type $\tau'$, we obtain orders with type $\tau$, and a new profile $P$. By indifference monotonicity, $a$ should be a winner in $P$. Observe that $d$ is always in the second indifference class in orders with order type $\tau$. In $P$, the score of $b$ is $S(b) = 2q-1$, the score of $d$ is $S(d) = (q-2)\alpha_2+2q-2$, the score of $a$ is $S(a) \ge 2q$ and the score of other candidates $c_j$ is  $S(c_j) \ge 2q$. Thus, since $q > \frac{1}{\alpha_2}+2$, $S(d) >  S(b)$ and $b$ is eliminated first. Now, by induction on $m$, all order types are associated with scoring vectors of the form $(1,0,\dots,0)$. The score of $d$ is $S(d) \ge 3q-1$. The score of all other candidates (including $a$) is upper bounded by $3q-2$ as long as $d$ is not eliminated. Therefore, $a$ and the candidates $c_j$ will be successively eliminated and there eliminations can only make the score of $d$ increase as the other scores will still be upper bounded by $3q-2$. $d$ wins the election in $P$, which contradicts indifference monotonicity. We proved that $s(\tau)_2 =  0$. 
	
	Since by definition of scoring vectors we have $s(\tau)_2 \ge s(\tau)_j $ for all $j \ge 2$, this implies that the order type $\tau$ is associated to the scoring vector $s(\tau) = (1,0,\dots,0)$. Induction on $p$ concludes that this is true for all order types on $m$ candidates. Induction on $m$ concludes that it is true for all order types.
\end{proof}

\section{Further Details on Experiments}

In this appendix, we give further details on our experiments. In particular, we describe the various datasets we used in \Cref{app:experiments-datasets}, and we show and discuss the complete set of experimental results we obtained for the single-winner rules (Approval-IRV and Split-IRV) in \Cref{app:experiments-results}.

\subsection{Datasets}
\label{app:experiments-datasets}

In our experiments, we generated profiles of weak orders in two steps: first, we sample a profile of linear orders (rankings), and second, we introduce indifferences. 

We now describe the several models we used to obtain profiles of linear orders.
\begin{itemize}
	\item \textbf{Impartial Culture:} rankings are drawn i.i.d. uniformly at random among all $m!$ possible rankings.
	\item \textbf{Mixture of $k$ Mallows:} we first randomly sample $k$ central rankings $\sigma_1, \dots, \sigma_k$ uniformly at random. Then, for each voter we chose uniformly at random one of the $k$ central rankings, and we sample a ranking deviating from this central ranking using a Mallows model with dispersion parameter $\phi\in [0,1]$. In our experiments, we take $\phi = 0.5$ and $k=4$. We recall that for a Mallows model of parameter $\phi$ and central ranking $\sigma$, the probability to sample a ranking $\rho$ is $\phi^{\text{KT}(\sigma,\rho)}/C$ where KT is the Kendall-Tau distance and $C$ a normalization constant.
	\item \textbf{$d$-dimensional Euclidean:} voters and candidates are associated to ideal positions $p(v), p(c)$ in a $d$-dimensional metric space (here $\mathbb R^d$), and voters rank candidates in order of increasing distance between the voter's position $p(v)$ and candidates' positions $p(c)$. In our experiments, we take $d \in \{1,2\}$, and we sample the positions uniformly at random from the unit square or the unit disc.
	\item \textbf{Sampling from real data:} using a dataset of real rankings (potentially with weights), we randomly sample rankings from the dataset with replacement (based on the weights if they exist, otherwise uniformly at random).
\end{itemize}

For the last method, we need real datasets. We used data from two sources, described below.
\begin{itemize}
	\item \textbf{French presidential elections:} in 2017 and 2022, French voters could participate in online surveys in which they could try different alternative voting methods, such as approval or IRV \citep{bouveret_2018_1199545,delemazure_2024_10568799}. We use the preferences given by the participants for the IRV method. As the authors note, the sample of participants is heavily biased towards the left, so in order to obtain a distribution of the votes more faithful to the population-wide distribution of opinions, we re-weighted the voters based on their official vote at the election, as proposed by the authors.
	\item \textbf{Irish election:} we use the datasets of elections held in Dublin, Ireland, in 2002 for which the IRV method with truncated ballot was used, as available in PrefLib \citep{MaWa13a}.
\end{itemize}

In both cases, voters were allowed to only give truncated rankings, for instance by giving only their top-4 candidates, and implicitly reporting indifference between all the others. In our experiments, we only kept the voters who gave full rankings. (This is the reason we did not use datasets from U.S. elections \citep[e.g.,][]{otis2022data}, which contain only very truncated rankings.) The following table describes the number of voters $n$ and candidate $m$ for each election. Note that the number of voters $n$ we provide is the one after removing all incomplete orders.

\begin{center}
\begin{tabular}{ccccc}
	\toprule
	& France 2017 & France 2022 & Dublin Meath & Dublin West \\
	\midrule
	$n$  & 404 &5\,126&3\,165& 4\,810\\
	$m$ & $11$ &$12$& $14$ & $9$ \\
	\bottomrule
\end{tabular}
\end{center}

Finally, the two methods we used to introduce indifferences in the complete rankings of these datasets are already described in \Cref{sec:expe}.

\subsection{Results}
\label{app:experiments-results}

\begin{figure}[t]
	\centering
	\includegraphics[width=1\linewidth]{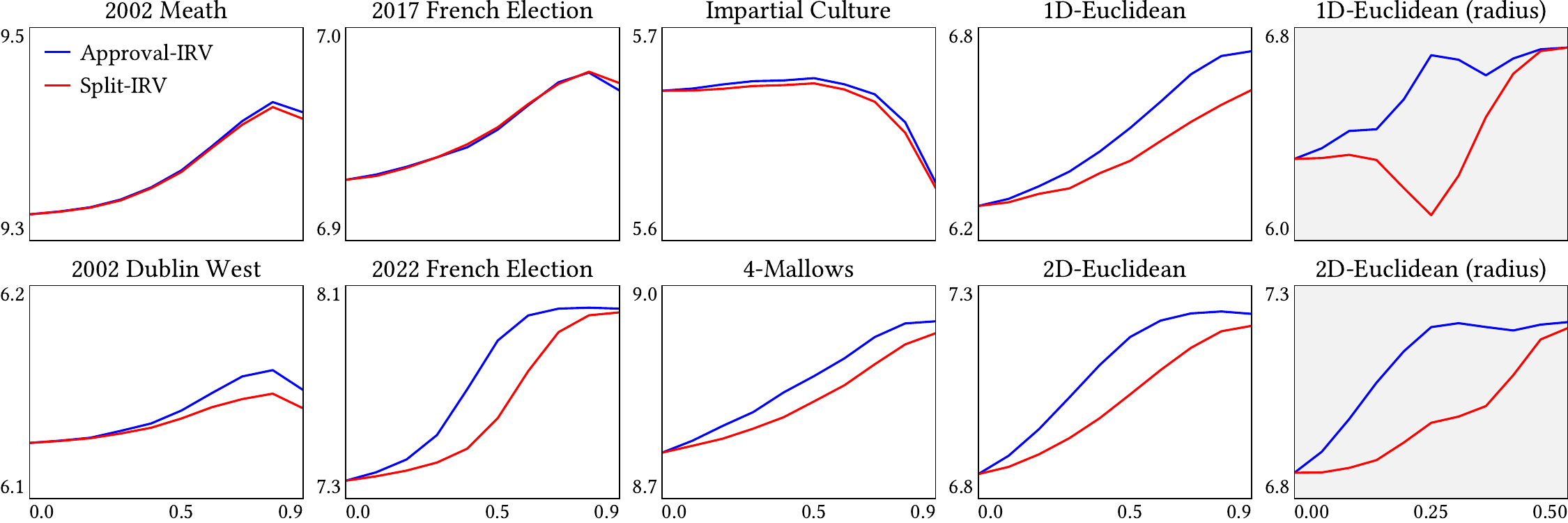}
	\caption{Average Borda score of the winner (normalized by dividing by $n$) for various datasets.}
	\label{fig:exp-borda}
\end{figure}

\begin{figure}[t]
	\centering
	\includegraphics[width=1\linewidth]{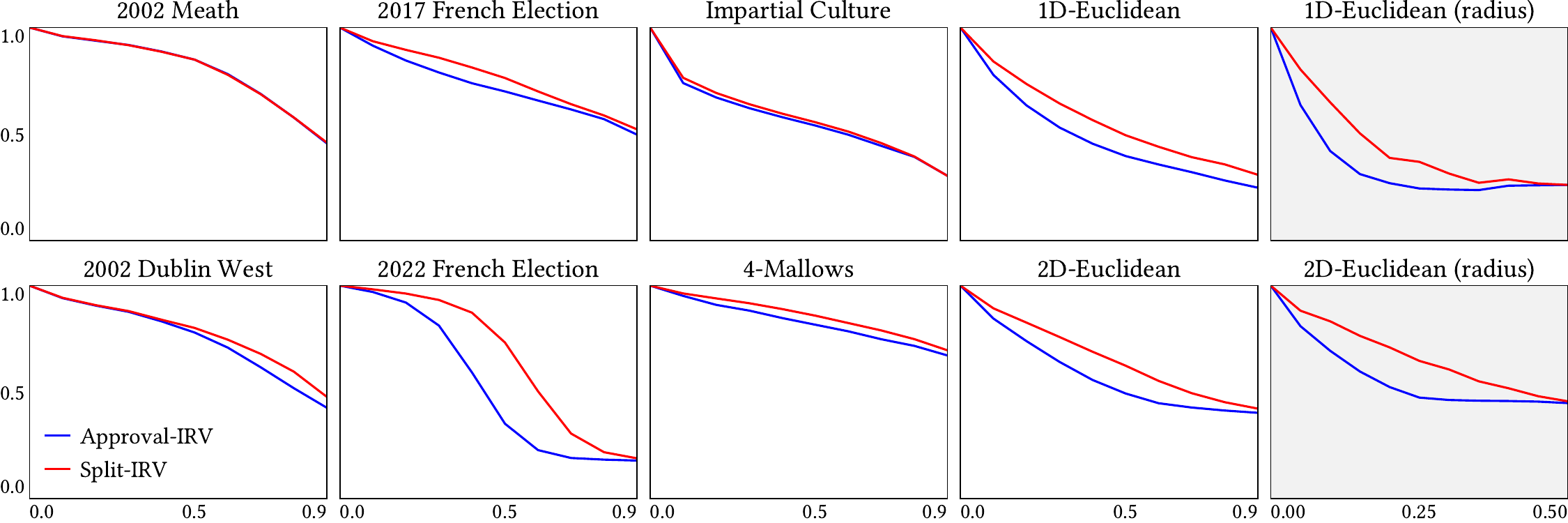}
	\caption{Frequency of agreement between the rule and linear-order IRV for various datasets.}
	\label{fig:exp-similarity}
\end{figure}

\begin{figure}[t]
	\centering
	\includegraphics[width=1\linewidth]{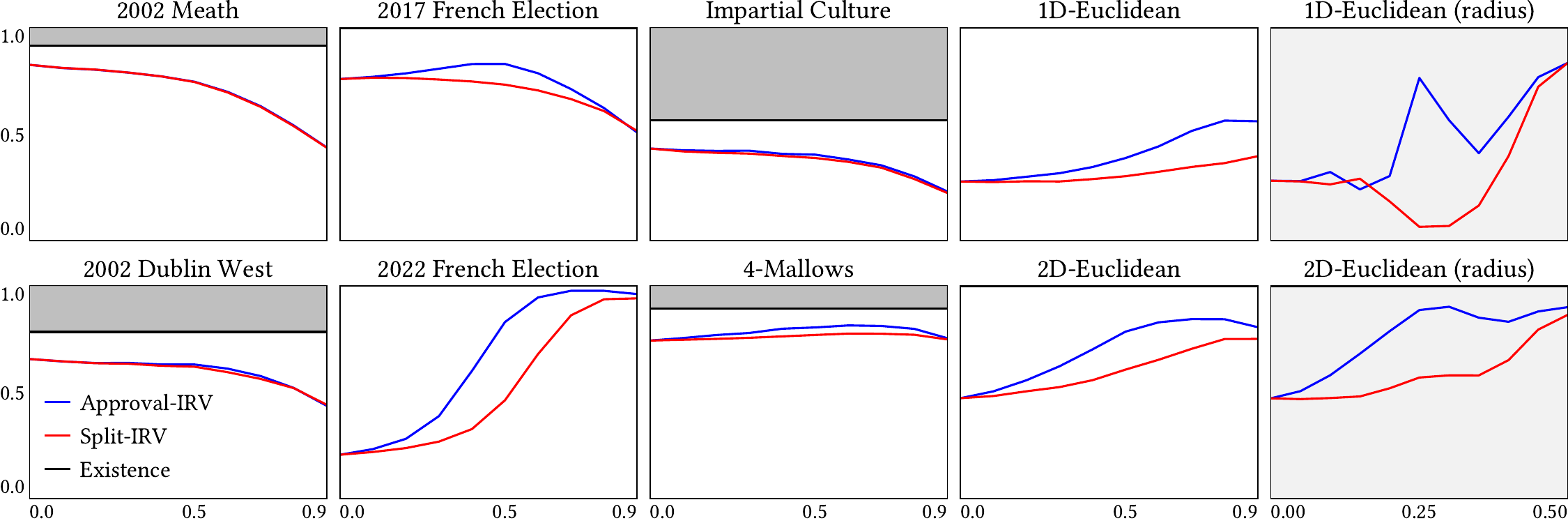}
	\caption{Frequency of finding the Condorcet winner, and frequency of such candidate existing for various datasets.}
	\label{fig:exp-condorcet}
\end{figure}

\begin{figure}[t]
	\footnotesize
	\setlength{\tabcolsep}{2pt}
	\begin{subfigure}{0.47\textwidth}
		\centering
		\includegraphics[width=\linewidth]{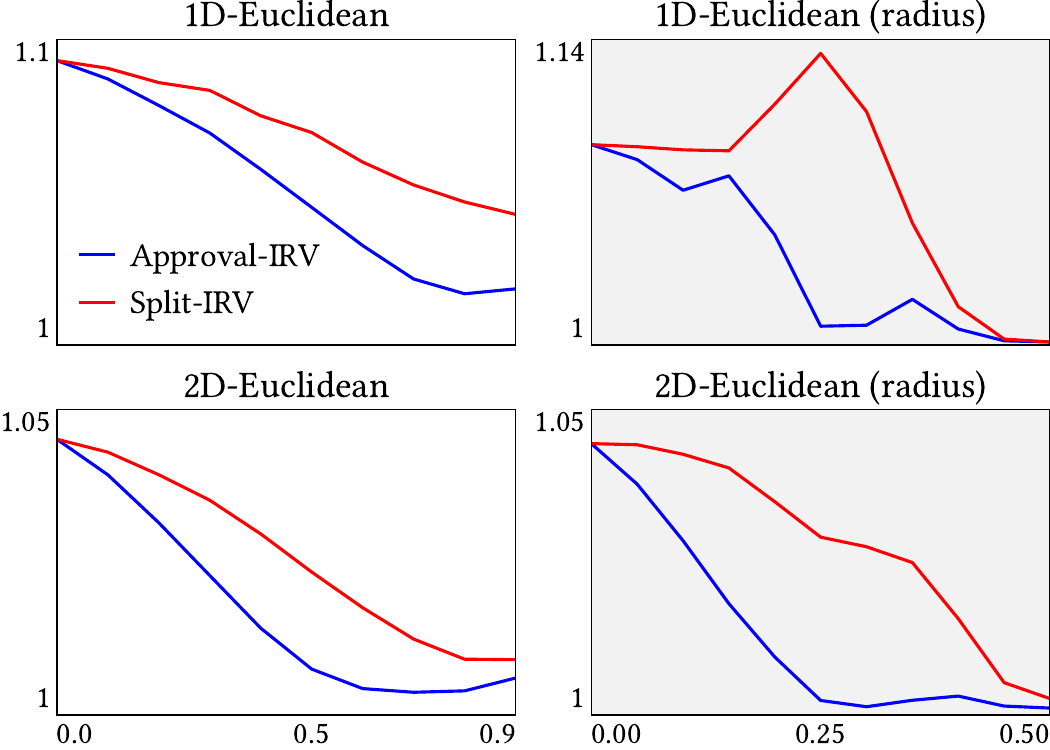}
		\caption{Average distortion of the winner.}
		\label{fig:exp-distortion}
	\end{subfigure}
	\quad
	\begin{subfigure}{0.47\textwidth}
		\centering
		\includegraphics[width=\linewidth]{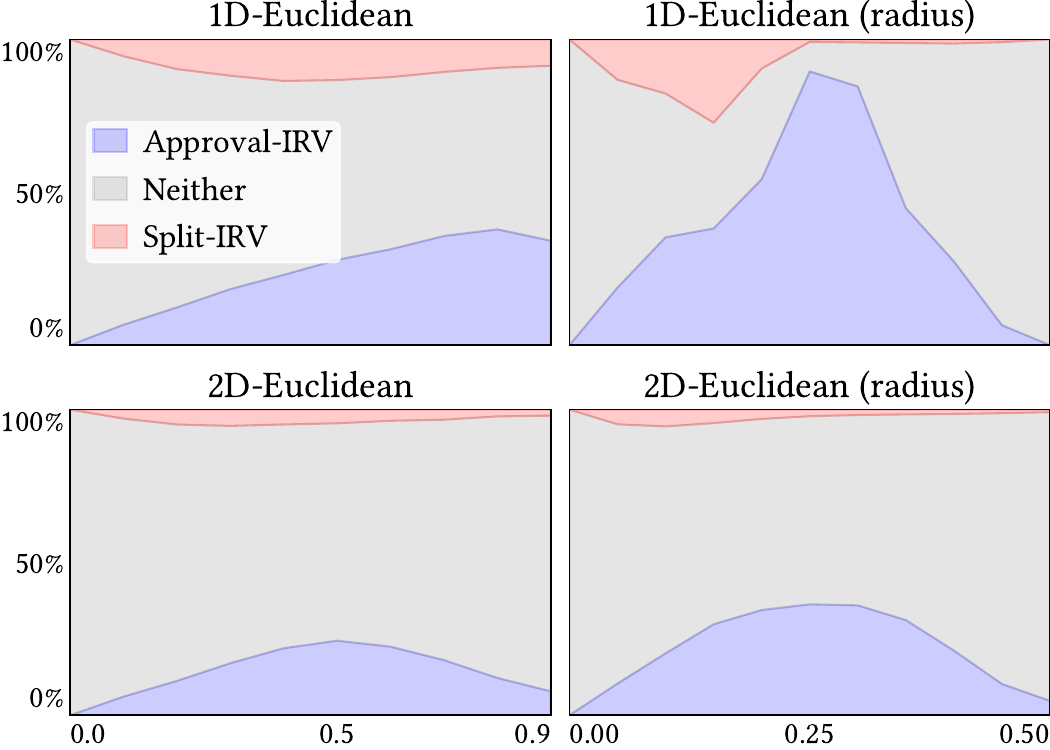}
		\caption{Rule with lowest distortion.}
		\label{fig:exp-best}
	\end{subfigure}	
	\caption{The average distortion of the winner (Figure a)  and the frequency of returning the candidate with the lowest distortion for each rule (Figure b).}
	\label{tabfig:expe-distortion}
\end{figure}

As explained in \Cref{sec:expe-eval}, we compare Approval-IRV and Split-IRV based on the winner they select. We present here the results on 5 additional datasets. Moreover, in addition to the average Borda score of the winning candidate chosen by each rule (\Cref{fig:exp-borda}) and the frequency of agreement between the rules and linear-order IRV (\Cref{fig:exp-similarity}), we provide in \Cref{fig:exp-condorcet} the frequency of finding the Condorcet winner for each rule. Of course, this frequency is upper bounded by the frequency with which a Condorcet winner exists, and so we also indicate that frequency for each dataset, noting that this value is not always $100\%$.

For Euclidean datasets, we show in \Cref{tabfig:expe-distortion} the average \emph{distortion} \citep{anshelevich2018approximating} of the winner selected by the rules. For a given Euclidean profile in which voters and candidates are associated to positions $p(v), p(c) \in \mathbb R^d$, the \emph{cost} of a candidate is defined as $\text{cost}(c) = \sum_v \|p(v)-p(c)\|$ (where $\|\cdot\|$ is the Euclidean norm), and the \emph{distortion} of the rule for a given profile $P$ in which the rule selects the candidate $w \in C$ is defined as 
\begin{align*}
	\text{dist}(f,P) = \frac{\text{cost}(w)}{\min_{c\in C} \text{cost}(c)},
\end{align*}
which is the factor by which the cost of $w$ is higher compared to the cost of the cost-optimal candidate.
The distortion of any candidate $w$ is lower bounded by $1$, and the lower it is, the closer it is to optimal. The \emph{average distortion} is simply the average of this value over all sampled profiles for a given dataset.

The conclusions we can draw from the additional datasets are similar to the observations that we already described in \Cref{sec:expe-eval}. We only add that Approval-IRV returns the Condorcet winner more frequently than Split-IRV, when such a candidate exists. We also observe a surprising ``spike'' in the curves obtained for the 1D-Euclidean model with indifferences when the radius is $r \approx 0.25$, and think it could be interesting to investigate this.

\subsection{Map of Elections}
\label{app:map}
Here, we provide a version of \Cref{fig:exp-map-1} from \Cref{sec:expe-eval} (showing Borda score differences on the map of elections) with a broader value range, to be able to identify instances where Approval-IRV performs much better than Split-IRV.

\begin{figure}[h]
	\includegraphics[width=0.55\linewidth]{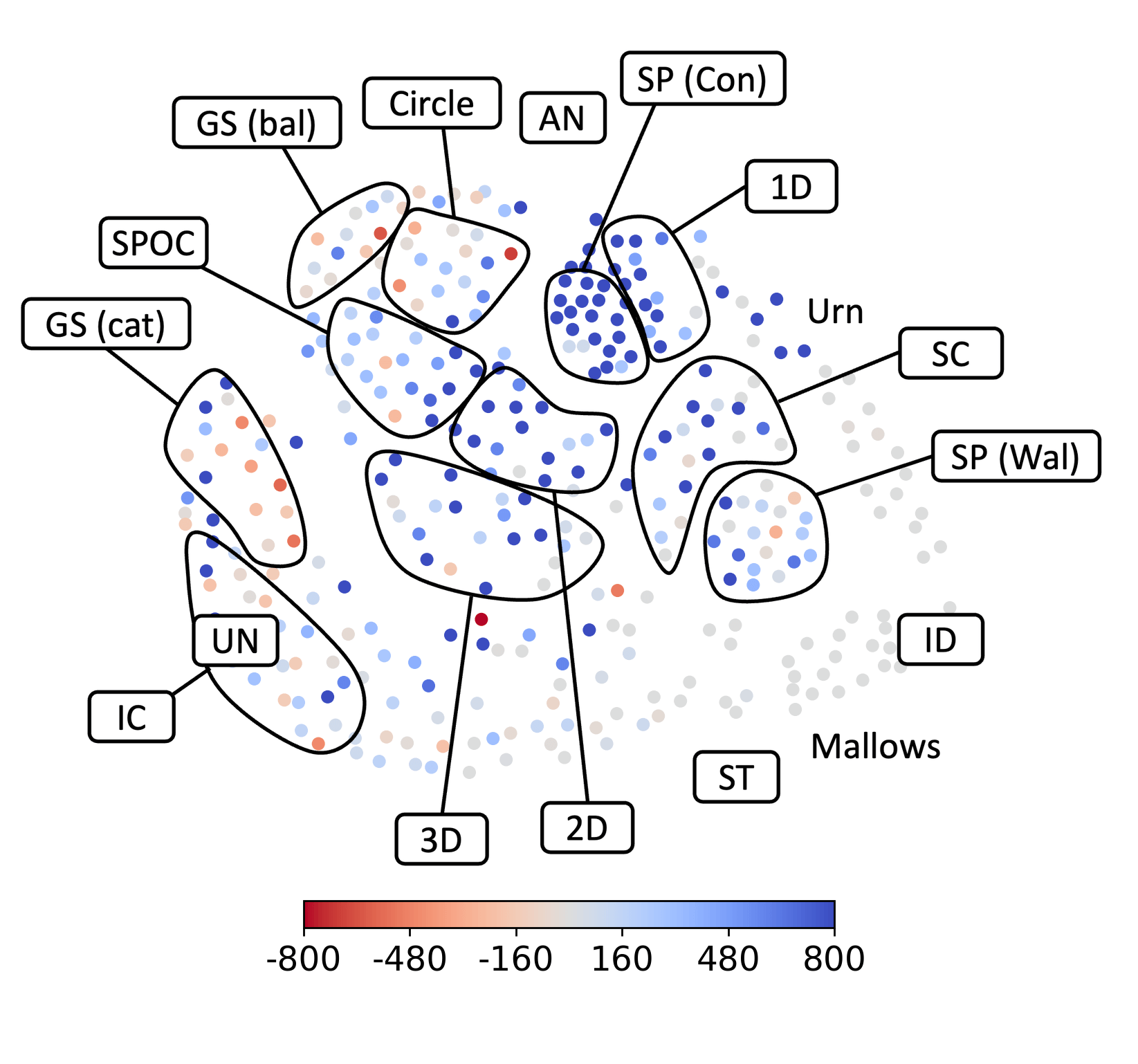}
	\vspace{-15pt}
	\caption{Map of elections showing the difference in Borda score (with respect to the underlying linear order profile) between the Approval-IRV and Split-IRV winner, capped to lie in $[-800, 800]$, summed over all probabilities of merging $p = 0.1, 0.2, \dots, 0.9$ in the coin-flip model and over 50 random sampled weak order profiles for each $p$.}
	\label{fig:exp-map-2}
\end{figure}
\end{document}